\documentclass[11pt,letterpaper,twoside]{report}

\makeatletter
\let\orig@chapter\@chapter
\def\@chapter[#1]#2{\ifnum \c@secnumdepth >\m@ne
	\if@mainmatter
	\refstepcounter{chapter}%
	\typeout{\@chapapp\space\thechapter.}%
	\addcontentsline{toc}{chapter}%
	{CHAPTER~\protect\numberline{\thechapter:}#1}%
	\else
	\addcontentsline{toc}{chapter}{#1}%
	\fi
	\else
	\addcontentsline{toc}{chapter}{#1}%
	\fi
	\chaptermark{#1}%
	\addtocontents{lof}{\protect\addvspace{10\p@}}%
	\addtocontents{lot}{\protect\addvspace{10\p@}}%
	\if@twocolumn
	\@topnewpage[\@makechapterhead{#2}]%
	\else
	\@makechapterhead{#2}%
	\@afterheading
	\fi}
\makeatother

\usepackage{etoolbox}
\usepackage{lipsum}
\usepackage{appendix}

\usepackage{geometry}
\usepackage{setspace}
\usepackage{titlesec}
\usepackage[titles,subfigure]{tocloft}
\usepackage{titletoc}

\usepackage{natbib}
\usepackage{apalike}

\usepackage{bibentry}
\nobibliography*

\usepackage{subfigure}
\usepackage{epsfig}
\usepackage{booktabs}
\usepackage{multicol}
\usepackage{listings}

\usepackage{amsthm}
\usepackage{amsmath}
\usepackage{amssymb}

\usepackage{times}
\usepackage{microtype}
\usepackage{textcomp}

\usepackage{xspace}

\usepackage[hidelinks]{hyperref} 

\usepackage{epstopdf}

\titleformat{\chapter}[hang] 
{\centering\normalfont\bfseries}{\MakeUppercase{\chaptertitlename}\ \thechapter:\ }{0em}{} 
\titleformat{\section}{\bfseries}{\thesection}{1em}{}
\titleformat{\subsection}{\bfseries}{\thesubsection}{1em}{}

\titlespacing{\chapter}{0pt}{0.62in}{-\parskip}
\titlespacing{\section}{0pt}{\parskip}{-\parskip}
\titlespacing{\subsection}{0pt}{\parskip}{-\parskip}
\titlespacing{\appendix}{0pt}{0in}{-\parskip}

\titlespacing{\paragraph}{0in}{0.08in}{0.07in}

\newcommand*{\noaddvspace}{\renewcommand*{\addvspace}[1]{}}

\newcommand{\sean}[1]{}

\renewcommand{\cite}{\citep}

\renewcommand{\contentsname}{\centering Contents}

\renewcommand{\mod}{\operatorname{mod}}

\newcommand{\kwfont}[1]{\textsf{#1}\xspace} 

\newcommand{\mkkw}[2]{
	\newcommand{#1}[0]{\kwfont{#2}}
}

\newcommand{\T}{\mathcal{T}}

\newcommand{\Or}{\mathcal{O}}
\newcommand{\A}{\mathcal{A}}
\newcommand{\B}{\mathcal{B}}
\newcommand{\V}{\mathcal{V}}

\newcommand{\F}{\mathcal{F}}

\renewcommand{\S}{\mathcal{S}}
\newcommand{\tr}[1]{ \begin{list}{}{\setlength{\leftmargin}{#1em}} \item}
	
	\newcommand{\tl}{ \end{list}}

\newcommand{\incompGraphLambda}{$G(A^{i,\lambda}_l,B^{i,\lambda}_l)$}

\newcommand{\mywhile}{\textrm{\bf{while}}}

\newcommand{\myif}{\textrm{\bf{if}}}

\newcommand{\myargmin}[2]{ 
	\begin{array}{c l}
		\textrm{argmin} & #2 \\
		#1
	\end{array}
}

\newcommand{\mylim}[2]{
	\begin{array}{l l} 
		\begin{array}{c}
			\textrm{lim}\\
			#1
		\end{array} 
		&  #2
	\end{array}
}

\DeclareMathOperator{\argmin}{argmin}

\providecommand{\abs}[1]{\lvert#1\rvert}
\providecommand{\norm}[1]{\lVert#1\rVert}

\makeatletter
\newcommand{\labitem}[2]{%
	\def\@itemlabel{\textbf{#1}}
	\item
	\def\@currentlabel{#1}\label{#2}}
\makeatother
\newcommand{{\convd}}{{\ \buildrel d \over \longrightarrow \ }}

\mkkw{\cfs}{CFS}

\mkkw{\edf}{EDF}
\mkkw{\edfwm}{EDF-WM}
\mkkw{\fp}{FP}
\mkkw{\fprm}{RM}
\mkkw{\fpdm}{DM}
\mkkw{\gedf}{G-EDF}
\mkkw{\gsnedf}{GSN-EDF}
\mkkw{\gfp}{G-FP}
\mkkw{\pedf}{P-EDF}
\mkkw{\pfp}{P-FP}
\mkkw{\cedf}{C-EDF}
\mkkw{\pssched}{PS}

\mkkw{\pfsched}{PF}
\mkkw{\pd}{PD}
\mkkw{\pds}{PD$^2$}
\mkkw{\cpds}{C-PD$^2$}

\mkkw{\jlfp}{JLFP}
\mkkw{\jldp}{JLDP}

\mkkw{\pfpgi}{P-FP-Rm}
\mkkw{\pfpdi}{P-FP-R1}
\mkkw{\pedfgi}{P-EDF-Rm}
\mkkw{\pedfdi}{P-EDF-R1}
\mkkw{\cedfiigi}{C2-EDF-Rm}
\mkkw{\cedfiidi}{C2-EDF-R1}
\mkkw{\cedfiiigi}{C6-EDF-Rm}
\mkkw{\cedfiiidi}{C6-EDF-R1}
\mkkw{\gedfgi}{G-EDF-Rm}
\mkkw{\gedfdi}{G-EDF-R1}

\mkkw{\cedfalldi}{C24-EDF-R1}

\mkkw{\capdsiigi}{C2-aPD$^2$-Rm}
\mkkw{\capdsiidi}{C2-aPD$^2$-R1}

\mkkw{\capdsiiigi}{C6-aPD$^2$-Rm}
\mkkw{\capdsiiidi}{C6-aPD$^2$-R1}

\mkkw{\gapdsgi}{G-aPD$^2$-Rm}
\mkkw{\gapdsdi}{G-aPD$^2$-R1}

\mkkw{\cspdsiigi}{C2-sPD$^2$-Rm}
\mkkw{\cspdsiidi}{C2-sPD$^2$-R1}

\mkkw{\cspdsiiigi}{C6-sPD$^2$-Rm}
\mkkw{\cspdsiiidi}{C6-sPD$^2$-R1}

\mkkw{\gspdsgi}{G-sPD$^2$-Rm}
\mkkw{\gspdsdi}{G-sPD$^2$-R1}

\mkkw{\schedfifo}{SCHED\_FIFO}
\mkkw{\schedrr}{SCHED\_RR}
\mkkw{\schedother}{SCHED\_OTHER}
\mkkw{\schedspor}{SCHED\_SPORADIC}
\mkkw{\prioprot}{PRIO\_PROTECT}
\mkkw{\scheddl}{SCHED\_DEADLINE}

\mkkw{\npcs}{NCP}
\mkkw{\srp}{SRP}
\mkkw{\pcp}{PCP}
\mkkw{\msrp}{MSRP}
\mkkw{\dpcp}{DPCP}
\mkkw{\mpcp}{MPCP}
\mkkw{\mpcpvs}{MPCP-VS}
\mkkw{\fmlp}{FMLP}
\mkkw{\fmlpp}{FMLP$^{\mathrm{+}}$}
\mkkw{\npfmlpp}{NP-FMLP$^{\mathrm{+}}$}
\mkkw{\omlp}{OMLP}
\mkkw{\pip}{PIP}

\mkkw{\pft}{PF-T}   
\mkkw{\pfc}{PF-C}   
\mkkw{\pfq}{PF-Q}  
\mkkw{\rwlin}{LX-RW}
\mkkw{\tft}{TF-T}
\mkkw{\tfq}{TF-Q}
\mkkw{\mtxt}{MX-T}
\mkkw{\mtxq}{MX-Q}

\long\def\symbolfootnote[#1]#2{\begingroup%
	\def\thefootnote{\fnsymbol{footnote}}\footnote[#1]{#2}\endgroup}

\theoremstyle{plain}

\theoremstyle{plain}
\newtheorem{thm}{Theorem}[section]
\newtheorem{cor}[thm]{Corollary}
\newtheorem{lem}[thm]{Lemma}
\newtheorem{prop}[thm]{Proposition}

\newtheorem{rem}{Remark}

\theoremstyle{definition}
\newtheorem{defn}[thm]{Definition}

\newcommand\rightparend[1]{{%
		\unskip\nobreak\hfil\penalty50
		\hskip2em\hbox{}\nobreak\hfil\textbf{#1}%
		\parfillskip=0pt \finalhyphendemerits=0 \par}}

\newtheorem{xxexample}{Example}[chapter]

\begin{document}


\newgeometry{left=1.25in,top=2in,right=1.25in,bottom=1in,nohead}
\pagenumbering{roman}

\begin{titlepage}
	\begin{singlespace}
		\begin{center}
			\vspace{2in}
			TREE ORIENTED DATA ANALYSIS
			
			\vspace{1in}
			
			Sean Skwerer
			
			\vspace{1in}
			
			\noindent A dissertation submitted to the faculty at the University of North Carolina at Chapel Hill in partial fulfillment of the requirements for the degree of Doctor of Philosophy in
			the Department of Statistics and Operations Research.
			
			\vspace{1in}
			
			Chapel Hill\\
			2014
		\end{center}
	\end{singlespace}
	
	\vspace{.5in}
	
	\begin{flushright}
		\begin{minipage}[t]{1.5in}
			Approved by:\\
			Scott Provan \\
			J.S. Marron \\
			Ezra Miller \\
			Gabor Pataki \\
			Shu Lu 
		\end{minipage}
	\end{flushright}
	
\end{titlepage}

\newgeometry{left=1.25in,top=8.33in,right=1.25in,bottom=1in,nohead}

\begin{center}
	\begin{singlespace}
		\copyright 2014\\
		Sean Skwerer \\
		ALL RIGHTS RESERVED
	\end{singlespace}
\end{center}

\clearpage

\newgeometry{left=1.25in,top=2in,right=1.25in,bottom=1in,nohead}

\begin{center}
	\vspace{2in}
	\textbf{ABSTRACT}
	
	\begin{singlespace}
		Sean Skwerer: Tree Oriented Data Analysis \\
		(Under the direction of J. S. Marron and Scott Provan)
	\end{singlespace}
\end{center}

Complex data objects arise in many areas of modern science including evolutionary biology, nueroscience, dynamics of gene expression and medical imaging. Object oriented data analysis (OODA) is the statistical analysis of datasets of complex objects.  
Data analysis of tree data objects is an exciting research area with interesting questions and challenging problems. 
This thesis focuses on 
tree oriented
statistical methodologies, and algorithms for solving related mathematical
optimization problems. 

This research is motivated by the goal of analyzing a data set of 
images of human brain arteries.
The approach we take here is to use a novel representation of
brain artery systems as points in phylogenetic treespace. 
The treespace property of unique global geodesics leads
to a notion of geometric center called a Fr\'{e}chet mean.
For a sample of data points, the Fr\'{e}chet function is the sum of squared distances
from a point to the data points, and the Fr\'{e}chet mean is
the minimizer of the Fr\'{e}chet function. 

In this thesis we use properties of the Fr\'{e}chet function to develop an algorithmic system for computing Fr\'{e}chet means. 
Properties of the Fr\'{e}chet function are also used to
show a sticky law of large numbers which
describes a surprising stability of the topological tree structure of sample Fr\'{e}chet means at that of the population Fr\'{e}chet mean.
We also introduce non-parametric regression of brain artery tree structure as a response variable to age based on weighted Fr\'{e}chet means.

\clearpage

\restoregeometry

\begin{center}
	\vspace*{52pt}
	For Robert and Laurie, my beloved and caring parents.
\end{center}

\pagebreak

\begin{center}
	\vspace*{52pt}
	\textbf{PREFACE}
\end{center}

\indent This thesis is based on collaborative work with my advisors J. S. Marron and Scott Provan in the domain of tree oriented data analysis.
Chapter 1, which provides an overview of tree oriented data analysis and background material, is the original work of the student.

Chapter 2, which describes the treespace Fr\'{e}chet mean optimization problem, and related theory and methods, is the original work of the student. Section 2.1 provides an overview and discussion of the problem, and these perspectives are an original contribution by the student.
Section 2.2 and section 2.3 focus on summarizing existing theory and methodology. 
Sections 2.4 and 2.5 contain novel unpublished research by the student.

Chapter 3 summarizes data analytic results.
This analysis was conducted by the student, with one exception, the result attributed to Megan Owen in Section 3.1. 
The discussion of Fr\'{e}chet mean degeneracy in this chapter  is an extension of research from a paper written by this student in collaboration with a number of authors (see reference) \cite{Skwerer2014}. 

Chapter 4 contains novel theoretical analysis by the student. Section 4.1 contains background material. Section 4.2 provides a summary of related research that the student contributed to as part of a collaborative research group. Section 4.3 contains an original presentation of basic definitions and novel unpublished results by the student. Section 4.4 contains the main result of the chapter which is unpublished research by the student.

Chapter 5 contains a novel method for regression of tree data objects against a Euclidean response variable. This method was described by  J. S. Marron and implemented by the student. The data analysis was conducted by the student. This is unpublished research.
\clearpage



\renewcommand{\contentsname}{\centerline{TABLE OF CONTENTS}}
\renewcommand{\cfttoctitlefont}{\hfill\bfseries}
\renewcommand{\cftdotsep}{1.5}

\cftsetpnumwidth{0.5em}
\cftsetrmarg{0.5em}

\setlength{\cftbeforetoctitleskip}{61pt}
\setlength{\cftaftertoctitleskip}{1em}

\renewcommand{\cftchapfont}{\normalfont}
\renewcommand{\cftchappagefont}{\normalfont}
\renewcommand{\cftchapleader}{\cftdotfill{\cftdotsep}}

\setlength{\cftbeforesecskip}{10pt}
\setlength{\cftbeforesubsecskip}{10pt}
\setlength{\cftbeforesubsubsecskip}{10pt}

\titlespacing{\chapter}{0pt}{1in}{-\parskip}

\begin{singlespace}
	\tableofcontents
\end{singlespace}

\clearpage


\renewcommand{\listtablename}{LIST OF TABLES}
\phantomsection
\addcontentsline{toc}{chapter}{LIST OF TABLES}

\setlength{\cftbeforelottitleskip}{-11pt}
\setlength{\cftafterlottitleskip}{1em}
\renewcommand{\cftlottitlefont}{\hfill\bfseries}
\renewcommand{\cftafterlottitle}{\hfill}

\setlength{\cftbeforetabskip}{10pt}

\addtocontents{lot}{\protect\noaddvspace}

\titlespacing{\chapter}{0pt}{0in}{-\parskip}

\cftsetpnumwidth{1em}
\cftsetrmarg{2em}

\begin{singlespace}
	\listoftables
\end{singlespace}

\clearpage


\renewcommand{\listfigurename}{LIST OF FIGURES}
\phantomsection
\addcontentsline{toc}{chapter}{LIST OF FIGURES}

\addtocontents{lof}{\protect\noaddvspace}

\setlength{\cftafterloftitleskip}{1em}
\renewcommand{\cftloftitlefont}{\hfill\bfseries}
\renewcommand{\cftafterloftitle}{\hfill}

\setlength{\cftbeforefigskip}{10pt}
\cftsetrmarg{1.0in}

\setlength{\cftbeforeloftitleskip}{-2in}

\begin{singlespace}
	\listoffigures
\end{singlespace}
\clearpage

\phantomsection
\addcontentsline{toc}{chapter}{LIST OF ABBREVIATIONS AND SYMBOLS}

\begin{center}
	\textbf{LIST OF ABBREVIATIONS AND SYMBOLS}
\end{center}

\newcommand{\Ab}[2]{\noindent  #1 \> #2 \\}
\newcommand{\Abi}[2]{\noindent #1 \hspace{1.5cm} \= #2 \\}

\begin{tabbing}
	\Abi{BHV}{Billera, Holmes, and Vogtman}%
	\Abi{LLN}{Law of Large Numbers}%
	\Abi{OODA}{Object Oriented Data Analysis}%
	\Abi{SLLN}{Strong Law of Large Numbers}%
\end{tabbing}

\clearpage


\pagenumbering{arabic}

% ...

\titlespacing{\chapter}{0pt}{.75in}{-\parskip}

\chapter{INTRODUCTION}
\label{ch:intro}
\section{Dissertation Overview}

Complex data objects arise in many areas of modern science including evolutionary biology, longitudinal studies, dynamics of gene expression and medical imaging. Object oriented data analysis (OODA) is the statistical analysis of datasets of complex objects. The \emph{atoms of a statistical analysis} are traditionally a number or a vector of numbers. In functional data analysis the atoms of interest are curves; for excellent treatment of functional data analysis see \cite{Ramsay2002,Ramsay2005}. OODA progresses from functions to more complex objects such as images, two-dimensional or three-dimensional shapes, and combinatorial structures such as graphs or trees. 

Data analysis of tree data objects, or \emph{tree oriented data analysis},  is an exciting research area with interesting questions and challenging problems. 
This thesis focuses on 
tree oriented
statistical methodologies, and algorithms for solving related mathematical
optimization problems. 
The mathematical focus
of this thesis is driven by the goal of analyzing a data set of 
images of human brain arteries collected by the CASILab at UNC-CH \cite{Handle}. From this perspective, this thesis
is aimed at making contributions to morphology, the study of the form and 
structure of organisms. 

In this thesis, trees are primarily modeled as points in \emph{Billera, Holmes, Vogtman treespaces}.
The original motivation for creating BHV treespaces was
to create a firm mathematical basis for statistical
inference of evolutionary histories by
developing a geometric space of phylogenetic trees. 
In phylogenetics, trees can be used as abstract representations 
of evolutionary histories. In such an abstraction the root
represents some common ancestor, the leaves represent
species, branches represent speciation events, and length
represents passage of time.
A \emph{phylogenetic tree}, (i) is a tree, (ii) has a positive length associated with
each of its edges, and (iii) has leaves which are in bijection with an
index set $\{0,1,\ldots,r\}$, which corresponds to a list of species. 
A BHV treespace is a geometric space of all phylogenetic trees with leaves
in bijection with the same index set i.e. $\{0,1,\ldots,r\}$.
For a more detailed description of phylogenetic trees
and BHV treespaces see Sec. \ref{sec:PhylogeneticTrees}.

A central question of tree oriented data analysis is ``what are appropriate notions of mean and variance for a set of trees?" 
Typically, the mean of a dataset is specified as the sum of the observations divided by the number of observations. The mean of real numbers could also be specified as the solution to an optimization problem. An arithmetic mean is the real number that 
minimizes the sum of squared distances to data points.
A more general notion of mean for metric spaces, called a \emph{Fr\'{e}chet mean} (a.k.a, \emph{barycenter} or \emph{center of mass}), is a point that minimizes the sum of squared distances to the data points. The Fr\'{e}chet mean is equivalent
to the arithmetic mean in the case when data points are vectors. 
BHV treespaces have nice properties for statistics,
including the existence of a unique shortest path between
every  pair of points, and the existence of a unique
Fr\'{e}chet mean for a set of points.
The focus of Ch. \ref{ch:FMmethods}
is mathematical theory and methods for solving
Fr\'{e}chet mean optimization problems defined for data sets
on BHV treespace. 

Prior to the research for this thesis, 
tube tracking algorithms 
were applied to a brain angiography dataset 
from the CASILab to 
reconstruct 3D models for the brain artery systems i.e.
tubular 3D trees \cite{Bullitt2002}. 
The main research advance for representing these trees
as points in a BHV treespace was finding a morphologically interpretable
fixed index set. 
The index set used in this research 
was determined by a technology in neuroimage analysis, called group-wise
landmark based shape correspondence, which optimally places landmarks
on the cortical surface of each member in a sample. This algorithm simultaneously
optimizes a term which spreads landmarks out in each subject and a term which forces landmarks
to similar positions for all subjects.
More details about this representation of artery trees as points in a BHV treespace are explained in Sec. \ref{sec:MapBrainArteryData}.

Results from analysis of this cortical landmark and brain artery data using Fr\'{e}chet means in BHV treespaces are presented in Ch. \ref{ch:AnalysisAngiographyData}. In summary, these results show there is little similarity in
the topological connections of brain arteries
from the base of the brain to points where arteries are nearest to cortical landmarks, with the level of resolution available. 

Fr\'{e}chet means in BHV treespaces exhibit an unusual stability
property which is known as \emph{stickiness}.
Contrasting the typical behavior of sample means, where typically small changes in the data result
in small changes in the sample mean, stickiness refers to the phenomenon of 
the sample mean sticking in place regardless of small changes in the dataset. 
Stickiness examples and a rigorous definition for stickiness are presented in Ch. \ref{ch:Stickiness}.
The observation of this property was attributed to Seth Sullivant, and studied further by the
the SAMSI working group for data sample on manifold stratified spaces, e.g. data sampled from BHV treespaces, during the 2012 Object Oriented Data Analysis program at SAMSI.
In this thesis, a new contribution to stickiness research is made, we characterize
the limiting behavior of Fr\'{e}chet sample means on BHV treespaces
as obeying a \emph{sticky law of large numbers}. This is the main result in Ch. \ref{ch:Stickiness}.

Kernel smoothing is a flexible method for studying the relationships between variables. It is used in estimating probability densities and in regression. 
In Ch. \ref{ch:TreespaceKernelSmoothing}, we present a novel
method for kernel smoothing regression of tree-valued response against a real-valued predictive variable. This method is applied to
the brain artery systems from the CASILab which will first be introducted in Sec. \ref{sec:MapBrainArteryData}.

The rest of this chapter is about phylogenetic trees and mapping brain artery trees to phylogenetic treespace. 

\section{Phylogenetic trees and BHV treespaces}\label{sec:PhylogeneticTrees}

\subsection{Graphs and trees}
A \emph{graph} is a set of points, called \emph{vertices} (use \emph{vertex} for a single point), and a set of lines connecting pairs of vertices, called \emph{edges}.  
A \emph{tree} is a connected graph which has no cycles of edges and vertices.
The \emph{degree} of a vertex is the number of edges connected to it. 
The vertices of a tree with degree one are called leaves, and the edges connected to them are called \emph{pendants}. 
Non-leaf vertices are called \emph{interior vertices}.
Edges which are not connected to leaves are called \emph{interior edges}.
An \emph{edge weighted tree} is a tree $T$ together with a positive length $|e|_T$ associated with every edge $e \in T$.
\emph{Contracting} an edge means its length shrinks to zero thereby identifying its two endpoints to form a single vertex.
A tree topology $T'$ that is created by contracting some edges in a tree $T$ is called a \emph{contraction} of $T$.
A \emph{star tree} is a tree with only pendant edges.

\subsection{Phylogenetic trees}

Evolutionary histories or hierarchical relationships are often represented graphically as phylogenetic trees.  In biology, the evolutionary history of species or operational taxonomic units (OTU's) is represented by a tree. Figure \ref{HaeckelTree} is a very early graphical depiction of a phylogenetic tree from \cite{Haeckel}.

\begin{figure}[H]
	\centering
	\includegraphics[width=1 \textwidth]{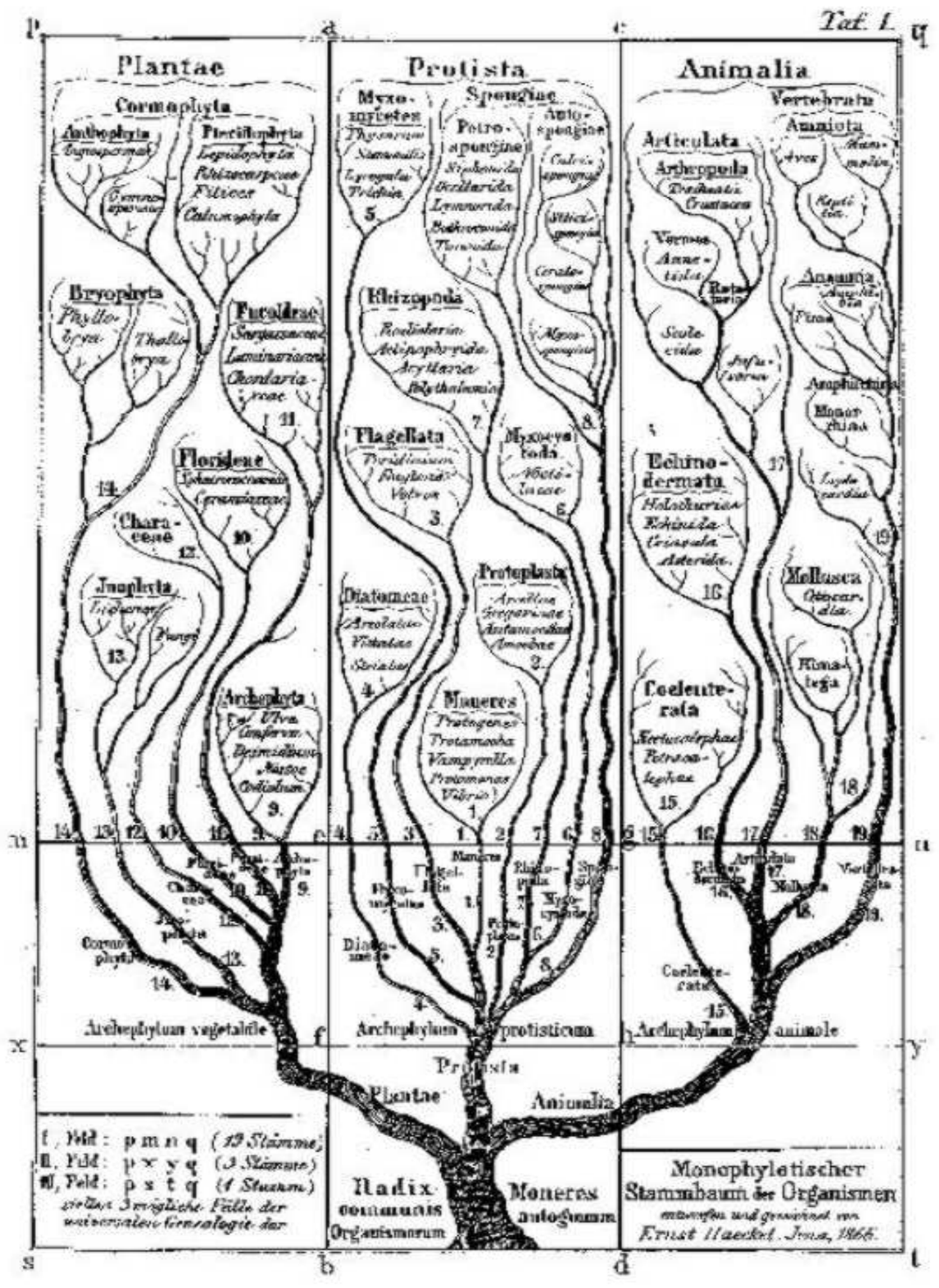}
	\caption{Phylogenetic tree depicting divergence of plants, protozoa, and animal kingdoms, Haeckel (1866).}
	\label{HaeckelTree}
\end{figure}

The root of the tree corresponds to a common ancestor. Branches indicate speciation of a nearest common ancestor into two or more distinct taxa. The leaves of the tree correspond to the present species whose history is depicted by the tree.

A \emph{labeled tree} is a tree $T$ with $r+1$ leaves distinctly labeled using the index set $I =\{0,1,\ldots,r\}$.
A \emph{phylogenetic tree} is a labeled edge-weighted tree.
The set of edges for a tree $T$ is written $E_T$.
Edge $e$ in $T$ is associated with a \emph{split}, $X_e\cup \bar{X}_e$ in $T$. This is a partition of $I$ into two disjoint sets of labels, $X_e$ and $\bar{X}_e$ on the two components of $T$ that result from deleting $e$ from $T$, with $X_e$ containing the index $0$.
The \emph{topology of a phylogenetic tree} is the underlying graph and pendant labels separated from the edge lengths.
The topology of a phylogenetic tree is uniquely represented by the set of splits associated with its edges.
Formally, two splits $X_e \cup \bar{X}_e$ and $X_f\cup \bar{X}_f$ 
are \emph{compatible} if and only if $X_e \subset X_f$ and $\bar{X}_f \subset \bar{X}_e$, or $X_f \subset X_e$ and $\bar{X}_e \subset \bar{X}_f$. 
Compatibility can be interpreted in terms of subtrees: the subtree with leaves in bijection with $\bar{X}_e$
contains the subtree with leaves in bijection with $\bar{X}_f$, or vice versa.
If every pair of splits in a set of splits is compatible then that set is said to be a compatible set.
Each distinct set of compatible splits is equivalent to a unique phylogenetic tree topology. A maximal tree topology is one in which no additional interior edges can be introduced i.e. $|E_T|=2r-1$, or equivalently every interior vertex has degree 3.

\subsection{Construction of BHV Treespaces}\label{sec:TreeSpace}
A BHV treespace, $\T_r$ is a geometric space in which each point represents a phylogenetic tree having leaves in bijection with a fixed label set $\{0,1,2,\ldots,r\}$.

A \emph{non-negative orthant} is a copy of the subset of $n$-dimensional Euclidean space defined by making each coordinate non-negative, $\mathbb{R}^{n}_{\geq 0}$. Here, only non-negative orthants are used, so we use orthant to mean non-negative orthant.
An \emph{open orthant}
is the set of positive points in an orthant.
Phylogenetic treespace is a union of many orthants, each corresponding to a distinct tree topology,
wherein the coordinates of a point are interpreted as the lengths of edges.
For a given set of compatible edges $E$, the associated orthant is denoted $\Or(E)$, 
and for a given tree $T$, the orthant in treespace containing
that point is denoted $\Or(T)$.
Trees in $\mathcal{T}_r$ have at most $r-2$ interior edges.
Each orthant of dimension $r-2$ corresponds to a combination of $r-2$ compatible edges. 
Orthants are glued together along common faces.
The shared faces of facets with $k$ positive coordinates are called the $k$-dimensional faces of 
treespace.

Take $\T_4$ as an example.
There are fifteen possible pairs of compatible edges.
Each compatible pair is associated with a copy of $\mathbb{R}^2_{\geq 0}$,
one axis of the orthant for each edge in the pair.  
The fifteen orthants are glued together along common axes.
Views of the two main features of $\T_4$ are displayed in Figure \ref{T_4_parts}.
See Figure \ref{T_4_link} for a visualization of the split-split compatibility graph of $\T_4$ \footnote{An interesting fact is that the compatibility graph of $\T_4$ is a Peterson graph. }.

\begin{figure}[H]
	\centering
	\subfigure[An open half book with three pages.  In this diagram the \emph{pages} of the book are three copies of $\mathbb{R}_{\geq 0}^2$ and the \emph{spine} is a copy of $\mathbb{R}_{\geq 0}$. The spine is labeled with the split $\{0,1\}|\{2,3,4\}$. Each page has the spine as one axis and the other axis is labeled with a split compatible with $\{0,1\}|\{2,3,4\}$. In $\T_4$ every one of the ten splits of $\{0,1,2,3,4\}$ is the label for the spine of the open half book.]{\includegraphics[width=0.8\textwidth]{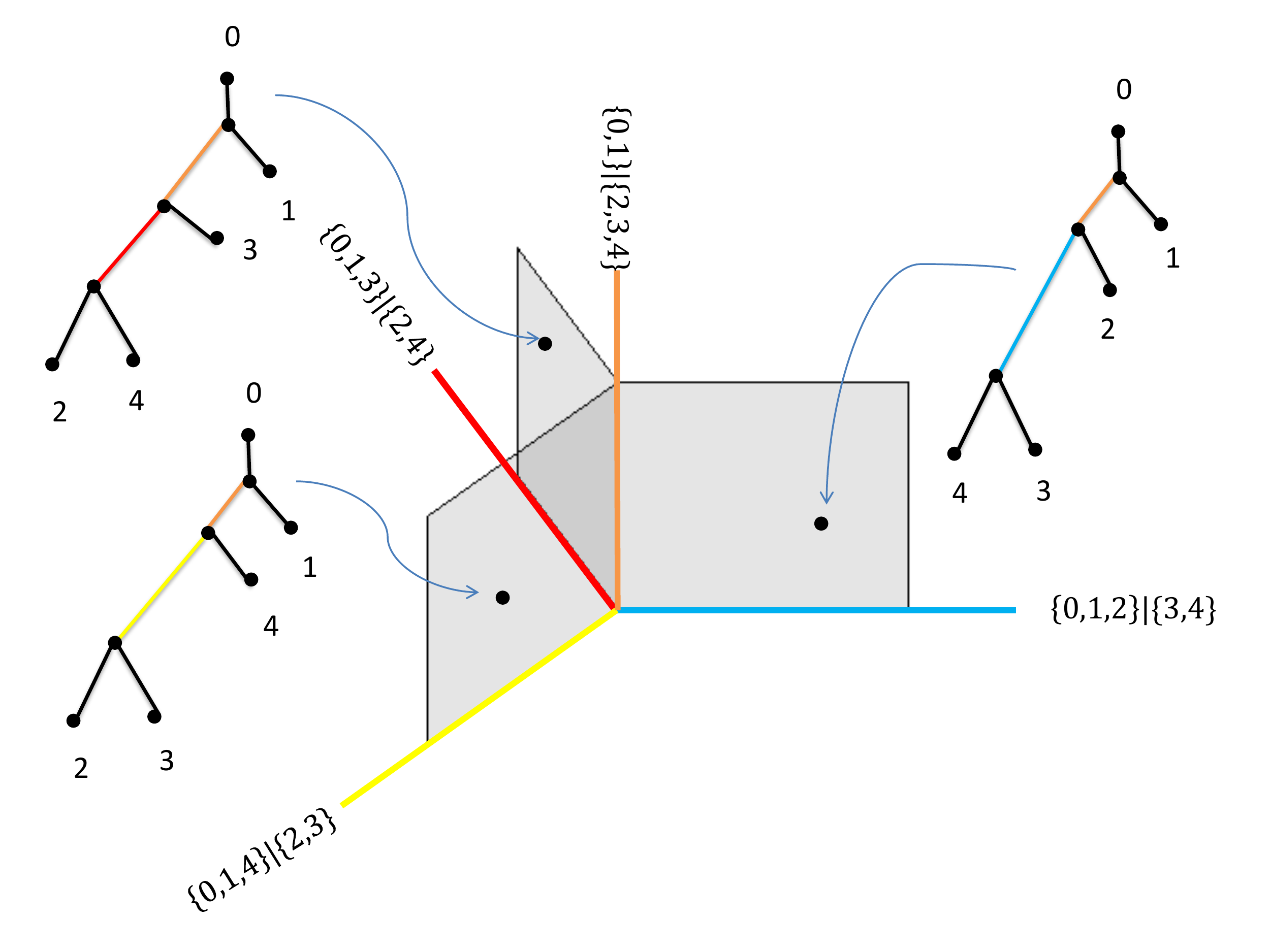}\label{open_book}}
	\subfigure[A five-cycle. A \emph{five-cycle} is five copies of $\mathbb{R}^2_{\geq 0}$ glued together along commonly labeled axes and at their origins.]{\includegraphics[width=0.8\textwidth]{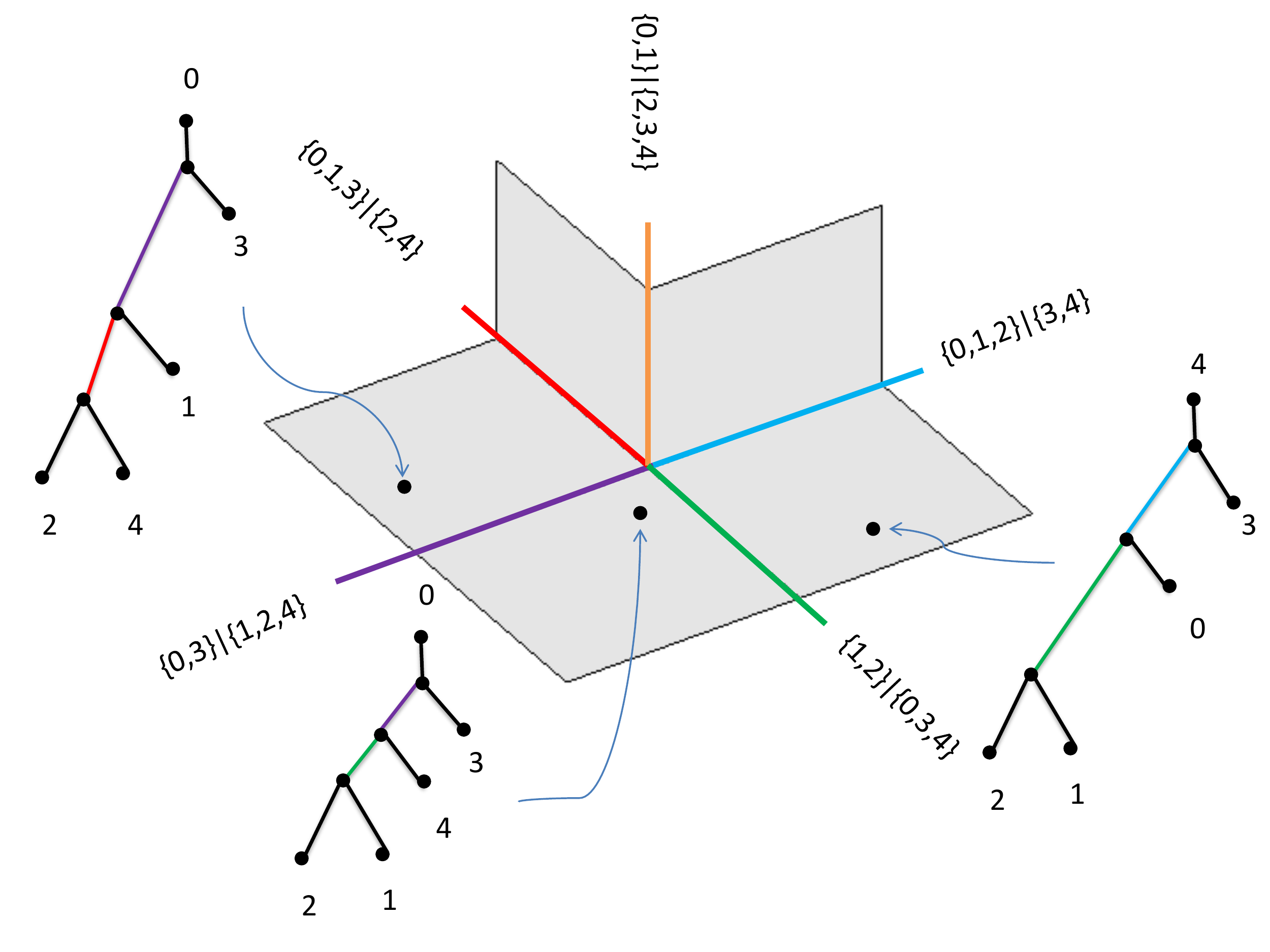}\label{five_cycle}}
	
	\caption{Spaces called a half-open book with three pages, depicted in (a) and a five cycle of quadrants, depicted in (b), appear in $\T_4$ repeatedly.
		For a global view of the structure of $\T_4$ see Fig. \ref{T_4_link}.}
	\label{T_4_parts}
\end{figure}

Each clique in the split-split compatibility graph represents a compatible combination of splits, or equivalently the topology of a phylogenetic tree. A graph is \emph{complete} if there is an edge between every pair of vertices. In a graph, a \emph{clique} is a complete subgraph. Each full phylogenetic tree is a maximal clique in the split-split compatibility graph because a clique represents a set of mutually compatible splits. The split-split compatibility graph of $\T_4$ has fifteen maximal cliques, each of which is represented an edge in the graph. The split-split compatibility graph of $\T_4$ determines how the orthants of $\T_4$ are glued.

\begin{figure}[H]
	\centering
	\includegraphics[width = 0.8\textwidth]{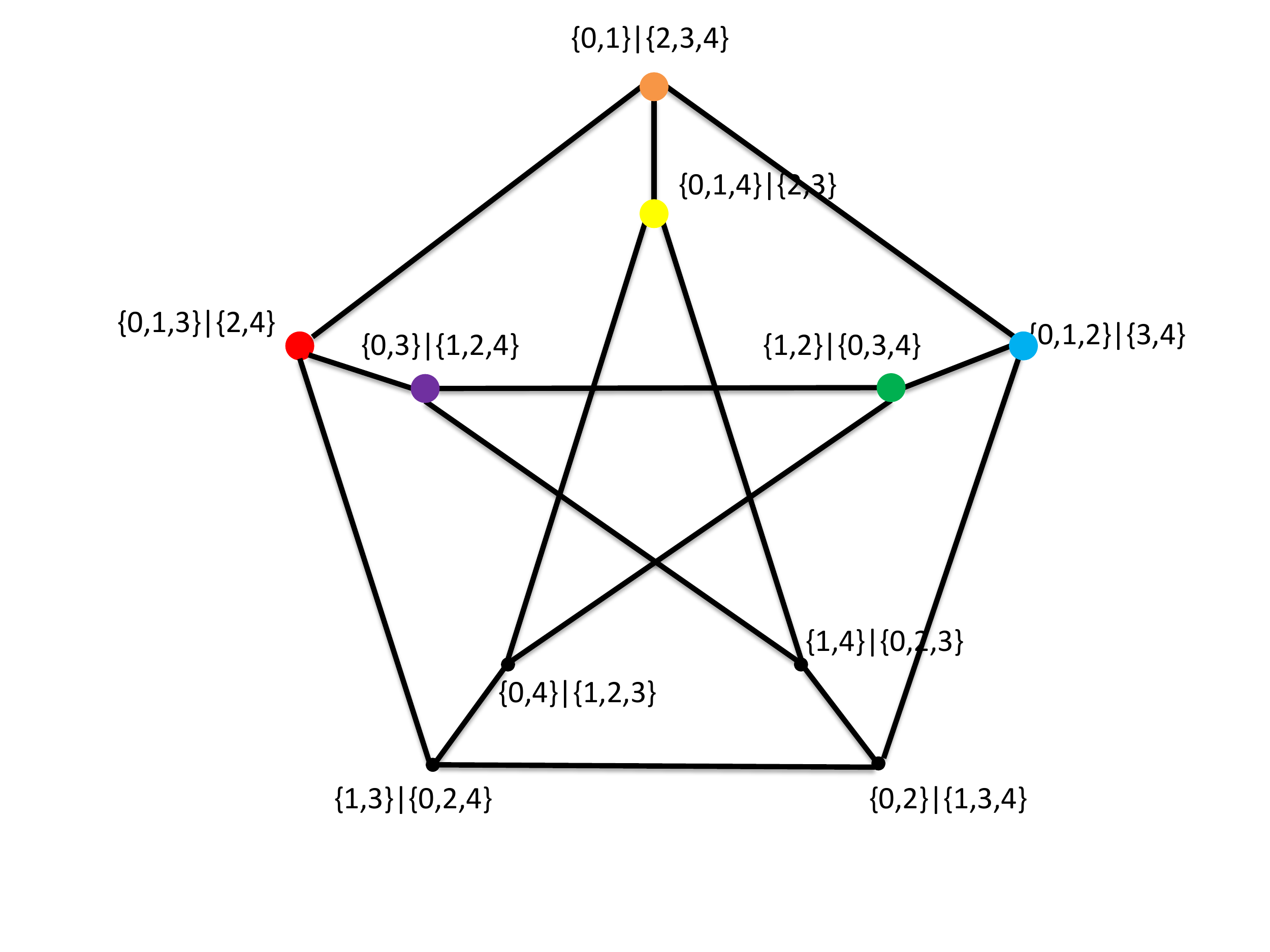}
	\caption{Split-split compatibility graph of ${\mathcal T}_4$. Each split has a node. Two splits are compatible if they are joined by an arc. This shows the overall connectivity of $T_4$, all possible splits for $\{0,1,2,3,4\}$, and all possible topologies for 4-trees. Each vertex and the three edges emanating from it in the graph corresponds to a copy of an open book like in Figure \ref{open_book}. Each five-cycle in the graph is a copy of a five-cycle depicted in Figure \ref{five_cycle}}
	\label{T_4_link}
\end{figure}

\section{Brain artery data}\label{sec:MapBrainArteryData}
The brain artery trees used in this study were reconstructed from a data set of Magenetic Resonance (MR) brain images collected by the CASILab at the University of North Carolina at Chapel Hill. This data set is publicly available and hosted at the MIDAG website \cite{Handle}. The database has images for various magnetic resonance modalities, including T1, T2, Magenetic Resonance Angiography (MRA), and Diffusion Tensor Imaging (DTI). The study enrolled 109 apparently healthy subjects. Each image is tagged with subject features of age, sex, handedness and self-identified race. The MRA was aquired at 0.5 x 0.5 x 0.8 mm$^3$ accuracy. 

Arteries branch out mostly as a tree from the heart and deliver blood to the entire body. In particular arteries transport oxygen and nutrient-rich blood to the brain. Magnetic Resonance Angiography (MRA) is a technique in medical imaging to visualize arteries. MRA uses the fact that blood flowing in the arteries has a distinct magnetic signature. Full 3D image acquisition is achieved by combining cross sectional 2D images. See Figures \ref{fig:MRA} and \ref{fig:MRArecon}  below for
an MRA slice and an artery reconstruction for the same
subject. A limiting factor is that MRA has a resolution threshold and consequently there are arteries that are too small for detection. MRA detects only arteries which feed blood rich with oxygen and nutrients to the body, and not veins which carry depleted blood back to the heart.

\begin{figure}[h]
	\centering
	\includegraphics[width=0.25\textwidth]{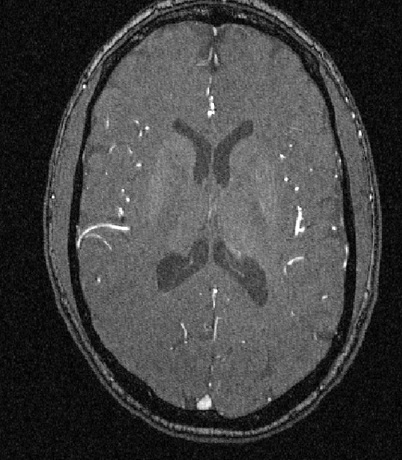} 
	\caption{One slice of a Magnetic Resonance Angiography (MRA) image. Bright regions indicate blood flow. Bright regions are tracked through MRA slices to recover arteries tubes as shown in Figure \ref{fig:MRArecon}.}
	\label{fig:MRA}
\end{figure}
\begin{figure}[h]
	\centering
	\includegraphics[width=0.45\textwidth]{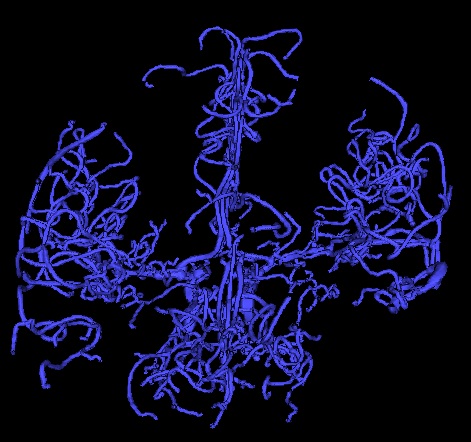} 
	\caption{The data object is a reconstructed brain artery tree consisting of four subtrees, left, right, anterior, and posterior. The goal of this research is statistical analysis of a sample of such data objects.}
	\label{fig:MRArecon}
\end{figure}

The arteries visible in MRA are generally naturally described as a tree. In most regions of the body and at the level of resolution possible, arteries branch like a tree without any loops.
A major part of this research has been opening up the possibility of using the space of phylogenetic trees as a mathematical basis for developing statistical methods for the study of artery trees. 
Phylogenetic trees have a common leaf set. However artery trees do not.
A common leaf set is artificially introduced by determining points on the cortical surface that
correspond across different people.  The next section describes the details of representing brain artery systems as points in BHV treespace.

\section{Mapping Brain Artery Systems to BHV treespace}

Figure \ref{fig:tree_landmarks} gives a detailed view of artery centerlines. 
Each tree consists of branch segments, and each branch segment consists of a sequence of spheres fit to the bright regions in the MRA image. 
The sphere centers are 3D points on the center line of the artery, and the radius approximates the arterial thickness.
A method for visualizing the structure of large trees was used to detect any remaining discrepancies \cite{AydinVis}.

\subsection{Correspondence}\label{subsec:correspondence}

In addition to the brain artery trees, the data set also includes reconstructions of the cortical surface.
A \emph{cortical correspondence} is made by determining points on the cortical surface that correspond across different people.
A group-wise shape correspondence algorithm based on spatial locations is used to place \emph{landmarks} on the cortical surface \cite{oguz2008cortical}. 
Each landmark is located in a corresponding spot on the cortical surface in every subject. 
In this study we use sixty-four landmarks for the right hemisphere and sixty-four landmarks for the left hemisphere (see Figure \ref{fig:tree_landmarks}).
The landmarks are combined with the artery tree by the following procedure. 
For each landmark, find the closest point on the tree of artery centerlines, called the landmark projection (see Figure \ref{fig:attach_lm}).
Each landmark and the line segment to its closest point become part of the tree. 
The tree is extracted by tracing the parts of the tree that are between the root and the landmarks.
Any parts of the artery tree that are not between a landmark projection and the base are trimmed (see Figure \ref{fig:subtree}). 
The base of the tree, called the \emph{root}, and the 128 landmarks, add up to a total of 129 common leaves. 
Once the tree has been extracted each edge is associated with a positive length.
The weight for each interior edge is the arc-length for the centerline of the artery tube.
The pendant for each landmark has length equal to the projection distance plus the
artery tube length from the projection point to the nearest artery branch. The pendant
length for the root of the tree is zero.

\begin{figure}[H]
	\centering
	\subfigure[]{\includegraphics[width = 0.4\textwidth]{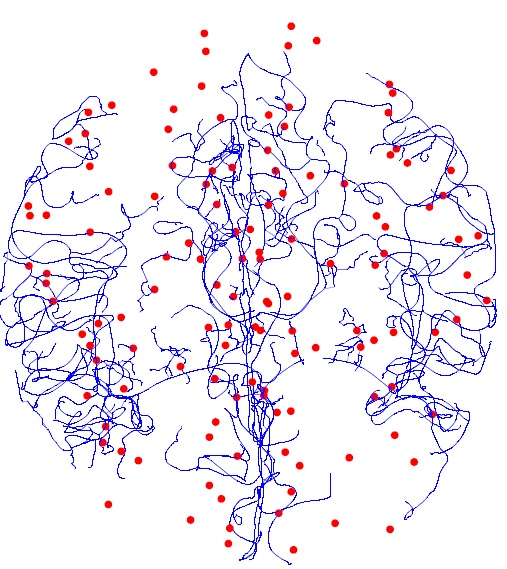}\label{fig:tree_landmarks}}
	\subfigure[]{\includegraphics[width = 0.4\textwidth]{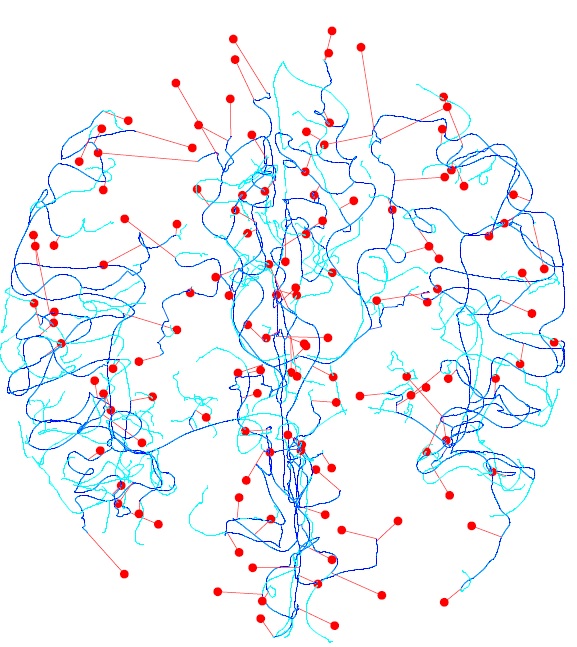}\label{fig:attach_lm}}
	\subfigure[]{\includegraphics[width=0.4\textwidth]{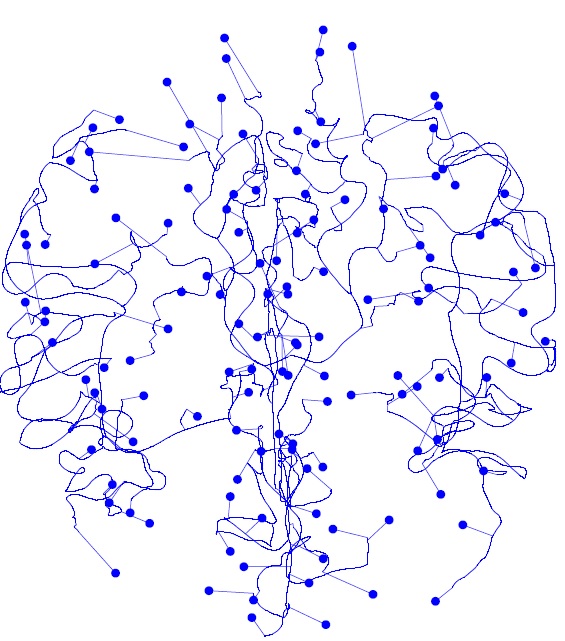}\label{fig:subtree}}
	\caption{Illustrates steps for \emph{cortical correspondence}. (a) Artery centerlines (blue) with the cortical landmarks (red). (b) Find the closest point on the artery tree to each landmark (red segments are landmark projections). Arteries which are not between the landmark projections and the base (cyan) are trimmed. (c) The result of this procedure is a cortical correspondence tree. }
	\label{fig:CortCorr}
\end{figure}

\section{Other approaches to tree oriented data analysis}
The branching structures of blood arteries and pulmonary airways are naturally modeled as trees. 
There is a large scope for progress in statistics for a population of trees. Currently, at the time that this dissertation is being written, research in tree oriented data analysis includes four major directions: combinatorial trees, Dyck paths, treeshapes and phylogenetic trees.

A seminal papaer in the \emph{combinatorial tree approach} to studying populations of anatomical trees laid foundations by proposing a metric, several notions of center, variation, and a method of principal components \cite{WangH}. Later fast PCA algorithms for 
combinatorial trees were developed \cite{AydinPCA}. 
The most recent innovation for combinatorial tree
is smoothing method for nonparametric regression with combinatorial tree structures as response variables against a univariate Euclidean predictor \cite{WangY} . 

Another approach uses
a Dyck path mapping of 
trees to curves \cite{Harris}. 
A Dyck path for a tree is produced 
by recording the distance on the 
tree to the root vertex during a depth
first left to right traversal. 
Representing trees as Dyck paths
opens up the possibility to use methods such as functional 
PCA \cite{ShenFunctionalTree}.

The \emph{treeshape approach} is an area of active research.
A tree-shape is a graph theoretic tree with a real matrix of fixed dimensions associated with each edge of the tree. Tree-shapes and statistics for tree-shapes were introduced in \cite{Feragen}. This approach allows very general descriptions of trees and thus allows for much richer representations of anatomical trees such as lungs or arteries than the above approaches. The generality of this approach comes with the cost of a very complicated sample space.
A related approach, unlabled-trees, is a special case of treeshape, where 
the edge attributes are restricted 
to be non-negative numbers.

\subsection{BHV treespace geodesics}\label{sec:Geodesics}
We now give an explicit description of geodesics in treespace.

Let $X \in \T_r$ be a variable point and let $T \in \T_r$ be a fixed point.
Let $\Gamma_{XT}= \{\gamma(\lambda)|0\leq \lambda \leq 1\}$ be the geodesic path from $X$ to $T$.
Let $C$ be the set of edges which are compatible in both trees, that is the union of the largest subset of $E_X$ which is compatible with every edge in $T$ and the largest subset of $E_T$ which is compatible with every edge in $X$. 

The following notation for the Euclidean norm of the lengths of a set of edges $A$ in a tree $T$ will be used frequently,
\begin{equation}
	||A||_T = \sqrt{ \sum_{e \in A}{|e|_T^2}}
\end{equation}
or without the subscript when it is clear to which tree the lengths are from.

A support sequence is a pair of disjoint partitions, $A_1\cup \ldots \cup A_k =E_X\setminus C$ and $B_1\cup\ldots\cup B_k=E_T \setminus C$.
\begin{thm}\label{thm:supports}\cite{OwenProvan}
	A support sequence $(\A,\B)=(A_1,B_1),\ldots,(A_k,B_k)$ corresponds to a geodesic if and only if it satisfies the following three properties:
	\begin{itemize}
		\item[\rm (P1)] For each $i>j$, $A_i$ and $B_j$ are compatible
		\item[\rm (P2)] $\frac{\norm{A_1}}{\norm{B_1}} \leq \frac{\norm{A_2}}{\norm{B_2}} \leq \ldots \leq \frac{\norm{A_{k}}}{\norm{B_{k}}}$
		\item[\rm (P3)] For each support pair $(A_i, B_i)$, there is no nontrivial partition $C_1 \cup C_2$ of $A_i$, and partition $D_1 \cup D_2$ of $B_i$, such that $C_2$ is compatible with $D_1$ and $ \frac{\norm{C_1}}{\norm{D_1}} < \frac{\norm{C_2}}{\norm{D_2}}$
	\end{itemize}
	The geodesic between $X$ and $T$ can be represented in ${\mathcal T}_n$ with legs
	\begin{displaymath}
		\Gamma_l=\left\{\begin{array}{ll}
			\left[\gamma(\lambda):\; \frac{\lambda}{1-\lambda}\leq\frac{\norm{A_1}}{\norm{B_1}}\right],
			&l=0\\[.7em]
			\left[\gamma(\lambda):\; \frac{\norm{A_i}}{\norm{B_i}}\leq\frac{\lambda}{1-\lambda}\leq\frac{\norm{A_{i+1}}}{\norm{B_{i+1}}}\right],
			&l=1,\ldots,k-1,\\[.7em]
			\left[\gamma(\lambda):\; \frac{\lambda}{1-\lambda}\geq\frac{\norm{A_k}}{\norm{B_k}}\right],
			&l=k\end{array}\right.
	\end{displaymath}
	The points on each leg $\Gamma_l$ are associated with tree $T_l$ having edge set
	
	\begin{displaymath}
		\begin{array}{rcl}
			E_l&=&B_1\cup\ldots\cup B_l\cup A_{l+1}\cup\ldots\cup A_k\cup C
		\end{array}
	\end{displaymath}
	
	Lengths of edges in $\gamma(\lambda)$ are
	
	\begin{displaymath}
		|e|_{\gamma(\lambda)}=\displaystyle\left\{\begin{array}{ll}
			\frac{(1-\lambda)\norm{A_j}-\lambda \norm{B_j}}{\norm{A_j}}|e|_X&e\in A_j\\[1em]
			\frac{\lambda \norm{B_j}-(1-\lambda)\norm{A_j}}{\norm{B_j}}|e|_{T}&e\in B_j\\[1.5em]
			(1-\lambda)|e|_X+\lambda |e|_{T}&e\in C\\
		\end{array}.\right.
	\end{displaymath}
	
	The length of $\Gamma$ is
	\begin{equation}\label{pathlength}
		d(X,T)=\bigg\Arrowvert(\norm{A_1}+\norm{B_1},\ldots,\norm{A_k}+\norm{B_k},|e_C|_{_X}-|e_C|_{_{T}})\bigg\Arrowvert
	\end{equation}
	and we call this the geodesic distance from $X$ to $T$.
\end{thm}

\chapter{Methods for Fr\'{e}chet Means in Phylogenetic Treespace}
\label{ch:FMmethods}

\section{Introduction}
The central research problem of this chapter is efficient computation of the Fr\'{e}chet mean of a discrete
sample of points in BHV treespaces.
The novel algorithmic system designed in this research project
improves upon the current solution methodologies in \cite{Miller} \cite{Bacak}.
These methods are applied to the sample of brain artery trees introduced in Chapter \ref{ch:intro}.

The contents of this Chapter are organized as follows:
In Sec. \ref{sec:FMproblem} we define the Fr\'{e}chet mean in BHV treespace, and give an overview discussion of the Fr\'{e}chet optimization 
problem.
In Sec. \ref{sec:global search methods} we present global methods
for optimizing the Fr\'{e}chet function i.e.\ methods which 
can move from one orthant of treespace to another. 
In Sec. \ref{sec:VistalCells} we describe how the combinatorics of treespace 
geodesics lead to polyhedral subdivision of treespace
into regions where the Fr\'{e}chet function has a fixed algebraic form.
The focus of Sec. \ref{sec:DiffAnalysis} is differential
properties of the Fr\'{e}chet function.
In Sec. \ref{sec:DampedNewton}, a method finding the minimizer
of the Fr\'{e}chet function in a fixed orthant of treespace is presented.
Application of these method to brain artery systems is the 
focus of Ch. \ref{ch:AnalysisAngiographyData}.

\subsection{Fr\'{e}chet means in BHV treespace}\label{sec:FMproblem}
For a given data set of $n$ phylogenetic trees $T^1, T^2,\ldots,T^n$ in $\T_r$, 
the \emph{Fr\'{e}chet function} is the sum of squares of geodesic distances from the data trees to a variable tree $X$. 
A geodesic $\gamma:[0,1]\to \T_r$ is the shortest path between its endpoints.
The geodesic from $X$ to $T^i$ is characterized by a geodesic support, $(\A^i,\B^i) = \left ( (A^i_1,B^i_1),\ldots, (A^i_{k^i},B^i_{k^i}) \right )$ (Thm. \ref{thm:supports}).
Given the geodesic supports  $(\A^1,\B^1),\ldots, (\A^n,\B^n)$ the Fr\'{e}chet function is
\begin{equation}\label{FrechetFunction}
	F(X) = \sum_{i = 1}^n d(X,T^i)^2=\sum_{i = 1}^n \left ( \sum_{l = 1}^{k^i} (\norm{x_{A_l^i}}+\norm{B_l^i})^2 + \sum_{e \in C^i} (|e|_T - |e|_{T^i})^2 \right).
\end{equation}
The objective is to solve the Fr\'{e}chet optimization problem
\begin{equation}\label{FrechetOptimization}
	\min_{X \in \T_r}{F(X)}
\end{equation}
Elementary Fr\'{e}chet function properties:
\begin{itemize}
	\item The Fr\'{e}chet function is continuous because the geodesic distances $d(X,T^i)$ are continuous \cite{OwenProvan,Miller}.
	\item The Fr\'{e}chet function $F(X)$ is strictly convex \cite{Sturm}, that is $F\circ \gamma:[0,1] \to \mathbb{R}$ is (strictly) convex for every geodesic $\gamma(\lambda)$ in $\T_r$. 
\end{itemize}
As a consequence of these properties we have the following result.
\begin{lem}
	The Fr\'{e}chet mean exists and is unique.
\end{lem}
\begin{proof}
	A strictly convex function either has a unique minimizer or can be made
	arbitrarily low. 
	Assuming that the data points are finite, 
	then a minimizer of the Fr\'{e}chet function must also be finite.
	Therefore the Fr\'{e}chet function has a unique minimizer.  
\end{proof}
The Fr\'{e}chet mean $\bar{T}$ is the unique minimizer of the Fr\'{e}chet function. 

\subsubsection{Problem Discussion}
The Fr\'{e}chet optimization problem, in BHV treespace, requires both selecting the minimizing tree topology and specifying its edge lengths.
Tree topologies are discrete and so the problem of selecting the minimizing tree topology is a combinatorial optimization problem; 
however it is possible to make search strategies which take advantage of the continuity of BHV treespace to find
the correct tree topology. 
It is natural to consider this problem in two modes of search: global, i.e.\ strategies which change the topology and edge lengths; and local, i.e.\ strategies which only adjust edge lengths. 
One motivation to consider global search and local search separately is that the local optimization problem
is convex optimization constrained to a Euclidean orthant.

The global search problem is challenging because 
the geometry of treespace creates difficulty 
in two essential parts of optimization (1) making progress towards optimality and (2)
verifying optimality. 
In treespace a metrically small neighborhood can actually be quite large in a certain sense. 
In constructing the space, the topological identification of the shared faces of orthants may create points in the closure of many orthants. 
In terms of trees, the neighborhood around a tree $X$, $N(X)$, is comprised
not only of trees with the same topology as $X$ but also
trees which have $X$ as a contraction. However, the list of tree topologies which have a particular tree $X$ as a contraction can be quite large. For example, if $X$ is a star tree then $X$ is a contraction of any tree i.e.\ $X$ is a contraction of $(2r-3)!!$ maximal phylogenetic tree topologies. 

Local optimality conditions for non-differentiable functions
are based on the rate of change of the objective function along directions issuing from
a point. 
Since the neighborhood of a point $X$ contains
all trees which have $X$ as a contraction, verifying 
that $X$ is optimal requires demonstrating that
any tree which contains $X$ as a contraction
has a larger Fr\'{e}chet function value.
For example, when $X$ is a star tree, $N(X)$ contains every tree
with the same pendants as $X$, and having infinitesimal interior edge lengths. In this sense finding a descent direction
can be essentially as hard as finding the topology of the Fr\'{e}chet mean itself.

Proximal point algorithms, a broad class of algorithms, are globally convergent
not only for the Fr\'{e}chet optimization problem, but are globally 
convergent for any well defined lower-semicontinuous convex optimization problem 
on a non-positively curved metric space \cite{Bacak}. 
This class of algorithms has nice theoretical properties, and certain implementations of proximal point algorithms are practical for the Fr\'{e}chet optimization problem
on non-positively curved orthant spaces. 

Proximal point algorithms are applicable to optimization problems on metric spaces. The general problem is minimizing a function $f$ on a metric space $M$ with distance function $d:M\times M \to \mathbb{R}$.
A proximal point algorithm solves a sequence of penalized optimization problems
of the form
\begin{align}
	\begin{displaystyle}
		P_k(f): \min_{x^{k}}{ \quad f(x^k)+\alpha_k d^2(x^{k-1},x^k)}
	\end{displaystyle}
\end{align}
where $\alpha_k$ influences the proximity of a solution to the point $x^{k-1}$.
\sean{Proximal point algorithms have a long history and 
	variations of proximal point algorithms for Euclidean 
	optimization problems appear in many areas 
	such as signal processing, statistics, and machine learning.}
Some good references for proximal point algorithms are \cite{Bacak2012a,Bertsekas2011,Li2009,Rockafellar1976}.

Implementing a generic proximal point algorithm to minimize the Fr\'{e}chet function on treespace does not seem advantageous. In particular, given a non-optimal point $X^0$ finding a point $X$ such that $F(X) < F(X^0)$ is not made any easier
by penalizing  the objective function $F(X)+\alpha d^2(X^0,X)$.
Penalizing the objective function with $\alpha d(X^0)$ does not 
provide any additional structure for checking the neighborhood, $N(X^0)$, of $X^0$ for a descent direction.

Split proximal point algorithms avoid directly tackling the complicated problem
of minimizing $F(X)$ by solving many much easier subproblems. 
For objective functions which can be expressed as a sum of functions, $f = f^1+\ldots +f^m$, a split proximal point algorithm alternates among penalized optimization problems for each function. 
Let $\{1,2,\ldots,m\}$ index the functions $f^1,\ldots,f^m$. 
A generic split proximal point algorithm is:
choose some sequence $i_1,i_2,\ldots$ where each 
term in the sequence is an element of $\{1,2,\ldots,m\}$ and sequentially solve
the split proximal point optimization problem:
\begin{align}
	\begin{displaystyle}
		P_k(f^{i_k}): \min_{x^{k}}{ \quad f(x^k)+\alpha_k d^2(x^{k-1},x^k)}
	\end{displaystyle}
\end{align}
Different versions of split proximal point algorithms are based on
the choice of the sequence $i_1,i_2,\ldots$ and the choice of the sequence $\{\alpha_k\}$. Naturally, a split proximal point procedure can be applied to the Fr\'{e}chet optimization problem by separating the Fr\'{e}chet function into a sum of squared distance functions, $F(X) = d^2(X,T^1)+\ldots +d^2(X,T^n)$.
For the Fr\'{e}chet function the split proximal point optimization problem is
\begin{align}
	\begin{displaystyle}
		P_k(d^2(X^k,T^{i_k})): \min_{x^{k}}{ \quad d^2(X^k,T^{i_k})+\alpha_k d^2(X^{k-1},X^k)}
	\end{displaystyle}
\end{align}
For the Fr\'{e}chet mean optimization problem on a geodesic non-positively curved space, the solution to a split proximal point optimization problem 
can be obtained easily in terms of geodesics.
The solution to  $P_k(d^2(X^k,T^{i_k}))$
must be on the geodesic between $X^{k-1}$ and $T^{i_k}$.
The term $d^2(X^k,T^{i_k})$ is the squared distance from the 
variable point to $T^{i_k}$ and the term $d^2(x^{k-1},x^k)$ is the squared
distance from the search point to $X^{k-1}$. Given
any point, there is at least one point
on the geodesic between $X^{k-1}$ and $T^{i_k}$ for
which the value of both terms is at least as small.
Since $X^k$ must be on this geodesic, the distance from $X^k$ to $X^{k-1}$
and the distance form $X^k$ to $T^{i_k}$ can be parameterized
in terms of  the proportion $t:0 \leq t \leq 1$ along the geodesic from $X^{k-1}$ to $T^{i_k}$:  $d(X^k,T^{i_k})=(1-t)d(X^{k-1},T^{i_k})$ and $d(X^k,X^{k-1})=td(X^{k-1},T^{i_k})$.
Parameterizing $d^2(X^k,T^{i_k})+\alpha_k d^2(X^{k-1},X^k)$
in terms of $t$ makes $P_k(d^2(X^k,T^{i_k}))$ into a problem
of minimizing a quadratic function in $t$.
The optimal step length is $t = \min\{1,\frac{\alpha}{1+\alpha}\}$.
Even more importantly, several versions of split proximal point algorithms have been shown to converge globally to the Fr\'{e}chet mean \cite{Bacak}.  
Split proximal point algorithms for Fr\'{e}chet means
on non-positively curved orthant spaces are discussed further in Section \ref{sec:global search methods}.

The overall strategy for minimizing the Fr\'{e}chet function will
be to use a split proximal point algorithm for global
search and switch to a local search procedure.
The motivation for switching to a local search procedure 
is that the if local search is initialized close to the optimal solution then faster convergence can be achieved.
The local optimization problem is minimizing the Fr\'{e}chet function in a 
fixed orthant $\Or$. 
One feature of the local optimization problem is that the Fr\'{e}chet function is a piecewise $C^\infty$ function
whose algebraic form depends on the geodesic distance from $X$ to the data trees.
But, the Fr\'{e}chet function is only $C^1$ when $X$ is restricted to the interior of a maximal dimension orthant. Analysis of differential properties of the Fr\'{e}chet function is presented in Sec. \ref{sec:DiffAnalysis}.

In $\T_r$, the Owen-Provan algorithm for the geodesic $\Gamma(X,T^i)$ has complexity $O(r^4)$, and with $n$ data points the total complexity of finding the algebraic form of $F(X)$ would be $O(n r^4)$. 
New algorithms for dynamically updating the algebraic form of $F(X)$ as $X$ varies are presented
in Section \ref{sec:updating}. 
Such algorithms will be especially useful in updating the objective function 
after small changes are made in the edge lengths of $X$; 
in particular these methods help accelerate
local search algorithms.

Here is a high-level outline of the algorithmic system for solving the Fr\'{e}chet mean optimization problem developed in this chapter:
\begin{flushleft}
	{\bf Treespace Fr\'{e}chet mean algorithm}\\
	{\bfseries input:} $T^1,T^2,\ldots,T^n,X^0 \in \T_r$, $\epsilon>0$, $\delta>0$, $a>0$, $0<c<1$, $\{\alpha_k\}$, $K\in \mathbb{N}$\\
	$F_0=\inf$; $F_1=F(X^0)$\\
	{\bfseries while} $F_1-F_0>a$\\
	\hspace*{1cm} SPPA for $K$ steps (Sec. \ref{sec:global search methods})\\
	\hspace*{1cm} c=1;\\
	\hspace*{1cm} {\bfseries while} $c \leq K$\\
	\hspace*{1cm} \hspace*{1cm} $X^k=argmin_{X} \left \{d^2(X,T^{i_k})+\alpha_k d^2(X^{k-1},X)\right\}$\\
	\hspace*{1cm}{\bfseries endwhile}\\
	\hspace*{1cm}{\bfseries while} approximate optimality conditions (\ref{eq:ApproxOptimal}) are not satisfied  \\
	\hspace*{1cm}\hspace*{1cm} compute a descent direction $P$ (Sec. \ref{sec:NewtonSteps})\\
	\hspace*{1cm}\hspace*{1cm} find a step-length, $\alpha$, satisfying decrease condition (Sec. \ref{sec:StepLength})\\
	\hspace*{1cm}\hspace*{1cm} {\bfseries let} $X^{k+1}=X^k+\alpha P$\\
	\hspace*{1cm}\hspace*{1cm} {\bfseries if} $|e| < \epsilon$ {\bfseries then} remove $e$\\
	\hspace*{1cm}{\bfseries endwhile}\\
	{\bfseries endwhile}
\end{flushleft}

The following sections discuss implementation details 
and present theoretical analysis pertaining to certain aspects
of the problem. The next section presents specific global
search procedures, both of which are versions of split proximal point algorithms. The remaining sections focus on aspects of the local search problem.

\section{Global search methods}\label{sec:global search methods}
Proximal point algorithms as applied to the Fr\'{e}chet optimization
problem have been studied in \cite{Bacak} and \cite{Sturm}.
The former views the Fr\'{e}chet mean problem from the paradigm 
of convex optimization while the later studies Fr\'{e}chet means
in the context of stochastic analysis in metric spaces of non-positive curvature.
In this Chapter two global search algorithms are discussed, the \emph{Inductive Mean Algorithm}, which is stochastic, and 
the \emph{Cyclic Split Proximal Point Algorithm}, which is deterministic.
Both of these are versions of split proximal point algorithms.

\subsection{Inductive means}
The \emph{Inductive Mean Algorithm} is a method for calculating the Fr\'{e}chet Mean based on \cite[Thm. 4.7]{Sturm}. This algorithm was developed independently in \cite{Bacak} and \cite{Miller}
and has been called Sturm's algorithm, named after
the K. T. Sturm who is attributed with its discovery
in proving that inductive means converge to the Fr\'{e}chet mean
in the more general context of probability distributions on non-positively curved metric spaces.

Consider a sequence $(Y_i)_{i\in \mathbb N}$ of independent identically distributed random observations from the uniform distribution on $\{T^1,T^2,\ldots,T^n\}$. 
Define a new sequence of points $(S_k)_{k \in \mathbb{N}}$ in $\T_r$ by induction on $k$ as follows:
\begin{displaymath}
	S_1 := Y_1,
\end{displaymath}
and
\begin{displaymath}
	S_k := \left(1-\frac{1}{k}\right)S_{k-1} + \frac{1}{k}Y_k,
\end{displaymath}
where the right hand side denotes the point $\frac{1}{k}$ fraction of the distance along the geodesic from $S_{k-1}$ to $Y_k$.
The point $S_k$ is called the \emph{inductive mean value} of $Y_1,\ldots,Y_k$. The expected squared distance between $S_k$ and $\bar{T}$ is less than or equal to $F(\bar{T})/k$.

\subsection{Cyclic Split Proximal Point Algorithm}
Choose a permutation, $(p_0,p_1,\ldots,p_{n-1})$, of $\{1,2,\ldots,n\}$.
Define a new sequence of points $(R_k)_{k \in \mathbb{N}}$ in $\T_r$ by induction on $k$ as follows:
\begin{displaymath}
	R_1 := Y_{p_1},
\end{displaymath}
and
\begin{displaymath}
	R_k := \left(1-\frac{1}{k}\right)R_{k-1} + \frac{1}{k}Y_{p_{\left( k\mod n \right)}}.
\end{displaymath}
The sequence $\{ R_k\}$ converges to $\bar{T}$ as $k$ approaches infinity \cite{Bacak}.

\section{Vistal cells and squared treespace}\label{sec:VistalCells}

The  value of 
the Fr\'{e}chet function at $X$
depends on the geodesics from $X$ to each of the data trees.
The goal in this section is to describe how treespace can be subdivided
into regions where the combinatorial form of geodesics from $X$ to the data
trees are all fixed. 
Descriptions of such regions will be used in analyzing the differential properties
of the Fr\'{e}chet function.

Analysis of the Fr\'{e}chet function starts at the level of a geodesic from a
\emph{variable tree} $X$ to a fixed \emph{source tree} $T$.
Given a fixed tree $T$, a vistal cell is a region $\V$ of treespace where the form 
of the geodesic from any tree $X$ in $\V$ to $T$ is constant.
The description of geodesics in Section \ref{sec:Geodesics} is now built upon further to study how the
combinatorics of geodesic supports for $\Gamma(X,T)$ can vary as $X$ varies.

\sean{
	BHV treespace can be subdivided into regions in which the Fr\'{e}chet function
	has constant form called \emph{multi-vistal cells} \cite[Thm. 3.27]{Miller}.
	
	Let $\Gamma$ be the geodesic between two phylogenetic trees in $\T_r$, $T$ and $T'$.
	The support for the geodesic between $T$ and
	$T'$ depends on the locations of $T$ and $T'$.
	If the tree $T$ is allowed to vary then the geodesic $\Gamma$ will
	vary as well.
	Small changes in the location of $T$ will result in small changes to the path $\Gamma$.
	The length of the path $\Gamma$ is a continuous function of $T$.
	If $T$ changes the support of $\Gamma$ may also change, but
	there will also be trees nearby $T$ with geodesics to $T'$ having the same support as $\Gamma$.
	For a fixed tree $T'$ the space of phylogenetic trees ${\mathcal T}_n$ can be subdivided into regions, called \emph{vistal cells}, for whose points all geodesic to $T'$ have the same support.
	
	Formally, the combinatorial structure of the geodesic $\Gamma$ from $T$ to $T'$ is the same for any
	$T$ in the vistal cell ${\cal V}(T';\Or;\A,\B;\mathcal{S})$, where $\Or = \Or(T)$ is the minimal
	orthant containing $T$.
	The space of squared coordinates $\T_r^2$ defined by the bijection from any point $x \in \T_r$
	$x_e \rightarrow y_e = x_e^2$.
	The image of a vistal cell in $\T^2_n$ is called the \emph{squared vistal cell} (s-vistal for short)
	${\cal V}^2(T';\Or;\A,\B;\mathcal{S})$.
	It is more convenient to describe the s-vistal cells in $\T_r^2$ because these sets have
	the benefit of being polyhedral cones. }

\begin{defn}
	\cite[Def. 3.3]{Miller} Let $T$ be a tree $\T_r$. Let $\Or$
	be a maximal orthant containing $T$. The \emph{previstal facet}, ${\cal V}(T,\Or ; \A, \B)$, is the set of variable trees, $X$ $\in \Or$, for which the geodesic joining $X$ to $T$ has support $(\A,\B)$ satisfying $(P2)$ and
	$(P3)$ with strict inequalities.
\end{defn}

The description of the previstal facet $\V(T,\Or; \A, \B)$ becomes linear after a simple change of variables. Let $x_e$ denote the coordinate in $\T_r$ indexed by edge $e$.
\begin{defn}
	\cite[Def. 3.4]{Miller} 
	The \emph{squaring map} $\T_r \to \T_r$ acts on $x \in T_r \subset \mathbb{R}^{E}_+$ by squaring
	coordinates:
	\begin{displaymath}
		(x_e | e\in E) \to (\xi_e | e\in E), \textrm{ where } \xi_e = x^2_e
	\end{displaymath}
\end{defn}

Denote by $\T_r^2$ the image of this map, and let $\xi_e = x_e^2$ denote the coordinate indexed by
$e \in E$. The image of an orthant in $\T_r$ is then the equivalent orthant in $T_r^2$, and the image of a previstal facet ${\cal V}(T,\Or;\A,\B)$ in $\T_r^2$ is a \emph{vistal facet} denoted by ${\cal V}^2(T,\Or;\A,\B)$. 
With this change of variables, $\norm{S} = \sum_{e \in S}{\xi_e}$.

The squaring map induces on the Fr\'{e}chet function $F$ a corresponding pullback function
\begin{displaymath}
	F^2(\xi) = F(\sqrt{\xi})\textrm{, where}(\sqrt{\xi})_e = \sqrt{\xi_e}.
\end{displaymath}
Since the Fr\'{e}chet function $F(X)$ has a unique minimizer $F^2(\xi)$ must also have a unique minimizer.
\begin{prop}
	\cite[Prop. 3.5]{Miller} The vistal facet ${\cal V}^2(T,\Or;\A,\B)$ is a convex polyhedral cone in $\T_r^2$ defined by the following inequalities on $\xi \in \mathbb{R}^{r-2}$.
	\begin{itemize}
		\item[(O)] $\xi \in \Or$; that is, $\xi_e \geq 0$ for all $e \in E$, and $\xi_e = 0$ for $e \notin  E$, where $\Or = \mathbb{R}^{r-2}_{\geq 0}$. 
		\item[(P2)] $\norm{B_{i+1}}^2$ $\displaystyle{ \sum_{e \in A_i} \xi_e \leq \norm{B_i}^2 \sum_{e \in A_{i+1}}{\xi_e}}$
		for all $i = 1,\ldots, k-1$.
		\item[(P3)] $\norm{B_i \setminus J}$ $\displaystyle{ \sum_{e \in A_i \setminus I} \xi_e \geq \norm{J} \sum_{e \in I}{\xi_e}}$
		for all $i = 1, \ldots, k$ and subsets $I \subset A_i$, $J \subset B_i$ such that $I \cup J$ is compatible.
	\end{itemize}
\end{prop}

\begin{prop}
	\cite[Prop. 3.6]{Miller} The vistal facets are of dimension $2r-1$, have pairwise disjoint interiors, and cover $\T_r^2$.
	A point $\xi \in \T^2_r$ lies interior to a vistal facet $\V^2(T,\Or; \A, \B)$ if and only if the inequalities in 
	(O), (P2), and (P3) are strict.
\end{prop}

Points which are not on the interior of vistal facets are in \emph{vistal cells},
the faces of vistal facets.
A point $\xi$ is on a vistal facet precisely when some of the inequalities in (O), (P2) or (P3) are satisfied at equality. In such a situation, there is more
than one valid support for the geodesic from $T$ to the pre-image $X$ of $\xi$.

A system of equations defining a vistal cell can be formed by combining 
the systems of equations from adjacent vistal facets and forcing
appropriate constraints to equality. 
A canonical description of vistal cells is given in \cite[Sec. 3.2.5]{Miller}.

Let $T^1,\ldots,T^n$ be a set of points in $\T_r$.
A region $\V$ in squared treespace where the geodesics $\Gamma_{XT^1},\dots, \Gamma_{XT^n}$ can 
be represented by a fixed set of supports is called 
a \emph{multi-vistal cell}. 
A multi-vistal cell is an intersection of vistal facets of $T^1,\ldots,T^n$.
Multi-vistal cells and their pre-images in $\T_r$, \emph{pre-multi-vistal cells},
are regions where the Fr\'{e}chet function can be represented
with a fixed algebraic form.

The systems of equations defining pre-vistal facets and pre-vistal cells
are quadratic cones with cone points at the origin of treespace. 
In squared treespace, the vistal facets and vistal cells are polyhedral cones.
Multi-vistal facets are  also polyhedral cones with cones points
at the origin of treespace, because they are
intersections of polyhedral cones with cone points
at the origin of treespace. 
This nice geometric structure is useful both in determining 
when a search point is on the boundary of a vistal cell, and thus when
the objective function has multiple forms, and for dynamically updating the objective function during line searches,
as described in Sec. \ref{sec:updating}.

\section{Differential analysis of the Fr\'{e}chet function in treespace}\label{sec:DiffAnalysis}

Analysis of how $F(X)$ changes with respect to $X$ provides 
useful insights for designing fast optimization algorithms.
This analysis is aimed at providing summaries for how the value of $F(X)$
changes with respect to $X$. 
These results also play an important role in Ch. \ref{ch:Stickiness}, which focuses
on stickiness of Fr\'{e}chet means in treespace. 


Let $X$ and $Y$ be points in $\T_r$ such that $X$ and $Y$ share a multi-vistal facet.
If this is the case, then either (i) $X$ and $Y$ have the same topology, (ii) $X$ is a contraction of $Y$ or (iii) $Y$ is a contraction of $X$.
Assume that if the topologies of trees $X$ and $Y$ differ then $X$ is a contraction of $Y$, that is $\Or(X) \subseteq \Or(Y)$. 
Let $\Gamma(X,Y;\alpha)$ be the parameterized geodesic from $X$ to $Y$
with $0 \leq \alpha \leq 1$.
\sean{
	Let $p_e = |e|_Y-|e|_X$.
	Let $Z_\alpha$ be a point on the geodesic segment from $X$ to $Y$, such that $|e|_Z=|e|_X + \alpha p_e$, where $0 \leq \alpha \leq 1$. }
\begin{defn}\label{def:DirDer}
	The \emph{directional derivative from $X$ to $Y$} is 
	\begin{align}
		F'(X,Y) =&
		\lim_{\alpha \to 0}\frac{F(\Gamma(X,Y;\alpha))-F(X)}{\alpha}
	\end{align}
\end{defn}
\sean{
	The focus here is the directional derivative. 
	Another common notion of derivative is
	the \emph{normalized directional derivative from $X$ to $Y$},
	\begin{align}
		\bar{F}'(X,Y) =&
		\lim_{d(X,Z) \to 0}\frac{F(Z)-F(X)}{d(X,Z)}=\lim_{\alpha \to 0}\frac{F(Z)-F(X)}{\alpha d(X,Y)}
	\end{align}
	Both types are mentioned here to clarify from the beginning
	that the focus of this analysis is the directional derivative and not
	the normalized directional derivative.}
The main results of this section are summarized as follows:
Cor. \ref{cor:DirDerTangentSpace} gives
the value of the directional derivative when $\Or(X)=\Or(Y)$ and
Thm. \ref{thm:DirDerValue} gives the value of the directional
derivative when $\Or(X) \subseteq \Or(Y)$,
when assuming $Y$ is contained on the interior of a
multi-vistal facet.
In Lem. \ref{lem:DDWellDefined} we show that when
when $Y$ is on a multi-vistal face 
the value of the directional derivative 
can be expressed equivalently using any of the
representations for the geodesics from $T^1,\ldots,T^n$ to $Y$.
In Lem. \ref{lem:DirDerContinuous} and Lem. \ref{lem:DirDerConvex}
we show that the directional derivative is continuous and convex 
with respect to $Y$ on $\T_r$. Thm. \ref{thm:DirDerDecomposition} states that the value of the directional derivative can be decomposed into
a contribution from the change in $F(X)$ resulting in adjusting positive
length edges in $X$, and a contribution from the 
change in $F(X)$ resulting in increasing the lengths of edges
from zero. 

\sean{
	The assumption that $X$ and $Y$ are in a shared multi-vistal cell can be made without loss of generality
	because directional derivatives 
	describe local behavior. Even if $X$ and $Y$ are not close enough to share a multi-vistal cell, the end of the geodesic containing $X$
	will have a small segment completely contained in the same 
	multi-vistal cell as $X$. Thus as $Z$ approaches $X$ along the geodesic between $X$ and $Y$, $Z$ will eventually stay within a shared mult-vistal cell with $X$.
}

When both $X$ and $Y$ are in the relative interior of the same maximal orthant of treespace,
where the gradient at $X$ is well defined in $\Or(Y)$, the directional derivative can be expressed
in terms of the gradient at $X$ inside $\Or(Y)$. However when $\Or(X) \subsetneqq \Or(Y)$, the
gradient at $X$ might not be well defined in $\Or(Y)$. Analysis of the directional derivative in the
later situation, which is one of the main focuses of this section, is important because it
facilitates concise specification of optimality conditions and an efficient algorithm for verifying that a point on a lower dimensional face of an orthant $\Or$ is the minimizer of the Fr\'{e}chet function within $\Or$. 
\begin{thm}\label{thm:Grad}
	\cite{Miller}[Thm. 2.2] The gradient of $F$ is well defined on the interior of every maximal orthant $\Or$.
\end{thm}
\noindent Idea of proof.
The Fr\'{e}chet function is smooth in each multi-vistal facet, and it can be shown that the gradient
function has the same value in every multi-vistal facet containing $X$ in the interior of $\Or$. Therefore the gradient is well -defined on the interior of $\Or$. 

\begin{cor}\label{cor:DirDerTangentSpace}
	When $\Or(X)=\Or(Y)$ the value of directional derivative from $X$ to $Y$ can be expressed
	in terms of the gradient at $X$, and the differences in edge lengths $p_e =|e|_Y-|e|_X$, as
	\begin{align}
		F'(X,Y) = \sum_{e\in E_X} p_e \left[\nabla F(X)\right]_e
	\end{align}
\end{cor}
\begin{proof}
	Expression of a directional derivative of a smooth function in terms of its gradient is a standard technique in calculus.
\end{proof}

The gradient may not be well defined on a lower-dimensional orthant of treespace.
For a point on a lower dimensional orthant of
treespace, a well defined analogue to the gradient is the \emph{restricted gradient}. 
\sean{
	The following lemma is used for justifying that the restricted gradient (in the proceeding definition)
	is well-defined on the interior of a multi-vistal cell relative to a particular orthant of treespace. 
	\begin{lem}
		Let $X$ be a point in the relative interior of a multi-vistal cell in any orthant of treespace, $\Or$, 
		and assume, without loss of generality, that $X$ does not have any zero length edges.
		The geodesic from $X$ to $T^i$ has exactly one valid support sequence.
	\end{lem}
	\begin{proof}
		When $X$ is interior to a maximal orthant, then it is on the interior of a vistal facet.
		When $X$ is on the interior of a lower-dimensional orthant it is on the shared orthant boundaries
		of several multi-vistal cells, and the positivity constraints $(O)$ for each of these
		are at equality.
		However, since $X$ is in the interior of a multi-vistal cell relative to $\Or$, all (P2) and (P3) inequalities hold at strict inequality. Therefore there is only one valid support sequence for the geodesic from $X$ to $T^i$.
	\end{proof}
}

\begin{defn}\label{def:FrGrad}
	Let $(A^i_1,B^i_1),\ldots,(A^i_{k^i},B^i_{k^i})$ be a support
	sequence for the geodesic from $X$ to $T^i$. The \emph{restricted gradient}
	is the vector of partial derivatives which correspond to points $Y$ with $\Or(X) \subseteq \Or(Y)$ and $Y-X$ parallel to the axes of $\Or(X)$.
	If $|e|_X >0$ then
	\begin{align}
		\left[\nabla F(X)\right]_e &= \lim_{\Delta e \to 0} \frac{F(X+\Delta e)-F(X)}{\Delta e}\\
		&= \sum_{i=1}^n\left\{ \begin{array}{ll}
			|e|_X \left(1+ \frac{||B^i_l||}{||A^i_l||}\right) & \textrm{if $e \in A_l^i$} \\
			\left(|e|_X-|e|_{T^i} \right) & \textrm{if $e \in C^i$}\\
		\end{array} \right.
	\end{align}
	and  if $|e|_X = 0$ then $\left[\nabla F(X)\right]_e=0$. 
\end{defn}
When $X$ is on the interior of a maximal orthant of treespace then the restricted gradient is the same as the gradient.
Note that in the case when $A^i_l=\{e\}$, $|e|_X \left(1+ \frac{||B^i_l||}{||A^i_l||}\right)=|e|_X + \norm{B^i_l}$.

Second order derivatives will be used in calculating Newton directions
in Sec. \ref{sec:IntPointMethods}.

\begin{defn}\label{def:FrHess}
	Let $X$ be a point in the interior of a multi-vistal cell relative to an orthant $\Or$ of treespace. 
	The restricted Hessian on $\Or$ is the matrix of second order partial derivatives
	with entries having the following values:
	\begin{equation}\label{FrHess}
		\left[\nabla^2 F(X)\right]_{ef} = 2 \sum_{i=1}^{r} { \left\{ \begin{array}{ll}
				1+\frac{\norm{B_l^i}}{\norm{A_l^i}} - \frac{\norm{B^i_l}}{\norm{A^i_l}^3}x_e^2 & \textrm{if $e=f$, $e \in A^i_l$, $|A^i_l|>1$} \\
				1& \textrm{if $e=f$, $e \in A^i_l$, $A^i_l=\{e\}$} \\
				1 & \textrm{if $e=f$ $e \in C^i$} \\
				-\frac{\norm{B^i_l}}{\norm{A_l^i}^3}x_e x_f & \textrm{if $e \neq f$ $e,f \in A_l^i$} \\
				0 & \textrm{otherwise}
			\end{array} \right. }
	\end{equation}
	If either $|e|_X=0$ or $|f|_X=0$ then $\left[\nabla^2 F(X)\right]_{ef}=0$. 
\end{defn}

\begin{thm}
	The value of the restricted gradient at a point $X$ can be expressed equivalently using the algebraic form of the Fr\'{e}chet function from any of the multi-vistal facets containing $X$.
\end{thm}
\begin{proof}
	The restricted gradient has the same value
	using any of the valid support sequences defined by vistal cells on the relative interior of $\Or$.
	Here we verify that at $X$ the gradient of $d^2(X,T^i)$ is the same for every valid support and signature. The gradient of $d^2(X,T^i)$ for the support $(\A,\B)$ is given as follows. Let the variable length of edge $e$ in $X$ be written as $x_e$.
	\begin{equation}\label{graddistance}
		\frac{\partial d^2(X,T^i)}{\partial x_e}= \left\{ \begin{array}{ll}
			2 \left(1+ \frac{||B_l||}{||A_l||}\right)x_e & \textrm{if $e \in A^i_l$} \\
			2 \left(x_e-|e|_T^i \right) & \textrm{if $e \in C^i$}
		\end{array} \right.
	\end{equation}
	The geodesic $\Gamma$ has a unique support$(\A,\B)$ satisfying
	\begin{equation}\label{facetcondition}
		\frac{\norm{A_1}}{\norm{B_1}} < \frac{\norm{A_2}}{\norm{B_2}} < \ldots < \frac{\norm{A_k}}{\norm{B_k}}.
	\end{equation}
	From \cite{Miller}, any other support $(\A',\B')$ for $\Gamma$ must have a signature $\mathcal{S}'$ in $(P3)$ with some equality subsequences. Suppose that $A_j'$ and $B_j'$ are in some equality subsequence satisfying $(P2)$ with $B_j'$ containing the edge $e$. Then for the support pair $A_i$ and $B_i$ such that $B_i$ contains $e$, it must hold that $\frac{\norm{A_i'}}{\norm{B_i'}} = \frac{\norm{A_j}}{\norm{B_j}}$. Now we can see that $ \left(1+ \frac{||B_j'||}{||A_j'||}\right)x_e = \left(1+ \frac{||B_i||}{||A_i||}\right)x_e$, and that the gradient of $d^2(X,T^i)$ is the same on every multi-vistal facet containing $X$ on the relative interior of $\Or$.
\end{proof}

Now we extend the results for directional derivatives to the situation when $\Or(X) \subset \Or(Y)$.
\begin{thm}\label{thm:DirDerValue}
	Suppose that $Y$ lies in the interior of multi-vistal facet $V_Y$,
	and $X$ is some point in $V_Y$. Let $(A^i_1,B^i_1),\ldots, (A^i_n,B^i_n)$ be the support 
	pairs for the geodesic from $Y$ to $T^i$ and let $C^i$ be the set of edges
	in $Y$ which are common in $T^i$. Let $E_X$ be the set of edges with positive lengths in $X$. 
	Let $P$ be the vector with components $p_e = |e|_Y-|e|_X$ so that $\Gamma(X,Y;\alpha)=X+\alpha P$, and let $Z_\alpha:=\Gamma(X,Y;\alpha)$. 
	Then the value of directional derivative from $X$ to $Y$ is 
	\begin{displaymath}
		F'(X,Y)=\sum_{e \in E_X} p_e \left[\nabla F(X)\right]_e + 2 \sum_{i=1}^n \left ( \sum_{l: \norm{A^i_l}_X = 0} \left(  \norm{A^i_l}_P\norm{B^i_l}_T \right) -\sum_{e \in C^i \setminus E_X} p_e|e|_{T^i}  \right).
	\end{displaymath} 
\end{thm}
\begin{proof}
	Let $Z$ be a point on the geodesic segment between $X$ and $Y$. The length of edge $e$ in $Z$ be $|e|_Z = |e|_X + \alpha p_e$.
	The Fr\'{e}chet function is the sum of squared distances from a variable point to each of the data points $T^1,\ldots, T^n$, so the directional derivative of the Fr\'{e}chet function
	is the sum over the indexes of the data points of the directional derivatives of
	the square distances.  
	\begin{align}
		F'(X,Y) &=
		\lim_{\alpha \to 0}\frac{F(Z_\alpha)-F(X)}{\alpha}\\
		& = \lim_{\alpha \to 0} \frac{\sum_{i=1}^n d^2(Z_\alpha,T^i)-\sum_{i=1}^n d^2(X,T^i)}{\alpha}\\
		& = \sum_{i=1}^n \left ( \lim_{\alpha \to 0} \frac{d^2(Z_\alpha,T^i)-d^2(X,T^i)}{\alpha} \right )
	\end{align}
	For a set of edges $A$, let $\norm{A}_X = \sqrt{ \sum_{e \in A} |e|_X}$.  If an edge $e$ has zero length in a tree, $X$, or is compatible with $X$ but not present then take $|e|_X$ to be $0$.
	The squared distance from $Z_\alpha$ to $T^i$ can be expressed as 
	\begin{align}
		d^2(Z_\alpha,T^i) & = \sum_{l=1}^{k^i} (\norm{A^i_l}_{Z_\alpha}+\norm{B^i_l})^2+\sum_{e\in C^i} (|e|_{Z_\alpha}-|e|_{T^i})^2.
	\end{align}
	The squared distance has three types of terms: a term representing
	the contribution from common edges, terms for support pairs with $\norm{A^i_l}_X>0$, and terms for support pairs with $\norm{A^i_l}_X=0$.
	In the first two cases the gradient is well-defined, and taking the inner-product of the directional vector and the gradient will yield their contributions
	to the directional derivative. In the third case the gradient is undefined, and its value will be obtained by analyzing the limit directly as follows.
	\begin{align}
		&\left ( \lim_{\alpha \to 0} \frac{\sum_{l=1}^{k^i} (\norm{A^i_l}_{Z_\alpha}+\norm{B^i_l})^2-\sum_{l=1}^{k^i} (\norm{A^i_l}_X+\norm{B^i_l})^2}{\alpha} \right ).\label{eq:DirDerSqDist2}
	\end{align}
	Bringing out the sum and canceling in the numerators yields,
	\begin{align}
		\sum_{l=1}^{k^i} \lim_{\alpha \to 0} \frac{\norm{A^i_l}^2_Z-\norm{A^i_l}^2_X+2\norm{B^i_l}\left(\norm{A^i_l}_Z-\norm{A^i_l}_X\right)}{\alpha}.
	\end{align}
	The limit of the fraction can be split into the sum of two limits, 
	\begin{align}
		&\lim_{\alpha \to 0} \frac{\norm{A^i_l}^2_Z-\norm{A^i_l}^2_X+2\norm{B^i_l}\left(\norm{A^i_l}_Z-\norm{A^i_l}_X\right)}{\alpha}\\
		=&\lim_{\alpha \to 0} \frac{\norm{A^i_l}^2_Z-\norm{A^i_l}^2_X}{\alpha}+\lim_{\alpha \to 0}\frac{2\norm{B^i_l}\left(\norm{A^i_l}_Z-\norm{A^i_l}_X\right)}{\alpha} .\label{eq:LimitFrac1}
	\end{align}
	If every edge in $A^i_l$ has length zero in $X$, and thus $\norm{A^i_l}_X=0$,
	the expression the limit on the left is zero and the limit on the right simplifies to 
	\begin{align}
		2\norm{B^i_l} \norm{A^i_l}_P. \label{eq:Incomp2}
	\end{align}

	The partial derivative of the squared distance from $X$ to $T^i$ with respect to the length of edge $e$, that is the component for edge $e$ in the restricted gradient vector, is
	\begin{align}
		[\nabla d^2(X,T^i)]_e = \left\{ \begin{array}{ll}
			|e|_X \left(1+ \frac{||B^i_l||}{||A^i_l||}\right) & \textrm{if $e \in A_l^i$} \\
			\left(|e|_X-|e|_{T^i} \right) & \textrm{if $e \in C^i$}.
		\end{array} \right.
	\end{align}
	The directional derivative of the squared distance simplifies to 
	\begin{align}
		&\lim_{\alpha \to 0} \frac{d^2(Z,T^i)-d^2(X,T^i)}{\alpha}\\
		&=\sum_{e \in  E_X} p_e [\nabla d^2(X,T^i)]_e+2\sum_{l:\norm{A^i_l}_X = 0}\left( \norm{B^i_l} \sqrt{\sum_{e\in A^i_l} p_e^2} \right)-2\sum_{e \in C^i \setminus E_X} p_e |e|_{T^i}.
	\end{align}
	Summing the directional derivatives of the squared distances over $T^1,\ldots, T^n$ yields
	the expression for the value of the directional derivative in the theorem.
\end{proof}
We now extend the results to the situation where $Y$ is allowed to be
on a vistal face.
In this situation there can be multiple
valid support sequences for the geodesics from $Y$ to $T^1, \ldots, T^n$.
\begin{lem}\label{lem:DDWellDefined}
	Suppose that $X$ and $Y$ are in the same multi-vistal facet, $\V$, and
	that $Y$ is on a face of $\V$ on the interior of an orthant.  The value of the directional derivative can be expressed equivalently using any valid support sequences for the geodesics from $Y$ to $T^1,\ldots,T^n$.
\end{lem}
\begin{proof}
	The form of $F'(X,Y)$ is constant within an open multi-vistal facet, and changes
	at boundaries of vistal facets.
	When $Y$ reaches the boundary of a vistal facet - that is either at least one of the (P2)
	constrains reaches equality, at least one of the (P3) constraints reaches equality, 
	or when the length of an edge reaches zero or increases from zero - this is called
	the collision of $Y$ with the boundary of that vistal facet. 
	A point $T^i$,  and associated geodesic $\Gamma(T^i,Y)$ are said to generate the vistal facet collision.
	When $Y$ collides with a (P2) boundary of a vistal facet at least two support pairs for the geodesic merge; and when $Y$ collides with a (P3) boundary 
	at least two support pairs for the geodesic
	could be split in such a way that the resulting support is valid. 
	In either case there are at least two valid forms for the geodesic.
	Let $(C_1,D_1),(C_2,D_2)$ be support pairs which are formed from a partition of the support pair $(A^i_l,B^i_l)$, such that either of the following support sequences for the geodesic from $Y$ to $T^i$ is valid: $(A^i_1,B^i_1),\ldots,(A^i_l,B^i_l),\ldots,(A^i_m,B^i_m)$ or $(A^i_1,B^i_1),\ldots,(C_1,D_1),(C_2,D_2),\ldots,(A^i_m,B^i_m)$; and $ \frac{\norm{C_1}}{\norm{D_1}}=\frac{\norm{C_2}}{\norm{D_2}}=\frac{\norm{A^i_l}}{\norm{B^i_l}} $. 
	Rescaling the lengths of edges in $A^i_l$ does not change the form of the geodesic for small $\alpha$ and $l \leq k$. Parameterizing the lengths of edges in terms of $\alpha$ and canceling $\alpha$ yields  $ \frac{\sqrt{\sum_{e\in C_1} p_e^2}}{\norm{D_1}}=\frac{\sqrt{\sum_{e\in C_2} p_e^2}}{\norm{D_2}}=\frac{\sqrt{\sum_{e\in A^i_l} p_e^2}}{\norm{B^i_l}} $. That fact, and the fact that $C_1 \cup C_2$ partition $A^i_l$ and $D_1 \cup D_2$ partition $B^i_l$ implies that 
	$\norm{D_1}\sqrt{\sum_{e\in C_1} p_2^2} + \norm{D_2} \sqrt{\sum_{e\in C_2} p_e^2}=\norm{B^i_l}\sqrt{\sum_{e\in A^i_l} p_e^2}$. Thus the directional derivative is continuous across vistal facet boundaries from (P2) and (P3) constraints. 
\end{proof}

Now we extend the results for directional derivatives to directions issuing from $X$ to points $Y$ in a small enough radius such that $\Or(X) \subseteq \Or(Y)$ 
and $X$ and $Y$ share a multi-vistal facet. 
\sean{
	\begin{rem}
		An $\epsilon$-neighborhood of $X$ is the set $N_{\epsilon}(X)=\{Y | d(X,Y)^2 < \epsilon\}$. The set ${N}_{\epsilon}(X)$ can be decomposed into
		a union of sets which are intersections of Euclidean orthants with hyperspheres, that is
		$N_{\epsilon}(X)=\cup_{c \in C_q} \{Y \in \Or(q) |d(X,Y)^2\}$.
		For a small neighborhood around a point $X$, $N_\epsilon(X)$, the set of geodesic segments $\{\Gamma(X,Y)| Y \in {N}_{\epsilon}(X)\}$ are either
		contained in the same orthant as $X$ or 
		move into orthants which contain $X$ in their closure.
		Geodesic passing through $X$ in the minimal orthant containing $X$, $\Or(X)$
		equate to changes in the positive length edges of $X$.
		Geodesic issuing into higher-dimensional orthants containing
		$X$ equate to changing edge lengths and increasing the lengths of 
		edges from zero to form a tree which contains $X$ as a contraction.
	\end{rem}
}

\sean{
	\begin{lem}
		For any $X$, there exists an positive $\epsilon$ such that $\Or(X) \subseteq \Or(Y)$ and $X$ and $Y$ share a vistal facet for all $Y$ in $N_{\epsilon}(X)$.
	\end{lem}
	\begin{proof}
		The vistal facets cover treespace. 
	\end{proof}
}

\begin{lem}\label{lem:DirDerContinuous}
	The directional derivative, $F'(X,Y)$, is a continuous function of $Y$ over the set of $Y$ such that $\Or(X) \subseteq \Or(Y)$ and $X$ and $Y$ share a vistal facet.
\end{lem}
\begin{proof}
	The directional derivative is a continuous function at the faces of orthants because when an edge length $|e|_Y$ increases from zero its contribution to $F'(X,Y)$ is a continuous function which starts at the value zero. Thus, when the topology of $Y$ changes $F'(X,Y)$ changes continuously as a function of the edge lengths. 
\end{proof}

The following lemma is used in the proof of Lem. \ref{lem:DirDerConvex}.

\begin{lem}\label{lem:LocalGeoTriangle}
	Let $Y^0$ and $Y^1$ be a points in $\T_r$ such that $\Or(X) \subseteq \Or(Y^0)$
	and $\Or(X) \subseteq \Or(Y^1)$.
	Let $Y^t=\Gamma(Y^0,Y^1;t)$ be the point which is proportion $t$ along the geodesic from $Y^0$ to $Y^1$.
	The point which is $\alpha$ proportion along the geodesic from $X$ to $Y^t$ is $t$ proportion along the geodesic between the point $\Gamma_{X,Y^0}(\alpha)$ and the point $\Gamma_{X,Y^1}(\alpha)$; that is, $\Gamma(X,Y^t;\alpha)=\Gamma(\Gamma(X,Y^0;\alpha),\Gamma(X,Y^1;\alpha);t)$.
\end{lem}
\begin{proof}
	Let $Y^0(\alpha) = \Gamma_{X Y^0} (\alpha)$ and let $Y^1(\alpha) = \Gamma_{X Y^1} (\alpha)$.
	Let $C = E_{Y^0(\alpha)} \cap E_{Y^1(\alpha)}$. By definition $E_X \subseteq C$.
	The length of $e$ in $Y^0(\alpha)$ is
	\begin{align}\label{eq:LengthY0}
		|e|_{Y^0(\alpha)}=
		\left\{ \begin{array}{ll}
			|e|_X + \alpha |e|_{Y^0} & \textrm{if $e \in C$} \\
			\alpha |e|_{Y^0} & \textrm{if $e \in E_{Y^0} \setminus C$}\\
		\end{array} \right.
	\end{align}
	and the length of $e$ in $Y^1(\alpha)$ is
	\begin{align}\label{eq:LengthY1}
		|e|_{Y^1(\alpha)}=
		\left\{ \begin{array}{ll}
			|e|_X + \alpha |e|_{Y^1} & \textrm{if $e \in C$} \\
			\alpha |e|_{Y^1} & \textrm{if $e \in E_{Y^1} \setminus C$}.\\
		\end{array} \right.
	\end{align}
	A geodesic support sequence which is valid for the geodesic between $Y^0$ and $Y^1$ is valid for the the geodesic between $Y^0(\alpha)$ and $Y^1(\alpha)$.
	The incompatibilities of edges in $A$ and $B$ are the same for any $\alpha$.
	Suppose that a support sequence satisfies (P2) and (P3) for some $\alpha$. Factoring out $\alpha$ from the numerators and denominators of the $(P2)$ and $(P3)$ ratios reveals that the combinatorics of the geodesic between
	$Y^0(\alpha)$ and $Y^1(\alpha)$ depends on the relative proportions
	of lengths of edges in $Y^0$ and $Y^1$, and not on the value of $\alpha$.
	That is,
	\begin{align}
		\frac{ \norm{A_l}_{Y^0(\alpha)}}{\norm{B_l}_{Y^1(\alpha)}}=\frac{\sqrt{\sum_{e \in A_l} \alpha |e|_{Y^0}}}{\sqrt{\sum_{e \in B_l} \alpha |e|_{Y^1}}} = \frac{ \norm{A_l}_{Y^0}}{\norm{B_l}_{Y^1}}.
	\end{align}
	Now we show that $|e|_{\Gamma_{Y^0(\alpha) Y^1(\alpha)}(t)} = |e|_{\Gamma_{X Y^t}(\alpha)}$.
	The combinatorics of the geodesic between $Y^0(\alpha)$ and $Y^1(\alpha)$
	do not depend on $\alpha$, therefore which edges have positive lengths in the $l^{th}$ leg of $\Gamma_{Y^0(\alpha)Y^1(\alpha)}$ does not depend on $\alpha$.
	The length of edge $e$ at $\Gamma_{Y^0(\alpha) Y^1(\alpha)}(t)$ is
	\begin{align}
		|e|_{\Gamma_{Y^0(\alpha) Y^1(\alpha)}(t)}&=\displaystyle\left\{\begin{array}{ll}
			\frac{(1-t)\norm{A_j}_\alpha-t \norm{B_j}_\alpha}{\norm{A_j}_\alpha}|e|_{Y^0(\alpha)}&e\in A_j\\[1em]
			\frac{t \norm{B_j}_\alpha-(1-t)\norm{A_j}_\alpha}{\norm{B_j}_\alpha}|e|_{Y^1(\alpha)}&e\in B_j\\[1.5em]
			(1-t)|e|_{Y^0(\alpha)}+t |e|_{Y^1(\alpha)}&e\in C.
		\end{array}\right.\\
	\end{align}
	Substituting $\norm{A_j}_\alpha=\alpha \norm{A_j}$, $\norm{B_j}_\alpha=\alpha \norm{B_j}$, \ref{eq:LengthY0}, and \ref{eq:LengthY1} yields
	\begin{align}
		|e|_{\Gamma_{Y^0(\alpha) Y^1(\alpha)}(t)}&=\displaystyle\left\{\begin{array}{ll}
			\alpha \frac{(1-t)\norm{A_j}-t \norm{B_j}}{\norm{A_j}}|e|_{Y^0}&e\in A_j\\[1em]
			\alpha \frac{t \norm{B_j}-(1-t)\norm{A_j}}{\norm{B_j}}|e|_{Y^1}&e\in B_j\\[1.5em]
			|e|_X + \alpha((1-t)|e|_{Y^0}+t |e|_{Y^1})&e\in C.\\
		\end{array}\right.
	\end{align}
	Now the length of $e$ in $\Gamma_{XY^t}(\alpha)$ is 
	\begin{align}\label{eq:LengthYt}
		|e|_{\Gamma_{XY^t}(\alpha)}=
		\left\{ \begin{array}{ll}
			|e|_X + \alpha |e|_{Y^t} & \textrm{if $e \in C$} \\
			\alpha |e|_{Y^t} & \textrm{if $e \in E_{Y^t} \setminus C$}.\\
		\end{array} \right.
	\end{align}
	The length of $e$ in $Y^t$ is given by 
	\begin{align}
		|e|_{Y^t}&=\displaystyle\left\{\begin{array}{ll}
			\frac{(1-t)\norm{A_j}-t \norm{B_j}}{\norm{A_j}}|e|_{Y^0}&e\in A_j\\[1em]
			\frac{t \norm{B_j}-(1-t)\norm{A_j}}{\norm{B_j}}|e|_{Y^1}&e\in B_j\\[1.5em]
			((1-t)|e|_{Y^0}+t |e|_{Y^1})&e\in C.\\
		\end{array}\right.
	\end{align}
	Therefore  $|e|_{\Gamma_{Y^0(\alpha) Y^1(\alpha)}(t)} = |e|_{\Gamma_{X Y^t}(\alpha)}$ holds.
\end{proof}

\begin{lem}\label{lem:DirDerConvex}
	The directional derivative $F'(X,Y)$ is a convex function of $Y$ over the set of $Y$ such that $\Or(X) \subseteq \Or(Y)$ and $X$ and $Y$ share a vistal facet.
\end{lem}
\begin{proof}
	Let $Y^0$ and $Y^1$ be a points in $\T_r$ such that $\Or(X) \subseteq \Or(Y^0)$
	and $\Or(X) \subseteq \Or(Y^1)$.
	Let $Y^t$ be the point which is proportion $t$ along the geodesic from $Y^0$ to $Y^1$.
	Let $\Gamma_{X Y^t}(\alpha):[0,1] \to \T_r$ be a function which parameterizes the geodesic from $X$ to $Y^t$.
	Using Lem. \ref{lem:LocalGeoTriangle} and the strict convexity of $F$ together yields
	\begin{align}
		F(\Gamma_{XY^t}(\alpha))<F(\Gamma_{XY^0}(\alpha))(1-t)+F(\Gamma_{XY^1}(\alpha))t. \label{eq:GeodesicConvexity}
	\end{align}
	The directional derivative from $X$ in the direction of $\Gamma_{XY^t}(\alpha)$ is
	\begin{align}
		F'(X,Y^t)& = \lim_{\alpha \to 0} \frac{ F(\Gamma_{XY^t}(\alpha))-F(X)}{\alpha}.
	\end{align}
	Substituting for $F(\Gamma_{XY^t}(\alpha))$ using the inequality on line (\ref{eq:GeodesicConvexity}) yields,
	\begin{align}
		F'(X,Y^t)& \leq \lim_{\alpha \to 0} \frac{ F(\Gamma_{XY^0}(\alpha))(1-t)+F(\Gamma_{XY^1}(\alpha))t-F(X)}{\alpha}.
	\end{align}
	Note that strict inequality may not hold even though the Fr\'{e}chet function is convex because
	in the limit the value may approach an infimum. 
	Simplifying by separating the fraction and limit reveals that the directional derivative is convex in $Y$,
	\begin{align}
		F'(X,Y^t) &\leq (1-t)\lim_{\alpha \to 0}\frac{F(\Gamma_{XY^0}(\alpha))-F(X)}{\alpha} + t\lim_{\alpha\to 0}\frac{F(\Gamma_{XY^1}(\alpha))-F(X)}{\alpha}\\
		&= (1-t) F'(X,Y^0) + tF'(X,Y^1).
	\end{align}
\end{proof}

\begin{lem}
	Let $X$ and $Y$ be points such that $\Or(X)\subseteq \Or(Y)$. $F'(X,Y)$ is a $C^1$ function of $Y$ on the interior of the orthant $\Or$.
\end{lem}
\begin{proof}
	Within any fixed multivistal face the algebraic form of $F$ is a sum of smooth functions, and the restricted gradient function is continuous at the boundaries of multivistal faces relative to $\Or$. 
\end{proof}

\begin{defn}
	Let $(A_1,B_1),\ldots,(A_k,B_k)$ be a support sequence for the geodesic from $Y$ to a tree $T$, as in the definition of directional derivative above, Def. \ref{def:DirDer}.
	Let any support pair $(A_l,B_l)$ such that $\norm{A_l}_X =0$ be called a \emph{local support pair}.
\end{defn}

Local support pairs will be the earliest support pairs in a support sequence for the geodesic between $Y$ and $T$. 
$Y$ and $X$ share a vistal facet; that is, their geodesics to $T$ can be represented with the same support sequence. According to $(P2)$, any support pair such that $\norm{A_l}_X =0$ must be among the first support pairs in the support sequence.
Thus, let $(A_1,B_1),\ldots,(A_m,B_m)$ be local support pairs, and let $(A_{m+1},B_{m+1}),\ldots,(A_k,B_k)$ be the rest of the geodesic support sequence being used to represent the geodesic between $Y$ and $T$.

Let $\tilde{B}$ be all edges from $T$ which are incompatible with at least one edge in $Y$ but not incompatible with any edge in $X$ and let $\tilde{A}$ be all edges from $Y$ which are incompatible
with some edge in $\tilde{B}$. 

\begin{lem}\label{lem:LocalSupportPairsComposition}
	Then any sequence of local support pairs, $(A_1,B_1),\ldots,(A_m,B_m)$ have the property that the sets $A_1,\ldots,A_m$ partition $\tilde{A}$ and $B_1,\ldots,B_m$ partition $\tilde{B}$.
\end{lem}
\begin{proof}
	
	Any edge in $T$ which is incompatible with an edge in a local support pair $(A_l,B_l)$ is compatible with every edge in $X$ because for a local support pair $\norm{A_l}_X=0$. Therefore any edge from $T$ in a local support pair must be in $\tilde{B}$. 
	
	An edge in $\tilde{B}$ is compatible with every edge in $X$ therefore cannot be in any of the support pairs with edges from $X$ and thus must be in a local support pair.
	
	Suppose an edge, $e$, from $Y$ is in a local support pair, $(A_l,B_l)$; then it must incompatible with at least one edge in $T$. All the edges which are in $B_l$ must be compatible with all edges in $X$ because  $\norm{A_l}_X=0$. 
	Since $e$ is incompatible with some edge in $T$ that is not incompatible with any edge in $X$, $e$ must be in $\tilde{A}$.
	
	Let $e$ be an edge in $\tilde{A}$. Edge $e$ is not in $X$, and edge $e$ is incompatible with at least one edge in $T$ which no edge in $X$ is incompatible with. Edge $e$ must be in a support pair so that along the geodesic the length of $e$ contracts to zero before all the edges in $\tilde{B}$ can switch on. Therefore $e$ must be a support pair with 
	at least one of the edges in $\tilde{B}$ that it is incompatible with.
	Since all edges in $\tilde{B}$ are in local support pairs, all edges in $\tilde{A}$ must also be in local support pairs. 
\end{proof}

\begin{defn}\label{defn:NormalSpace}
	Let $\Or^\perp (X)$ be the orthogonal space to $\Or(X)$ at $X$, that is the union of all orthogonal spaces in all orthants containing $\Or(X)$.
\end{defn}

\begin{cor}\label{cor:LocalSupportPairs}
	Let $X,Y\in \T_n$, with $\Or(X)\subseteq \Or(Y)$ and with $X$ and $Y$ in a common multi-vistal cell, $V_{XY}$, let $Y_X$ be the projection of $Y$ onto $\Or(X)$,
	and let $Y_\perp$ be the projection of $Y$ onto $\Or^\perp(X)$ at $X$.
	Then any sequence of local support pairs which is valid for the geodesic from $Y$ to $T$ is also valid for the geodesic from $Y_\perp$ to $T$ and vice versa. 
\end{cor}
\begin{proof}
	Lem. \ref{lem:LocalSupportPairsComposition} implies local support pairs for the geodesic from $Y$ to $T$ and $Y_\perp$ to $T$ would be composed from
	the same sets of edges. Factoring out $\alpha$, we see that the relative lengths of edges in $\tilde{A}$ are the same in $Y$ and $Y_\perp$.
\end{proof}

\begin{thm}\label{thm:DirDerDecomposition}
	(Decomposition Theorem for Directional Derivatives)
	Let $X,Y\in \T_n$, with $\Or(X)\subseteq \Or(Y)$ and with $X$ and $Y$ in a common multi-vistal cell, $V_{XY}$, let $Y_X$ be the projection of $Y$ onto $\Or(X)$,
	and let $Y_\perp$ be the projection of $Y$ onto $\Or^\perp(X)$ at $X$.
	Then,
	\begin{align}
		F'(X,Y) = F'(X,Y_X)+F'(X,Y_\perp).
	\end{align}
\end{thm}
\begin{proof}
	Note that since $X$ and $Y$ are in the same orthant, the geodesic $\Gamma_{XY}$ is just the line segment $XY$. Let $P=Y-X$, and let $P_X$ and $P_\perp$ be its decomposition into the parts corresponding to $Y_X$ and $Y_\perp$.
	Let $Z$ be a point on $XY$ denoted by $|e|_Z = |e|_X + \alpha p_e$.
	Let $Z_X=X+\alpha P_X$ and let $Z_\perp=X+\alpha P_\perp$ be the component of $Z$ orthogonal to $\Or(X)$. 
	By Cor. \ref{cor:DirDerTangentSpace}, the value of the directional derivative from $X$ to $Y_X$ is
	\begin{align}
		F'(X,Y_X)=\lim_{\alpha \to 0} \frac{F(Z_X)-F(X)}{\alpha}=\sum_{e\in E_X} p_e \left[\nabla F(X)\right]_e.
	\end{align}
	and the directional derivative from $X$ to $Y_\perp$ is
	\begin{align}
		F'(X,Y_\perp) &= \lim_{\alpha \to 0} \frac{F(Z_\perp)-F(X)}{\alpha}\\
		&=2 \sum_{i=1}^n \left ( \sum_{l: \norm{A^i_l}_X = 0} \left( \norm{B^i_l}\sqrt{\sum_{e \in A^i_l} p_e^2 } \right) -\sum_{e \in C^i \setminus E_X} |e|_{T^i} p_e  \right).
	\end{align}
	The partial derivative at $X$ is well defined and is equal to zero for edges which (i) have length zero in $X$, (ii) positive length in $Y$, and (iii) are in support pairs such that $\norm{A^i_l}_X>0$. Therefore
	the claim of the theorem holds. 
	\sean{
		By Cor. \ref{cor:LocalSupportPairs}, local support pairs for the geodesic from $Z_\perp$ to $T$ are valid local support pairs for the geodesic from $Z$ to $T$. 
		Therefore the claim of the theorem holds.}
\end{proof}

\section{Interior point methods for optimizing edge lengths}\label{sec:IntPointMethods}

In this section, the local search problem is defined, the fundamentals for iterative local search algorithm - i.e.\ initialization, an improvement method and an optimality qualification - are
discussed, and an iterative search algorithm is presented. 

Consider a variable tree $X \in \T_r$ and a fixed set of edges $E$. 
The goal is to minimize the Fr\'{e}chet function but under the restriction that 
the topology of $X$ may only have edges from $E$.
Under this restriction, the geometric location of $X$ is restricted to the orthant defined by the set of edges $E$, $\Or(E)$. 
As the edge lengths of $X$ vary the geodesic from $X$ to $T^i$ will also vary, 
and the support sequence $(\A^i,\B^i)=(A^i_1,B^i_1),\ldots,(A^i_{k^i}, B^i_{k^i})$ will change whenever $X$ crosses
the boundary of a vistal cell. 
Local search can be formulated as the following convex optimization problem.

\setcounter{equation}{0}\label{opt:P1}
\noindent Objective
\begin{align}
	\min \quad
	& F(X) =\sum_{i=1}^n \left(\sum_{l=1}^{k^i} (\norm{A^i_l}+\norm{B^i_l})^2 + \sum_{e \in C^i} (|e|_X-|e|_{T^i})^2\right)
\end{align}
Constraints
\begin{align}
	& |e|_X  \geq 0 \; \forall \; e \in E	&  & 
\end{align}
The minimizer of this optimization problem, $X^*$, satisfies $F(X^*)\leq F(X)$ for all $X$ in $\Or(E)$.

\subsubsection{Optimality Qualifications}
There are two cases for the optimal solution $X^*$: either every edge in $X^*$ has a positive length or at least one edge in $X^*$ does not. 
If every edge of $X^*$ has positive length, then $X=X^*$ if and only if $\nabla F(X)=0$ because the Fr\'{e}chet function is continuously differentiable in the interior of $\Or$. 
The optimality condition for a point
on a lower dimensional face of treespace must be expressed in terms of directional derivatives. In that case the optimality condition is
\begin{equation}
	\begin{array}{c c l}
		F'(X,Y) & \geq 0  & \textrm{for all $Y$ such that $\Or(X) \subseteq \Or(Y)$}
	\end{array}
\end{equation}
By using Thm. \ref{thm:DirDerDecomposition} to separate the directional derivative into the contribution
from the component of $Y$ in $\Or(X)$, and the component of $Y$ 
which is perpendicular to $\Or(X)$ the optimality condition becomes
\begin{equation}\label{generalOptCond}
	\begin{array}{c c l}
		\left[\nabla F(X)\right]_e & = 0 & \textrm{ for all $e: |e|_X > 0$} \\
		F'(X,Y) & \geq 0 & \textrm{for all $Y$ such that the component of $Y$ in $\Or(X)$ is 0}
	\end{array}
\end{equation}
For the local search problem i.e.\ identifying the minimizer
of the Fr\'{e}chet function on an orthant of treespace, $\Or(E)$
there must be a unique solution because the Fr\'{e}chet function is strictly convex and $\Or(E)$ is a convex set.
Also, optimality conditions for the local search problem 
are only different from global optimality conditions in one aspect, 
which is that rather than requiring $F'(X,Y)\geq 0$ for all
$Y$ perpendicular to $\Or(X)$, it is only necessary
to consider the subset of such points $Y$ that are in $\Or(E)$.

\noindent{\bf{Approximate optimality conditions}}\\
Conditions for a point $X$ on the interior of an orthant to be approximately optimal are:
\begin{equation}\label{eq:ApproxOptimal}
	\begin{array}{c c c c}
		& & \abs{\left[\nabla F(X)\right]_e} < \delta & \textrm{ for all $e: |e|_X > \epsilon$} \\
		\left[\nabla F(X)\right]_e \geq 0 & \textrm{ or } & \abs{\left[\nabla F(X)\right]_e} < \delta& \textrm{ for all $ e: |e|_X  < \epsilon$}
	\end{array}
\end{equation}
If these approximate optimality conditions are satisfied then $F(X^*)$ will not differ much
from the Fr\'{e}chet function value when the lengths of edges with positive derivatives
are set to 0.

\subsection{A damped Newton's method}\label{sec:DampedNewton}
The following algorithm is designed to find approximately optimal edge lengths for a fixed tree topology.
Detailed explanations for the steps 
of this algorithm are in the following subsections.

\begin{flushleft}
	{\bf Interior point algorithm for optimal edge lengths}\\
	{\bfseries input:} $T^1,T^2,\ldots,T^n,X^0 \in \T_r$, $\epsilon>0$, $\delta>0$, $0<c<1$\\
	{\bfseries while} approximate optimality conditions (\ref{eq:ApproxOptimal}) are not satisfied  \\
	{\bfseries do} \\
	\hspace*{1cm} compute a descent direction $P$ (Sec. \ref{sec:NewtonSteps})\\
	\hspace*{1cm} find a feasible step-length, $\alpha$, satisfying decrease condition (Sec. \ref{sec:StepLength})\\
	\hspace*{1cm} {\bfseries let} $X^{k+1}=X^k+\alpha P$\\
	\hspace*{1cm} {\bfseries if} $|e| < \epsilon$ {\bfseries then} remove $e$\\
	{\bfseries endwhile}
\end{flushleft}

\subsubsection{Newton steps}\label{sec:NewtonSteps}
A successful iterative algorithm will make substantial progress to an optimal point.
This can be achieved using a modified Newton's method. 
Newton's method uses a descent vector which points to the minimizer of a quadratic approximation of the objective function. The quadratic approximation in Newton's method uses the first three terms of the Taylor expansion of $F(X)$.
For the Fr\'{e}chet function the entries of the Hessian matrix of second order partial derivatives is given in Def. \ref{def:FrHess}.

The Hessian matrix is positive definite because the Fr\'{e}chet function
is strictly convex. 
Thus, when the Hessian and gradient are well defined,
the second order Taylor approximation is 
\begin{align}
	g(X; p)=& F(X)+ \sum_{e \in E}\nabla (F(X))_e (p_e)
	&+ \sum_{e \in E} \sum_{e' \in E} (p_e)(p_{e'}) \left[\nabla^2 F(X)\right]_{ee'} 
\end{align}
The minimizer in $p$ of $g(X;p)$ is the Newton vector $p^N = - \nabla F(X) (\nabla^2 F(X))^{-1} $. In cases where Hessian is not well defined or 
poorly scaled the Newton vector either must be modified or shouldn't be used at all.
There are two such cases: (i) the Hessian is poorly scaled if any set $A_l^i$ has $\norm{A_l^i}$ close to zero and $|A_l^i|>1$; and (ii) the Hessian is not well defined on the shared faces of multivistal cells for edges
in a support pair which changes from one side of the shared face to the other. 

In case (i), as $\norm{A_l^i}$ approaches zero
$\abs{\left[\nabla^2 F(X)\right]_{ee'}}$ increases without bound
for $e, e' \in A_l^i$. 
Thus, the Hessian matrix may become ill conditioned.
But even as $\norm{A_l^i}$ approaches zero, other entries of the gradient and Hessian are stable.
If $e \in A_l^i$, and $e' \notin A_l^i$, and $e, e' \in A_l^i$ for $l \neq q$, then
those entries $\abs{\left[\nabla^2 F(X)\right]_{ee'}}$ approaches zero. 

But even if case (i) occurs, the edges which have lengths approaching 
zero will have very little effect on the lengths of large scale edges.
The mixed second order partial derivatives, with one edge $e$ in $A_l^i$ with $\norm{A_l^i} \to 0$, and the other edge $e'$ in $A_l^i$ with $\norm{A_l^i}$ bounded positively from below, are essentially zero. Therefore, these mixed partial derivatives should have very little effect in the quadratic approximation $g(X;p)$.

\sean{
	
	In this case, two quadratic approximation could be used. 
	One approximation for edges which appear in sets $A_l^i$ with $\norm{A_l^i} \to 0$, and the another for edges which only appear in support sets $A_l^i$ with norm bound positively below. 
	Suppose that $S$ is the set of edges in supports such that $\norm{A_l^i} < \epsilon$ (some small positive constant). 
	If $S$ is not empty, the quadratic model is
	$g(X;p) = F(X)+\sum$
}

In case (ii), $\left[\nabla^2 F(X)\right]_{ee'})$ does not agree on the shared face of multi-vistal cells. In this case a direction of improvement is based on the gradient and the algebraic form of the Hessian matrix for points strictly within that shared face.
\sean{
	But how many forms of the Hessian can there be?
	Can the Newton direction be bad? How bad? Bad enough you wouldn't use it?
}
The Newton direction is obtained by solving for the minimizer of the second order Taylor series
quadratic model. The solution is obtained by solving the linear system $\nabla^2 F(X) p = -\nabla F(X)$. 
Since the entries of $\nabla^2 F(X)$ may jump across vistal cells the entries of this matrix will change
in a non-uniform way and the solution $p^N$ will jump. 

\sean{
	Note, an alternative to Newton's methods is a subspace minimization method. This is achieved by minimizing the quadratic model in the subspace spanned by the gradient and the Hessian, and bounded by a sphere or by an appropriately scaled ellipsoid when some edge lengths are approaching zero.}

\subsubsection{Choosing a step length}\label{sec:StepLength}
The taking a full step along the Newton direction
minimizes the quadratic approximation of the 
Fr\'{e}chet function, however taking a full step may result in a new point which actually has a larger
Fr\'{e}chet function value or that may
be on outside the orthant boundaries.

The first precaution is to calculate the maximum step
length $\alpha_0$ such that $|e|_{k+1}=|e|_k+\alpha_0 p_e \geq 0$ for all $e$.  
If $\alpha_0 \leq 1$ then let $\alpha = \alpha_0 c_0$
where $0<c_0 <1$.

Choosing
step-length which satisfies the following \emph{sufficient decrease condition}
will ensure  a substantial decrease
in the objective function value
at every step.
Let $0<c_1<1$.
\begin{align}\label{eq:SufficientDecrease}
	F(X^k+\alpha p) \leq F(X^k)+c_1 \alpha \nabla F_k' p
\end{align}
The \emph{curvature condition}, which rules out
unacceptably short steps, requires the step-length, $\alpha$, to satisfy
\begin{align}
	\nabla F(X^k+\alpha p)' p \geq c_2 \nabla F_k' p
\end{align}
for some constant $c_2$ in the interval $(c_1,1)$.

\subsubsection{Initialization}
For initializing an interior point search, any point in $\Or(E)$ would suffice, but it is preferential to start with a good guess for edge lengths. 
The global search algorithms presented in the next section could provide a starting point for a local search.
One good start
strategy
can be derived by noticing that the Fr\'{e}chet function can be separated into a quadratic part and a part involving sums of norms.
\begin{align}
	F(X)
	&  =\sum_{i=1}^n \sum_{l=1}^{k^i} \norm{A^i_l}^2+2\norm{A^i_l}\norm{B^i_l} +\norm{B^i_l}^2 + \sum_{e \in C^i} (|e|_X^2-|e|_{T^i})^2
\end{align}
The only terms that cannot be expressed in a quadratic function of the edge lengths are collected into function
\begin{align}
	S(X) 
	&= \sum_{i=1}^n \sum_{l=1}^{k^i} 2 \norm{A^i_l}\norm{B^i_l}
\end{align}
The quadratic parts are collected into function 
\begin{align}
	Q(X)
	& = \sum_{i=1}^n \sum_{l=1}^{k^i} \norm{A^i_l}^2+\norm{B^i_l}^2 +\sum_{e \in C^i} (|e|_X-|e|_{T^i})^2,
\end{align}
The minimizer of $Q(X)$, $X^*_Q$, can be easily found by solving $\nabla Q(X)=0$; the solution is
\begin{align}
	|e|_{X^*_Q}=\frac{\sum_{i=1}^n |e|_{T^i}}{n}.
\end{align}
The optimal value $|e|_{X^*_Q}$ is non-negative, and if $e$ is common in any of $T^1,\ldots,T^n$, then $|e|_{X^*_Q}$ is positive.
The gradient of $S(X)$ is non-negative at any feasible $X$, which implies that
the optimal edges lengths in $X^*$ must be no larger than the edge lengths in $X^*_Q$ i.e.\ the optimal solution is in the closed box
\begin{align}
	0\leq |e|_X & \leq |e|_{X^*_Q} & & \forall e \in E.
\end{align}
Thus a reasonable starting point would be $X^*_Q$.

\newpage
\section{Updating Geodesic Supports}\label{sec:updating}
In this section we consider the problem of updating the Fr\'{e}chet function while navigating an orthant of treespace.

The algebraic form of the Fr\'{e}chet function depends on the geodesics from the variable tree $X$ to the data trees $T^1,\ldots,T^n$. 
Updating these involves changing the algebraic form so it continues to represent
the individual geodesics as $X$ moves from one point of treespace to another. 
Finding the correct algebraic form from scratch takes order
$O(r^4n)$ complexity; in this section we offer more
efficient methods based on iteratively updating the geodesics. 
These methods take advantage of the polyhedral form of the vistal subdivision of squared treespace \cite{Miller}, 
which are based in turn on the optimality properties of a geodesic from \cite{OwenProvan}.
The next subsection describes the setup for systematically updating the algebraic form of the Fr\'{e}chet function.

\subsubsection{Setup and notation}
All discussion in this section takes place in squared treespace. 
Let $X^0$ and $X^1$ be fixed trees in the same orthant,
so that the geodesic between them is a line segment. 
Let $X^\lambda = (1-\lambda)X^0+\lambda X^1$ be a variable tree on this segment.
The length of each edge in  $X^\lambda$ is $|e|_{X^\lambda}= (1-\lambda)|e|_{X^0}+\lambda |e|_{X^1}$, and the change in the length of $e$ with respect to $\lambda$ is $d_e = |e|_{X^1}-|e|_{X^0}$. 
Thus  $|e|_{X^\lambda}= |e|_{X^0}+\lambda d_e$.
Let $\Gamma^{i,\lambda}$ be the geodesic from $X^\lambda$ to $T^i$ with supports $(\A^{i,\lambda},\B^{i,\lambda}) = (A^{i,\lambda}_1,B^{i,\lambda}_1),\ldots,(A^{i,\lambda}_{k^{i,\lambda}},B^{i,\lambda}_{k^{i,\lambda}}) $.
These supports will be constant in the vistal cell $\V^\lambda$ containing $X^\lambda$. We describe conditions under which $X^\lambda$ leaves $\V^\lambda$, and the associated updates to the supports. 

\subsubsection{Intersections with (P2) constraints}
The (P2) bounding inequalities for $\V^\lambda$ can be written in the form 
\begin{equation}\label{(P2)basic}
	\begin{array}{c c}
		\displaystyle \norm{B_{l+1}^{i,\lambda}}^2 \sum_{e \in A_l^{i,\lambda}}{|e|_{X^0}+\lambda d_e} \leq  \norm{B_l^{i,\lambda}}^2 \sum_{e \in A_{l+1}^{i,\lambda}}{|e|_{X^0}+\lambda d_e} & i = 1,\ldots, k^l-1, l = 1,\ldots, n.
	\end{array}
\end{equation}
Simplification yields
\begin{equation}
	a^i_l\lambda + b^i_l \geq 0
\end{equation}
where 
\begin{displaymath}
	a^i_l = \norm{B_{l}^{i,\lambda}}^2 \sum_{e \in A_{l+1}^{i,\lambda}}{d_e}-\norm{B_{l+1}^{i,\lambda}}^2 \sum_{e \in A_{l}^{i,\lambda}}{d_e},
\end{displaymath} and
\begin{displaymath}
	b^i_l = \norm{B_{l}^{i,\lambda}}^2 \sum_{e \in A_{l+1}^{i,\lambda}}{|e|_{X^0}}-\norm{B_{l+1}^{i,\lambda}}^2 \sum_{e \in A_{l}^{i,\lambda}}{|e|_{X^0}}.
\end{displaymath}
There are several cases for solutions, $\lambda^i_l = -b^i_l/a^i_l$, and each signifies a distinct situation.
The case $\lambda=0$ implies $X^0$ is on a (P2) boundary of $\V^\lambda$.
Any positive solution, $0<\lambda \leq 1$ corresponds to a point along the geodesic segment at which the segment intersects a (P2) constraint.  
Finding $\lambda > 1$ signifies an intersection with a (P2) boundary beyond $X^1$, and
if a solution is negative then there is an intersection with the geodesic in the opposite direction to $d$.
The first (P2) inequalities to be violated in moving along the geodesic segment are 
\begin{equation}
	\displaystyle \myargmin{i,l: \lambda_{i,l}>0}{ \left\{ \lambda^i_l\right\}.}
\end{equation}
If a (P2) constraint is satisfied at equality, the corresponding supports may be combined to make a new support satisfying (P2) at strict inequality, and still satisfying (P1) and (P3). Combining the current flow values into this new support pair will provide a warm start for subsequently tracking intersections with (P3) constraints, and further (P2) violations.

\subsubsection{Intersections with (P3) constraints}

From \cite[Prop. 3.3]{Miller} inequality constraints for (P3) are
\begin{equation}\label{(P3)basic}
	\begin{array}{l}
		\norm{B_l^{i,\lambda} \setminus J}^2 \sum_{e \in A_l^{i,\lambda} \setminus I} |e|_{X^0}+\lambda d_e - \norm{J} ^2\sum_{e \in I} |e|_{X^0}+\lambda d_e \geq 0 \\ \;for\;all\; i = 1,\ldots,k \; and\; subsets \;
		I\subset A_l^{i,\lambda},\;J\subset B_l^{i,\lambda} \; such \; that \; I \cup J \; is \; \\compatible.
	\end{array}
\end{equation}
Determining whether or not a support pair satisfies (P3) can be restated as the following extension problem.

\noindent{\bf Extension Problem}

\noindent{\bf Given:} Sets $A \subseteq E_X$, and $B \subseteq E_{T^i}$

\noindent{\bf Question:} Does there exist a partition $C_1  \cup C_2$ of $A$ and a partition of $D_1 \cup D_2$ of $B$ such that
\begin{itemize}
	\item[(i)] $C_2 \cup D_1$ corresponds to an independent set in $G(A,B)$,
	\item[(ii)] $\frac{ \norm{C_1}}{\norm{D_1}} \leq \frac{\norm{C_2}}{\norm{D_2}}$
\end{itemize}
The extension problem can be reformulated from a maximum independent set problem to a maximum flow problem. 
Each support pair $(A_l^{i,\lambda},B_l^{i,\lambda})$ has a corresponding incompatibility graph as defined in \cite[Sec. 3]{OwenProvan}.
The vertex weights of the incompatibility graph at a point along the geodesic segment can be parameterized in terms of $\lambda$ as
\begin{equation}
	w_e^\lambda = 
	\left \{ \begin{array}{c c}
		\frac{|e|_{X^0}+\lambda d_e}{\sum_{e' \in A_l^{i,\lambda}} |e'|_{X^0}+\lambda d_e'} & e \in A_l^i \\
		\\
		\frac{|e|_{T^i}}{\sum_{e'\in B^{i,\lambda}_l}{|e'|_{T^i}}} & e \in B_l^i 
	\end{array} \right .
\end{equation}
Although $w_e^\lambda$ is a non-linear function of $\lambda$,
matters are simplified by rescaling the lengths of edges to have sum 1 within each support pair separately. 
The approach is to complete the parametric analysis of
that extension problem and then map back 
to find $\lambda$ for the original weights.

Let $V^0$ and $V^1$ be formed
by rescaling the lengths of edges in $X^0$ and $X^1$ to sum to 1. 
Suppose that some (P3) constraint defined by $\sum b_e |e|_X \geq 0 $ is satisfied at equality by $\tilde{\lambda}$, that is $\sum b_e((1-\tilde{\lambda})|e|_{V^0}+\tilde{\lambda} |e|_{V^1})=0$.  The following gives a transformation between the solution when the weights in each support pair are scaled to have sum 1, and before scaling. 

\begin{align}
	0= &\sum b_e((1-\tilde{\lambda})|e|_{V^0}+\tilde{\lambda} |e|_{V^1})\\
	=&\sum b_e\left((1-\tilde{\lambda})\frac{|e|_{X^0}}{\sum_{e' \in A^i_l}{|e'|_{X^0}}}+\tilde{\lambda} \frac{|e|_{X^1}}{\sum_{e' \in A^i_l}{|e'|_{X^1}}}\right)\\
	=&\sum b_e\left(c_0 |e|_{X^0}+c_1 |e|_{X^1}\right)\\
	=&\sum b_e\left((1-\lambda)|e|_{X^0}+\lambda |e|_{X^1}\right)
\end{align}
Thus $\sum b_e((1-\lambda)|e|_{X^0}+\lambda |e|_{X^1})=0$ is
satisfied by 
\begin{align}
	\lambda = \frac{c_1}{c_0+c_1} = \frac{\tilde{\lambda}/\sum_{e \in A^i_l}{|e|_{X^1}}}{(1-\tilde{\lambda})/\sum_{e' \in A^i_l}{|e'|_{X^0}}+\tilde{\lambda}/\sum_{e' \in A^i_l}{|e'|_{X^1}}}.
\end{align}

Assuming the weights of edges in each support pair are already scaled to have 
sum 1, the weights are parameterized as a linear function as
\begin{equation}
	\tilde{w}_e^{\tilde{\lambda}} = 
	\left \{ \begin{array}{c c}
		|e|_{A^i_l}+\tilde{\lambda} \tilde{d}_e  & e \in A_l^i \\
		\\
		|e|_{B^i_l} & e \in B_l^i 
	\end{array} \right .
\end{equation}
where $\tilde{d}_e = \frac{|e|_{X^1}}{\sum_{e' \in A^i_l}{|e'|_{X^1}}}-\frac{|e|_{X^\lambda}}{\sum_{e' \in A^i_l}{|e'|_{X^0}}}$ is the change in capacity for the arc from the source to node $e \in A_l^{i,\lambda}$.

Parametric analysis of the extension problem will yield a method for 
updating the objective function along the segment from $X^0$ to $X^1$. 
In an incompatibility graph
the capacity for an arc from the source $s$ to a node $e$ in $A^i_l$ is
$c_{se} = w_e$, for arcs from a node $f$ in $B^i_l$ to the terminal node $t$
the capacity is $c_{ft}=w_f$ (fixed), and the capacity for an arc
from a node in $A^i_l$ to an incompatible node in $B^i_l$
is infinity. 
Assume an initial flow for the directed graph \incompGraphLambda\, is calculated for $\lambda = 0$.
For arc $(e,e')$, variable flow is $z_{ee'}$ and the flow for $\lambda = 0$ is $z_{ee'}^0$.
Residual capacity for arc $(e,e')$ is $r_{ee'}= c_{ee'}-z_{ee'}+z_{e'e}$.
Recall, that $\tilde{d}_e = |e|_{V^1}-|e|_{V^0}$ is the change in capacity for the arc from 
the source to node $e \in A_l^{i,\lambda}$ from $X^0$ to $X^\lambda$.
As $\lambda$ increases the flow may become infeasible because an arc capacity has decreased, or
an augmenting path may exist because the arc capacity has increased. 
The net change in the total arc capacity is zero because $\sum_e \tilde{d}_e = \sum_e (|e|_{V^1}-|e|_{V^0}) = 0$, and thus
the total increase in arc capacity must equal the total decrease in arc capacity.
Therefore the maximum flow will be inhibited not by a change in total capacity, but rather by a bottleneck preventing
a balance of flow as $\lambda$ increases. 
To balance the flow, excess flow from arcs with decreasing capacities must shift to arcs with increasing capacities.
The key is to identify directed cycles in the residual graph oriented along arcs with increasing capacities
and against arcs with decreasing capacities.

If $d_e>0$ then $e$ is a ``supply'' node and if $d_e < 0$ then $e$ is a ``demand" node. 
A (P3) constraint is violated precisely at the smallest positive $\lambda$ such that balancing supply and demand is not possible.
An augmenting path is a path in the residual network
from a supply node to a demand node. If there is an augmenting path
from each supply node to each demand node then it is possible to maintain
a feasible flow for some $\lambda >0$ by pushing flow along augmenting paths. 
In pushing flow along a set of augmenting paths, where $P_e$ is the
augmenting path for supply node $e$, the residual capacity along 
arc $(e',e'')$ is 
\begin{align}
	r_{e'e''}= r^0_{e'e''}+\lambda \left( \sum_{e:(e'',e')\in P_e} d_e-\sum_{e:(e',e'')\in P_e} d_e\right) .
\end{align}
For a set of augmenting paths an arc is a \emph{bottleneck} at $\lambda$ if it has
has zero residual capacity at $\lambda$.
Once a bottleneck is reached at least one augmenting path
is no longer valid.
Thus a given set of augmenting paths
cannot feasibly balance the total flow at
the smallest positive value $\lambda^*$ which has a bottleneck arc.

Each supply node with flow blocked by a bottleneck needs
a new augmenting path. If such a path cannot be found
then supply and demand cannot be balanced for $\lambda > \lambda^*$.
This signifies that $X^{\lambda^*}$ is on a (P3) boundary of its vistal cell. Thus support pair $(A^i_l,B^i_l)$
could be partitioned into a support pairs $(C_1,D_1)$ and $(C_2,D_2)$ (or even
into a sequence of support pairs as described in \cite[Lem. 3.23]{Miller}) to create
a valid support for the same geodesic from $X^\lambda$ to $T^i$.
If a supply node does not have any augmenting path, then increasing $\lambda$
will result in excess flow capacity which cannot be utilized to push flow from $s$ to $t$.

A minimum cover for $\lambda >\lambda^*$ can be constructed; and in what follows
``the minimum cover" refers to the one which is being constructed.
If $e$ does not have an augmenting path, then increasing $\lambda$ will cause
the residual capacity $r_{se}$ to become positive in a maximum flow.
Therefore $e$ cannot be part of the minimum cover. Consequently, to cover
the edges adjacent to $e$, every node adjacent
to $e$ in $B^i_l$ must be in the minimum cover. 
Supply nodes which have augmenting paths will be in the minimum cover
unless all of their adjacent arcs are adjacent to nodes in $B^i_l$ which are already
in the minimum cover. 
In summary the minimum cover is $C_1 \cup D_2$ where
$D_2$ is comprised of elements in $B^i_l$ which are adjacent to elements of $A^i_l$ which
do not have augmenting paths and 
$C_1$ is comprised of elements
in $A^i_l$ which do not have augmenting paths or which
have all their adjacent arcs covered by elements from $B^i_l$. 
The new support sequence is formed by replacing $(A^i_l, B^i_l)$ with 
$(C_1,D_1),(C_2,D_2)$ where $C_2 = A^i_l \setminus C_1$ and $D_1 = B^i_l \setminus D_2$.

Standard net flow techniques can be used to find a feasible flow from
supply nodes to demand nodes. If no feasible flow exists, then 
the cut from the Supply-Demand Theorem can be used
to construct a minimum cover for the solution to the (P3) extension problem.

\noindent {\bf (P3) Intersection Algorithm}\\
\noindent{\bfseries initialize} $\Lambda = 0$, $r_{ee'} = c_{ee'}$, $z_{ee'} = 0$\\
\noindent \hspace*{1 cm} find a maximum flow in $G(A^{\lambda,i}_l,B^{\lambda,i}_l)$\\
\noindent{\bfseries while} $\Lambda < 1$\\
\noindent \hspace*{1 cm} {\bfseries do} \\
\noindent \hspace*{1 cm} find a feasible flow in the residual network \\
\noindent \hspace*{1 cm} {\bfseries if} supplies and demands are infeasible \\
\noindent \hspace*{2 cm} halt, a (P3) boundary intersection \\
\noindent \hspace*{1 cm} {\bfseries endif}\\
\noindent \hspace*{1 cm} augment flow until some residual capacity reaches zero\\
\noindent \hspace*{1 cm} calculate smallest $\lambda^* > \Lambda$ with a bottleneck arc \\
\noindent \hspace*{1 cm} using the residual capacities:\\
\noindent \hspace*{1 cm} $r_{e'e''}^{\lambda^*}= r^\Lambda_{e'e''}+(\lambda^*-\Lambda) \left( \sum_{e:(e'',e')\in P_e} d_e-\sum_{e:(e',e'')\in P_e} d_e\right) $\\
\noindent \hspace*{1 cm} $\Lambda = \lambda^*$\\
\noindent{\bfseries endwhile}

\chapter{Tree data analysis with Fr\'{e}chet means}
\label{ch:AnalysisAngiographyData}
The focus of this chapter is application of Fr\'{e}chet means in tree-oriented data analysis of brain artery systems.

\section{Discussion}

As described in Ch. \ref{ch:intro}, a landmark based
shape correspondence of the cortical surface
is used in mapping angiography trees to points in BHV treespace.
In this chapter we apply the Fr\'{e}chet mean algorithms from Ch. \ref{ch:FMmethods}  to summarize the angiography dataset. 
In this section we refer to the brain
artery systems mapped to treespace via landmark embedding
as \emph{brain artery trees}.

A star tree is a phylogenetic tree which only has pendants.
The pendant lengths for minimizing the 
Fr\'{e}chet sum of squares are quite easy to obtain.
Each pendant is in every tree topology, and therefore the terms
in the Fr\'{e}chet function contributed by pendants are 
completely independent from other edges.
In fact, the terms involving a pendant are minimized by
arithmetic average length of that pendant. 
Ignoring the interior structure and optimizing pendant
lengths yields the \emph{optimal star tree}.

In earlier stages of research leading up to this thesis
it became clear that the Fr\'{e}chet mean of the 
brain artery trees was either the optimal star tree or very close to the optimal star
tree. 
This observation is based
on results from a study done by Megan Owen, a collaborator involved in the SAMSI Object Oriented Data Analysis program.
In her study she used Sturm's algorithm to approximate 
the Fr\'{e}chet mean of brain artery
trees.
Let $S_k$ be an estimate of the Fr\'{e}chet mean
after $k$ steps of Sturm's algorithm. 
In that study it was shown that in 50,000 steps of Sturm's Algorithm the Fr\'{e}chet sum of squares at $S_k$ always
exceeded the Fr\'{e}chet sum of squares of the optimal star tree.

Having a star tree or nearly star tree Fr\'{e}chet mean is consistent with two other 
observations about the distribution of artery trees.
The first observation is that the geodesic between any  pair of subjects sweeps down near the origin of the space. Take for example
the geodesic visualization in Fig. \ref{GeoVis}. 
Notice how the geodesic sweeps down drastically near the middle.
This pattern is ubiquitous among pairs in this data set.
The second observation is that the brain artery trees are all closer to the optimal star
tree than they are to any of each other. 

\sean{ In Figure \ref{GeoVis}, it is easy to see that there is little common structure between these two cortical correspondence trees.
	For the cortical correspondence the first and second subjects have 5 edges in common, 119 edges from the tree of the first subject are not compatible in the tree of the second subject, and 118 edges from the other subject are not compatible with the tree of the first. 
	Notice that subsets of edges from each subject are compatible. In the tree, half-way along the geodesic, blue edges from the tree of the first subject and red edges from the tree of the second subject are present together. 
	Figure \ref{GeoSummary} displays the number of edges and the total length of edges at each point along the geodesic between the cortical correspondence trees. 
	Although the number of edges drops along the cortical correspondence geodesic it is always at least 98. This indicates the size of compatible subsets of edges is non-negligible. The total length dips down to around 330 (mm) from a high of around 2500 (mm).
	A \emph{topological transition} occurs on a treespace path when some edges either leave and/or some edges enter the tree topology (see \cite{OwenProvan} for more information about the combinatorial structure of geodesics). 
	The geodesic between the trees in Figure \ref{GeoVis} has 16 topological transitions.
}
\begin{figure}[H]
	\begin{tabular}{l l l l l}
		\includegraphics*[width=0.3\textwidth]{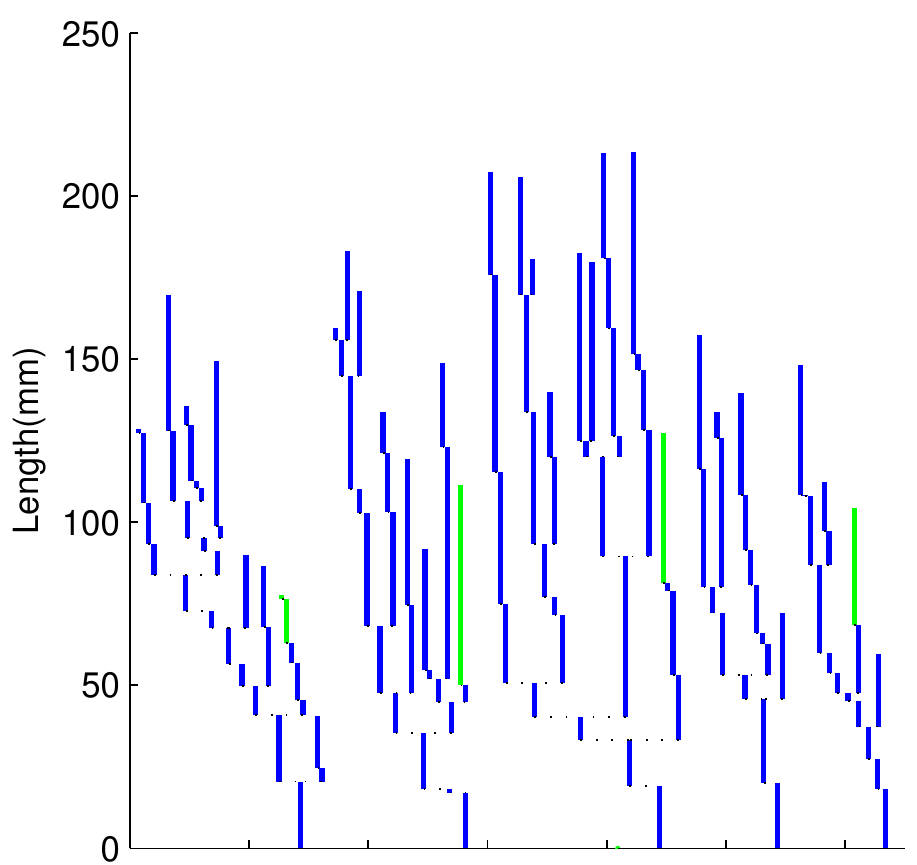} &
		\includegraphics*[width=0.3\textwidth]{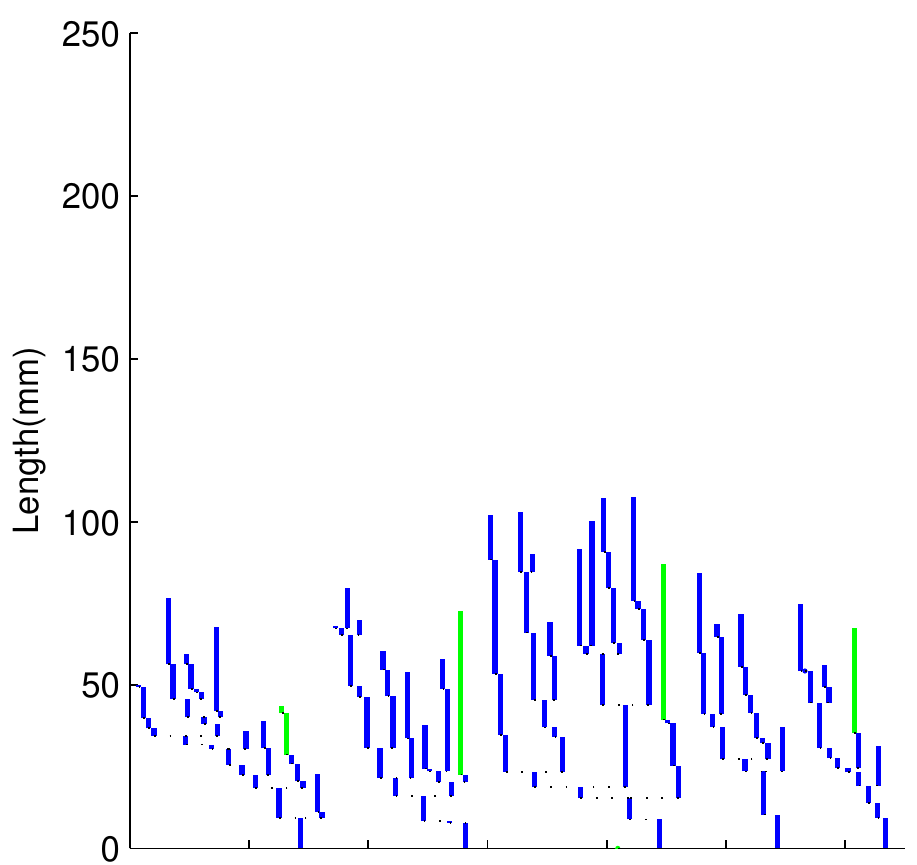} &
		\includegraphics*[width=0.3\textwidth]{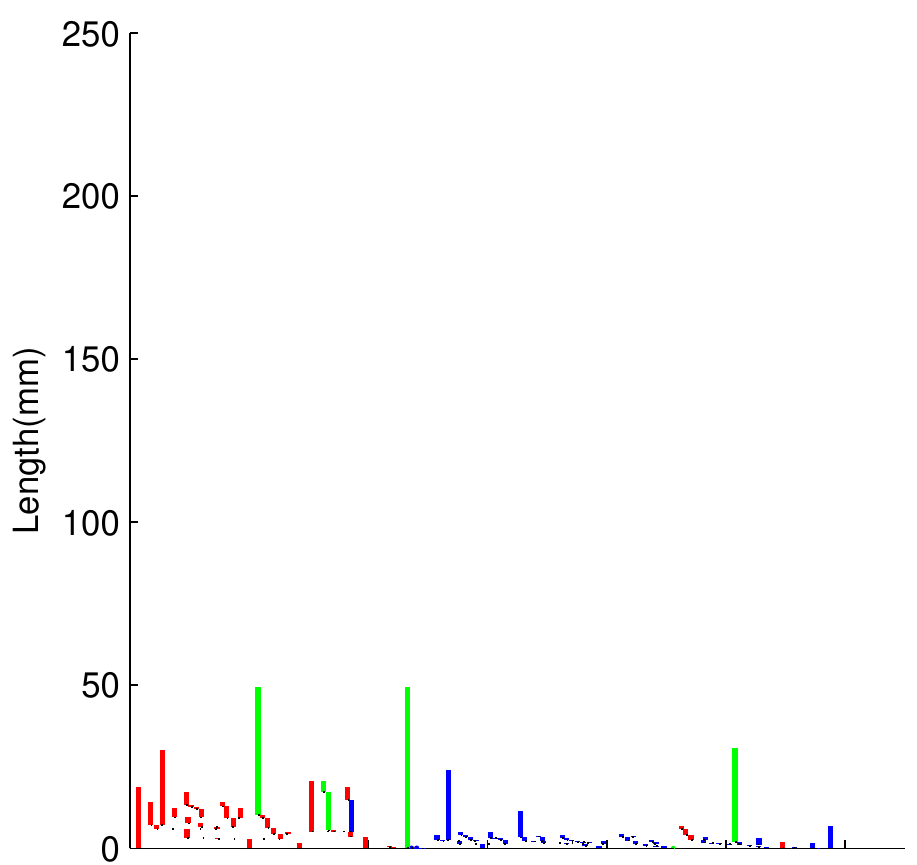} \\
		\includegraphics*[width=0.3\textwidth]{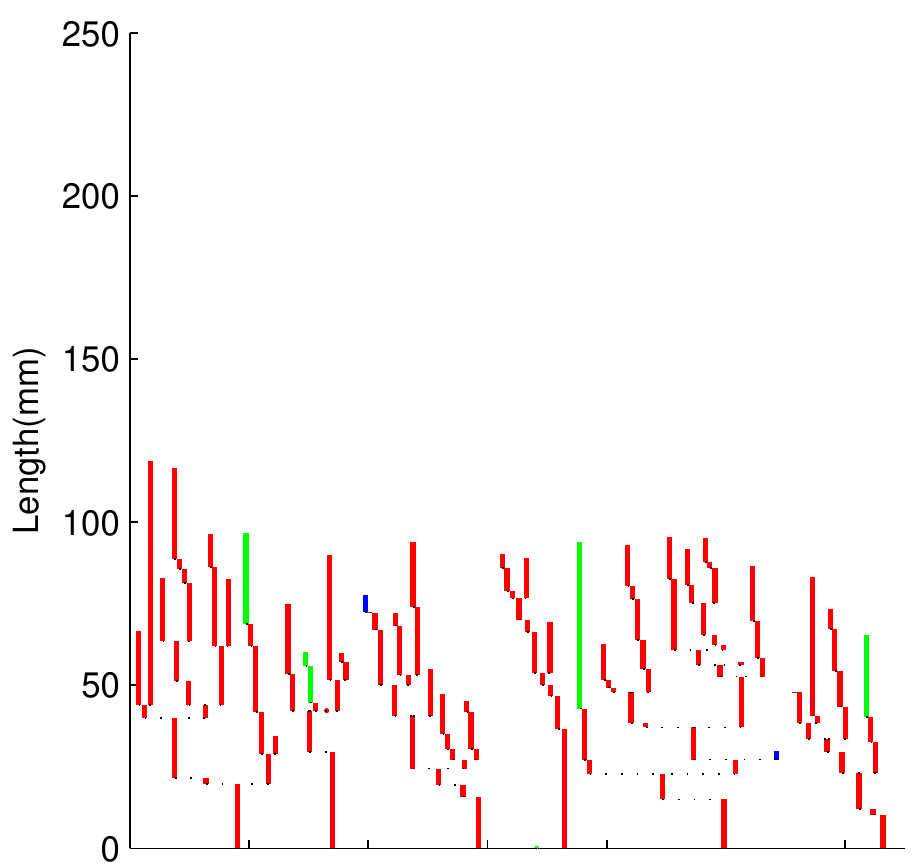} &
		\includegraphics*[width=0.3\textwidth]{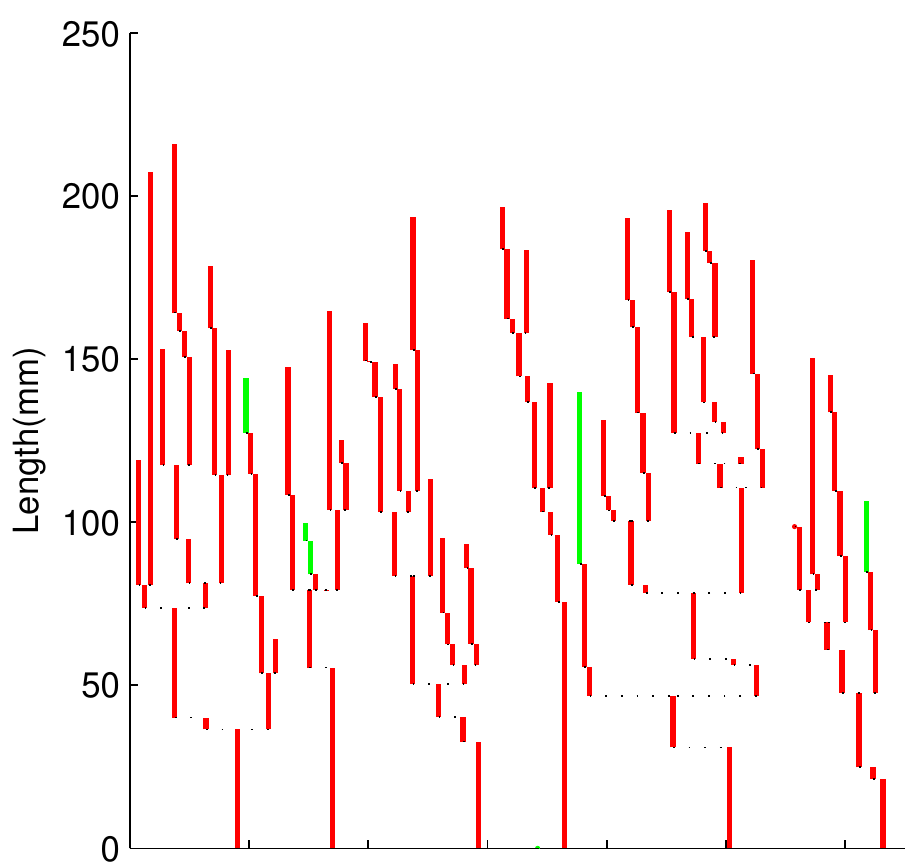} 
	\end{tabular}
	\caption{Example of a geodesic between artery trees showing five points along the geodesic. Blue edges are only in the tree of the first subject, red only in the other subject and green are common to both (pendants are not included in this visualization). The drastic dip along the geodesic path is characteristic for pairs of artery trees from this data set.}
	\label{GeoVis}
\end{figure}

The artery tree Fr\'{e}chet mean being degenerate is
not an anomalous case.
Theory about the behavior of Fr\'{e}chet sample mean will be
discussed in Ch. \ref{ch:Stickiness}, in which
the focus is formal study of the tendency for Fr\'{e}chet means
to be degenerate in treespace. 
In the next section we explore how the Fr\'{e}chet mean
varies over the range of treespace embeddings using 3 to all 128 landmarks.  

\section{Reduced landmark analysis}

Nested subsets of landmarks are used in a hierarchical approach
to studying the variability of vascular geometry with respect to
the cortical surface across this sample.
Here we study the effect of decreasing the number of landmarks on
the Fr\'{e}chet mean, the sample statistic which
describes the center of the distribution.


\begin{figure}[ht]
	\hspace{-3cm}\includegraphics{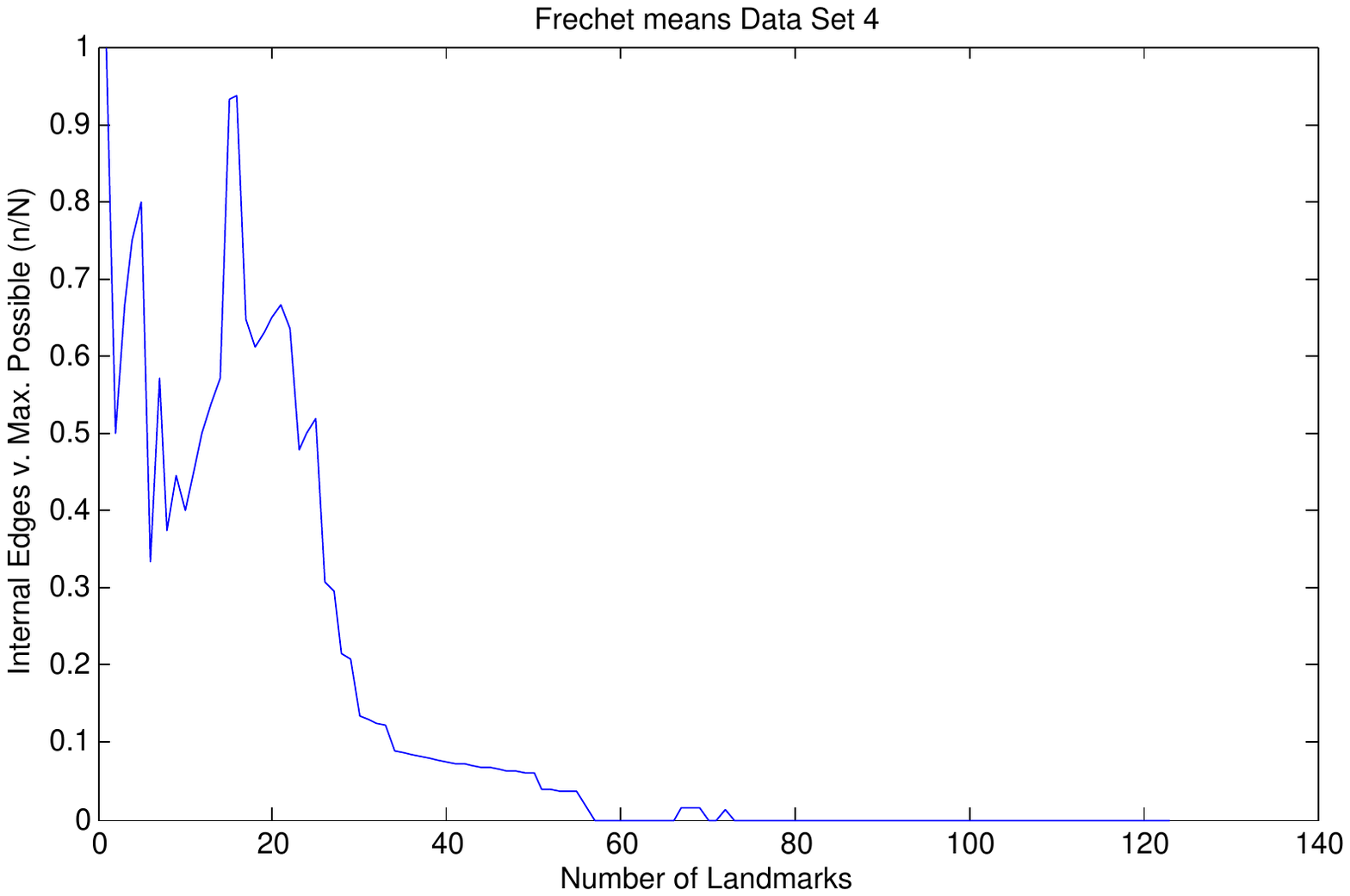}
	\caption{Ratio of number of interior edges versus maximal possible number of edges in Fr\'{e}chet means as a function of increasing the number of landmarks.
	}\label{LMRatio}
\end{figure}

The Fr\'{e}chet means for tree representations using 3 up to 128 landmarks
are calculated.
Fig. \ref{LMRatio} shows the ratio of the number of interior edges out of the maximum possible. 
Notice that in many cases the Fr\'{e}chet mean has a highly degenerate topology, with the number of interior edges dropping off steadily when around 24 landmarks are used. 

Rather than visualize the entire sequence of Fr\'{e}chet mean trees
we select a few interesting ranges for discussion.
The ratio for the number of interior edges out of the maximum possible, as shown in Fig. \ref{LMRatio}, jumps up and down in the range from 3 landmarks to 12 landmarks, and steadily climbs up until reaching 18 landmarks. 
The Fr\'{e}chet means for 6, 7, 8, and 9 landmarks are in Fig. \ref{FM6-7-8-9} (on the next page). 
There we can see the pattern of edge contractions and expansions.
Notice that the pair of landmarks 17 and 49 are stable in this range.
Views of the trees in the range from 15 to 18 landmarks are in Fig. \ref{FM15-16-17-18}, there we can see that although the ratio
of interior edges to maximal possible is steadily increasing
there are both edges contracting and expanding in this range as well,
though relatively fewer edges are contracting.

The small increase from zero interior up to 1 edge, at 69 landmarks is noteworthy. 
In fact the same interior edge is present in the Fr\'{e}chet mean using 69, 70, 71 and 74 landmarks. It is suprising that this split isn't in the Fr\'{e}chet mean for 72 or 73 landmarks. Moreover, this split has a length of about 2mm in all four of those cases. This deserves further investigation. It either suggests that the lengths
of edges which are incompatible with this split are fluctuating up and down in
the range 69 to 74, or that this split is present in the Fr\'{e}chet mean for 72 and 73 edges, but Sturm's algorithm did not capture this feature before the search terminated. This line of investigation will be pursued in future research.

\begin{figure}[ht]
	\subfigure[]{\includegraphics[width=0.45\textwidth]{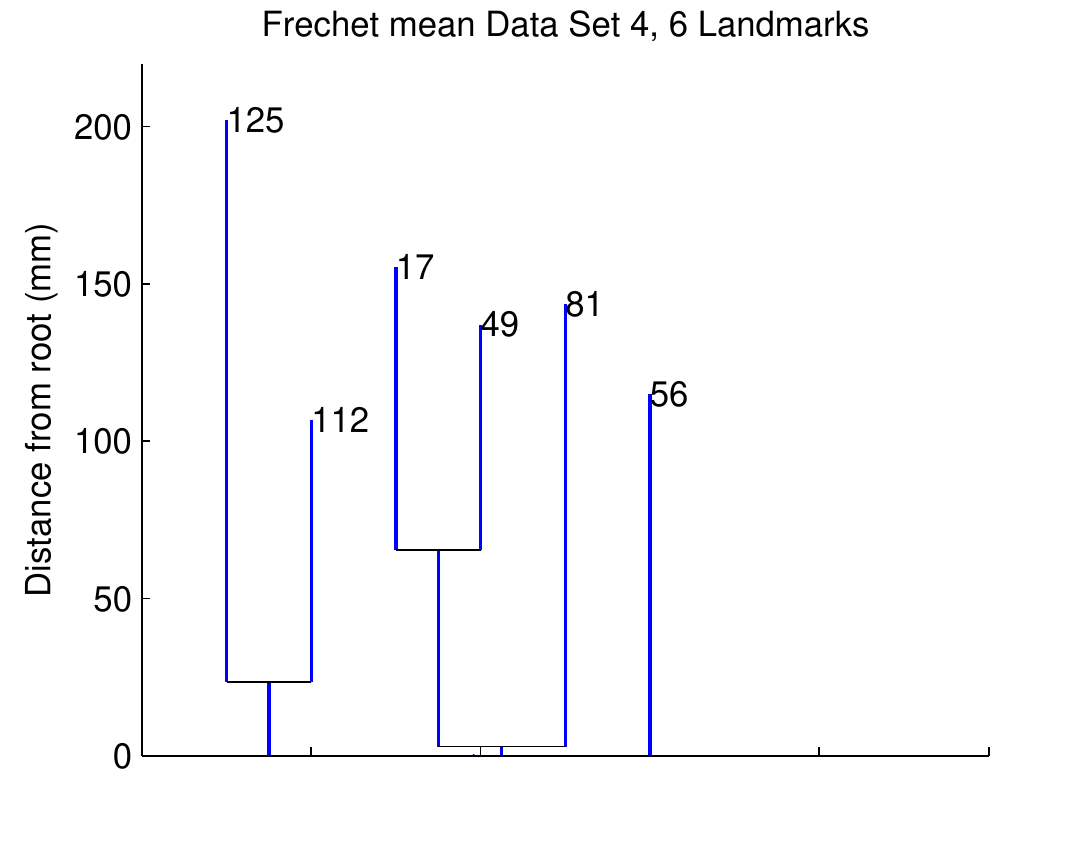}}
	\subfigure[]{\includegraphics[width=0.45\textwidth]{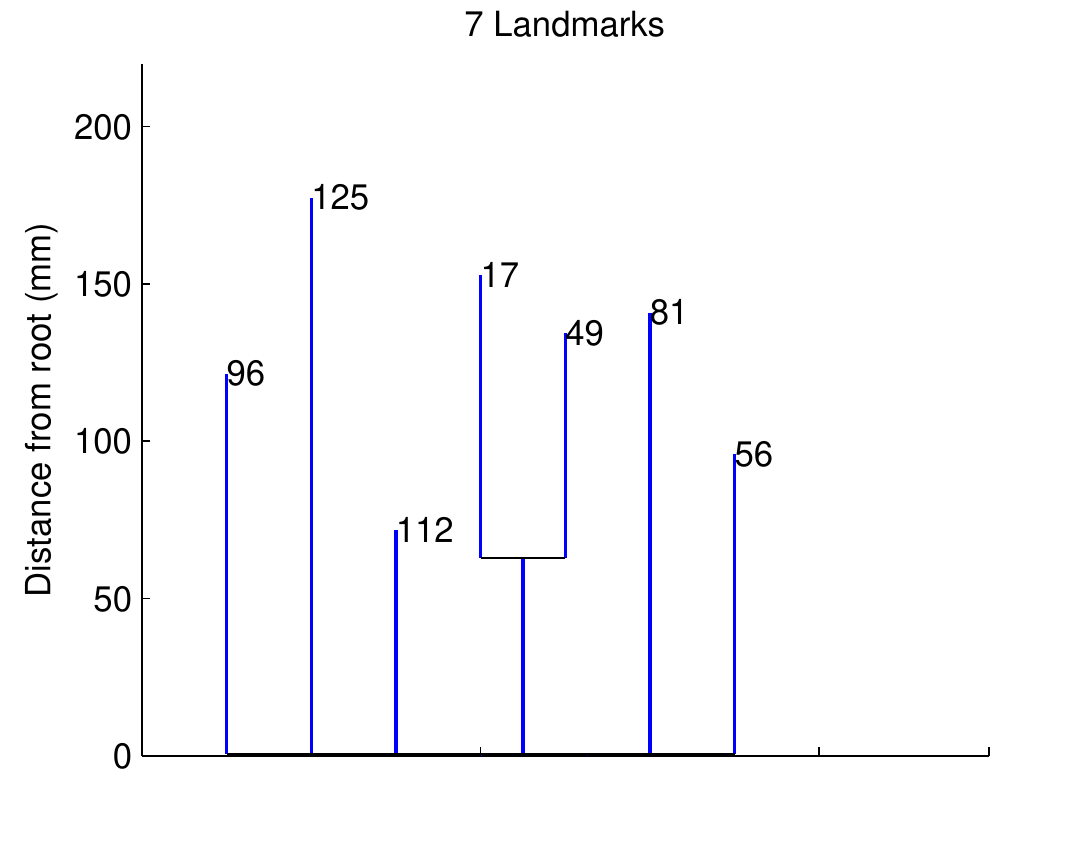}}
	\subfigure[]{\includegraphics[width=0.45\textwidth]{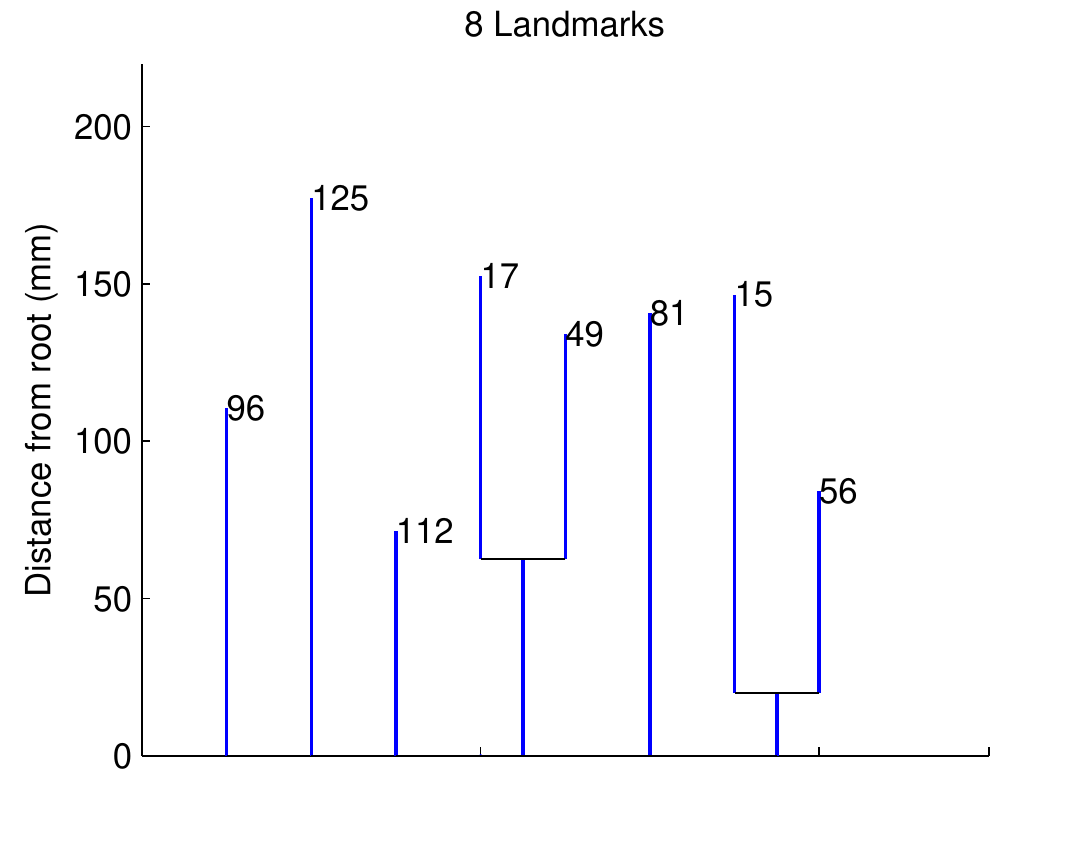}}
	\subfigure[]{\includegraphics[width=0.45\textwidth]{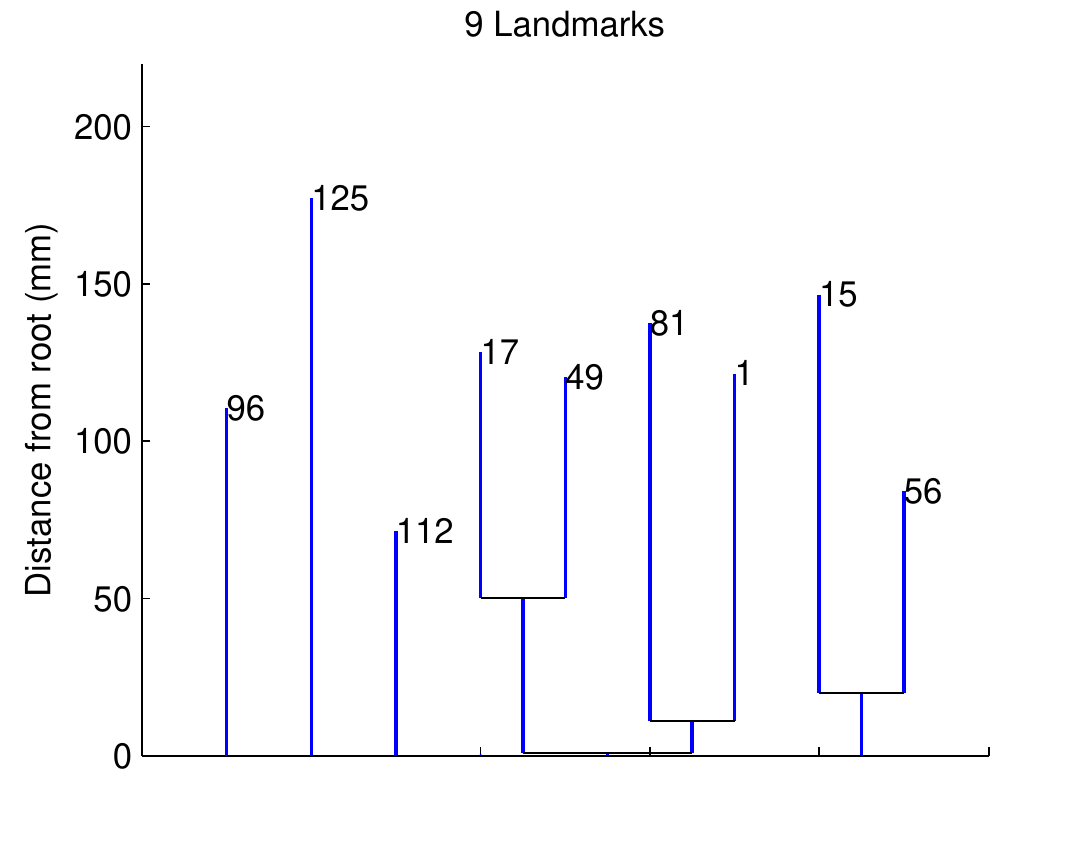}}
	\caption{Tree views for the Frechet means for data sets having 6, 7, 8, and 9 landmarks. }\label{FM6-7-8-9}
\end{figure}

\begin{figure}[ht]
	\subfigure[]{\includegraphics[width=0.45\textwidth]{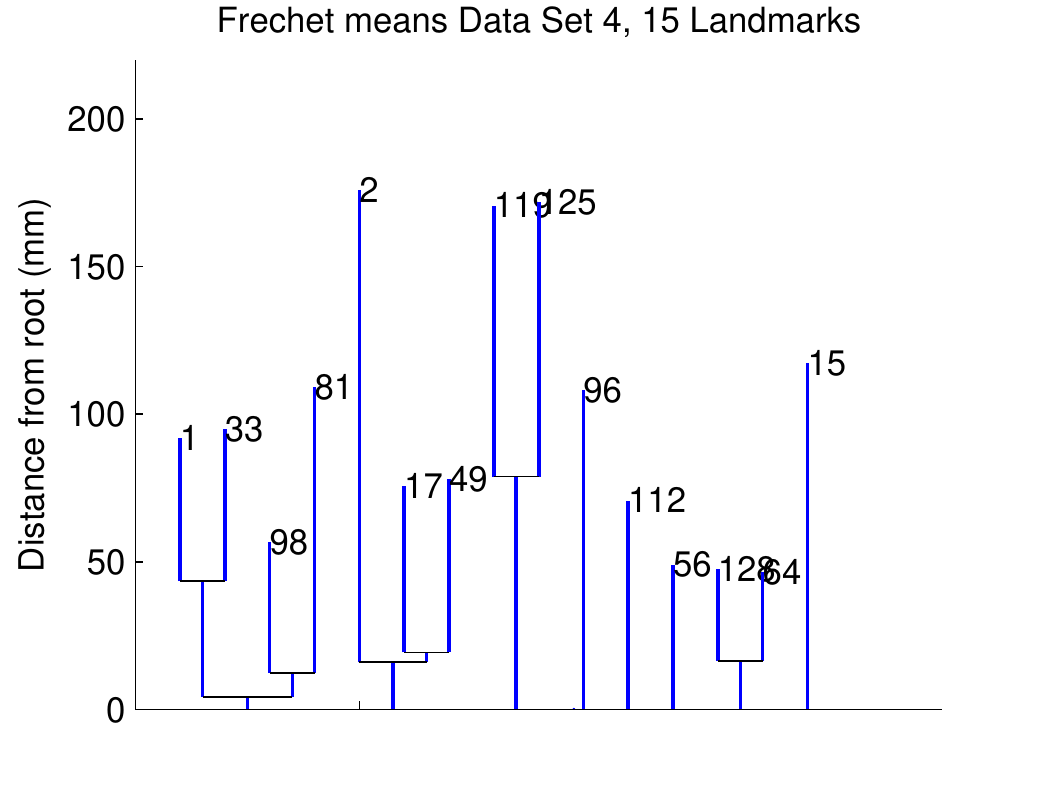}}
	\subfigure[]{\includegraphics[width=0.45\textwidth]{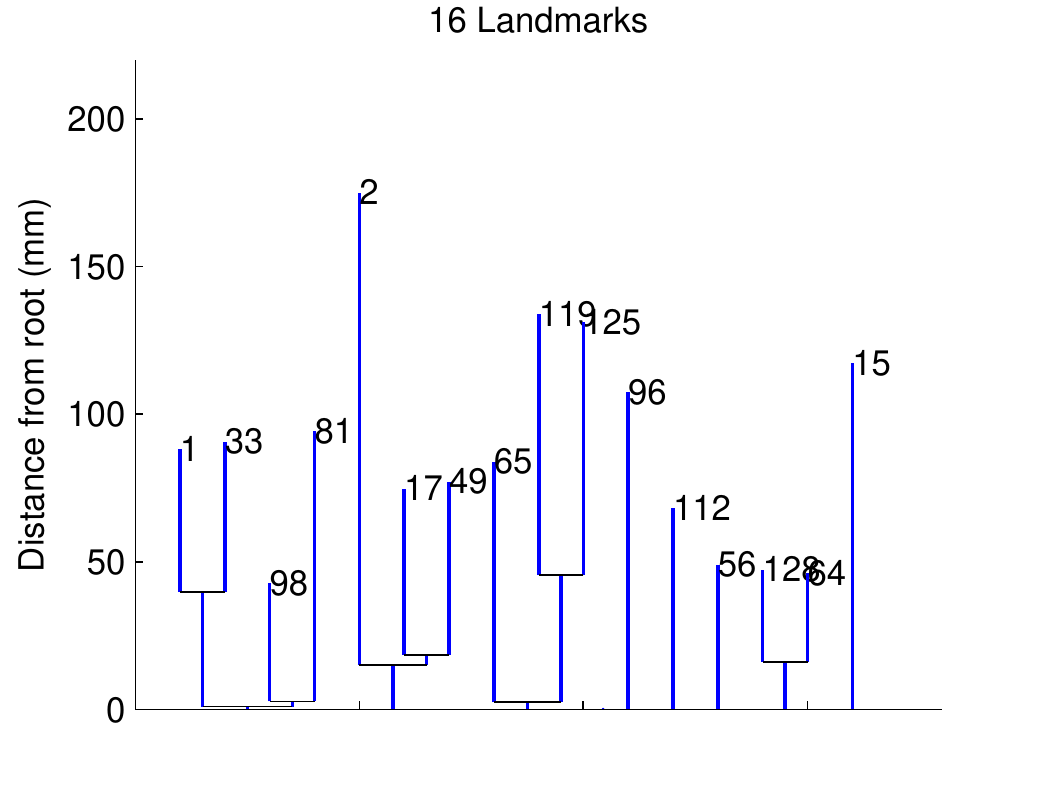}}
	\subfigure[]{\includegraphics[width=0.45\textwidth]{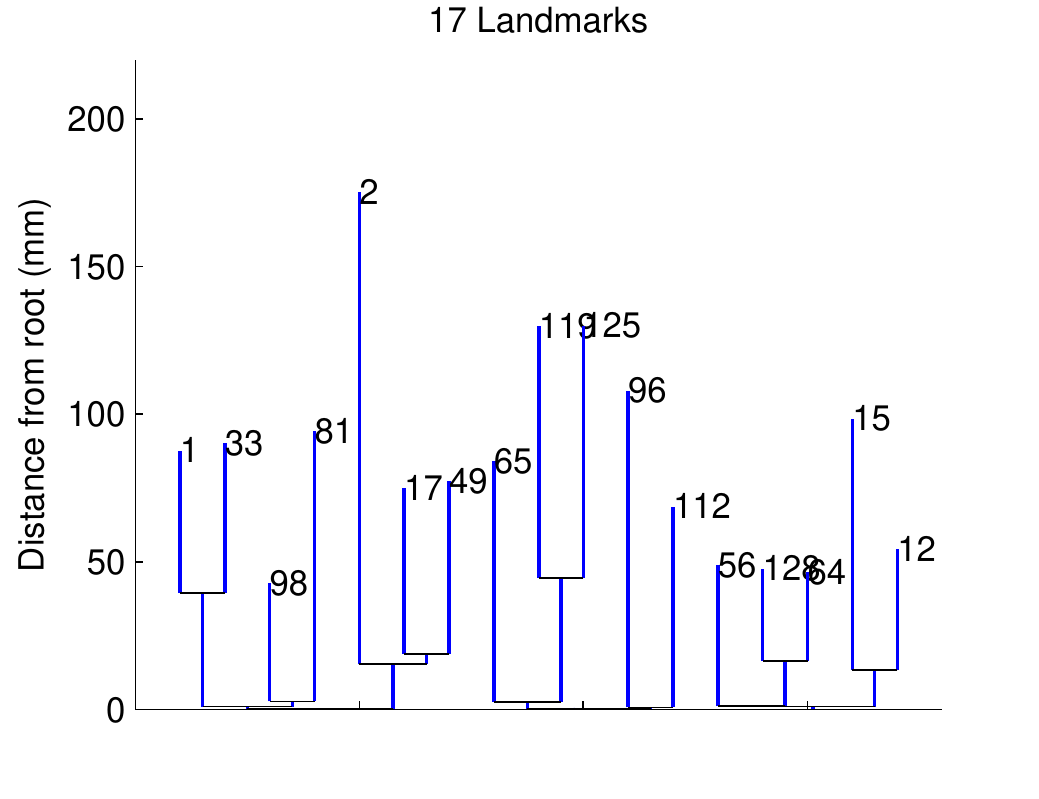}}
	\subfigure[]{\includegraphics[width=0.45\textwidth]{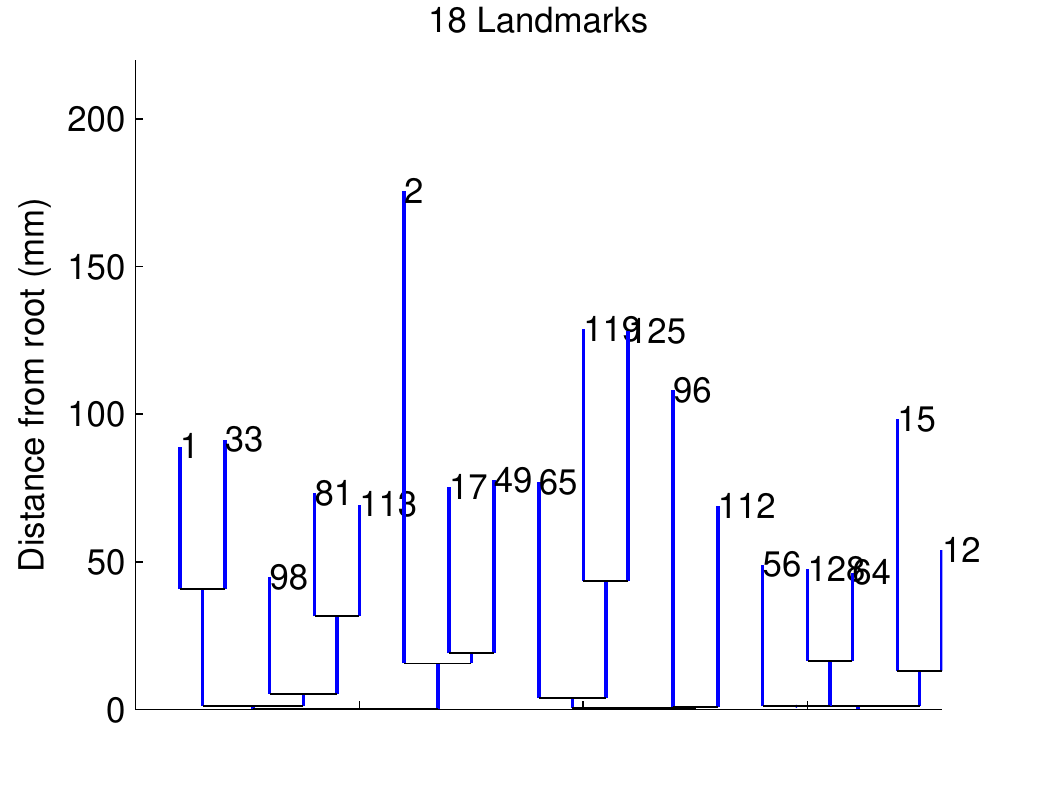}}
	\caption{Tree views for the Frechet means for data sets having 15, 16, 17, and 18 landmarks. }\label{FM15-16-17-18}
\end{figure}

The tendency of the Fr\'{e}chet mean to be degenerate can be attributed
to incompatible edges in among the data trees. 
Such incompatible edges represent that the artery system
infiltrates regions near landmarks in topologically variable ways.

In the next chapter we take a theoretical approach to understanding
the general behavior of Fr\'{e}chet means in treespace.

\sean{
	\begin{figure}[ht]
		\hspace{-3cm}\includegraphics{AttachLandmarks2.pdf}
		\caption{Cortex landmarks (red *) and attachments (red -), trimmed artery
			tree (blue), discarded arteries (cyan), landmarks labels (black $\mathbb{N}$).}
	\end{figure}}
	
	\sean{
		\section{Fr\'{e}chet means for gene trees simulated from Wright-Fisher models}
	}

\chapter{Fr\'{e}chet means and stickiness}
\label{ch:Stickiness}
\section{Introduction}
In the analysis of brain artery systems in Ch. \ref{ch:AnalysisAngiographyData}
we observed that the Fr\'{e}chet mean becomes increasingly
degenerate as the number of landmarks increases.
Stickiness describes a finer point, which is
that not only is the topology of the sample mean
tree degenerate, but that if the sample is large enough 
then the topology of the sample mean is stable at
the topology of the population mean when data points
are added or removed from the sample.

Denote by $F_n(X)$ and $\bar{T}_n$ the Fr\'{e}chet function and the Fr\'{e}chet mean for a sample of $n$ independent and identically distributed observations on BHV treespace, $T^1,T^2,\ldots, T^n$, as defined in Sec. \ref{sec:FMproblem}.
The focus of this chapter is the probability distribution of $\bar{T}_n$ that is the Fr\'{e}chet mean sampling distribution for $n$ observations.  
The main result Thm. \ref{thm:TreespaceLLN} is a \emph{sticky law of large numbers} for
sampling distributions of Fr\'{e}chet means on BHV treespace.

In treespace, depending on the population, the Fr\'{e}chet mean sampling distribution can and often will
exhibit the unusual property of stickiness.
Contrasting the typical behavior of sample means, in most contexts small changes
in the data result in small changes in the sample mean, but stickiness
refers to the phenomenon of a sample mean which does not respond 
to small changes in the data. 
More precisely stickiness refers to the tendency of a sampling distribution 
to be fully supported on a lower dimensional subset of the sample space than 
the population distribution.

\begin{defn}
	A probability measure on treespace, $\T_r$, is said to
	to have a \emph{sticky Fr\'{e}chet mean}, $\bar{T}$,
	if $\bar{T}$ is degenerate and for every point $T$ in $\T_r$ 
	there exists some $\alpha>0$ 
	such that the topology of $\bar{T}$ 
	is the same when the probability mass 
	at $T$ is increased by up to $\alpha$.
\end{defn}

\sean{
	\begin{defn}
		A probability measure on treespace, $\T_r$, is said to
		to have a \emph{half-sticky Fr\'{e}chet mean}, $\bar{T}$,
		if treespace can be 
		partitioned into two sets $U$ and $V$,
		such that for every finite point $T$ in $U$ 
		there exists some $0<\alpha <1$ 
		such that the topology of $\bar{T}$ 
		is the same when the probability mass 
		at $T$ is increased by up to $\alpha$,
		and for every finite point $T$ in $V$ 
		when the probability mass 
		at $T$ is increased at all,
		at least one
		edge switches from
		having length zero to having a positive length
		in $\bar{T}$.
	\end{defn}}
	
	Motivated to theoretically quantify
	stickiness,
	researchers have characterized the limiting distributions of sample means in BHV treespaces with 5 or fewer leaves \cite{Barden2013} and open books \cite{StickyCLT} with sticky central limit theorems.
	Sticky central limit theorems
	\begin{itemize}
		\item[1.] state the existence of a random finite sample size $N$ such
		that for sample sizes beyond $N$ the sample means stick with probability 1 (
		$N$ quantifies the intensity of stickiness); and
		\item[2.] characterize the limiting sampling distribution on its support.
	\end{itemize}
	
	The remaining contents of this chapter are as follows. 
	Sec. \ref{ch:StickinessReview} summarizes existing
	results for stickiness of sampling distributions on
	$\T_3$.
	Sec. \ref{sec:ProbMeasureTreespace} mathematically describes 
	probability measures on BHV treespace and further characterizes the notion of a sticky Fr\'{e}chet mean. 
	Sec. \ref{ch:StickyLLNTreespace} contains
	the main result of this chapter, a sticky law of large
	numbers for sampling distributions of Fr\'{e}chet means
	for population distributions spread over a negatively curved
	region of treespace.

	\section{Stickiness for Fr\'{e}chet means on $\T_3$}\label{ch:StickinessReview}
	The space  $\T_3$ of phylogenetic trees with index set $\{0,1,2,3\}$ with pendant lengths included, as depicted in Fig. \ref{T_3}, is equivalent to 
	an \emph{open half-book}. This space is composed of three copies, $L_1, L_2, L_3$ of $\mathbb{R}^5_{\geq 0}$, called \emph{pages}, pasted
	together on a four dimensional face, $L_0$, called the \emph{spine}. 
	For more details about the construction of phylogenetic treespaces see Sec. \ref{sec:TreeSpace}.
	Each page of the book corresponds
	to one of the three possible tree topologies, and the spine corresponds to the star tree topology.
	Let $L_k^+=L_k\setminus L_0$ be the portion of the $k^{th}$ page which is disjoint from the spine. Thus, $\T_3$ can be formed by a disjoint
	partition of the pages and spine, as $L_0 \cup L_1^+ \cup L_2^+ \cup L_3^+$.
	
	\begin{figure}[h!]
		\centering
		\includegraphics[width=0.95 \textwidth]{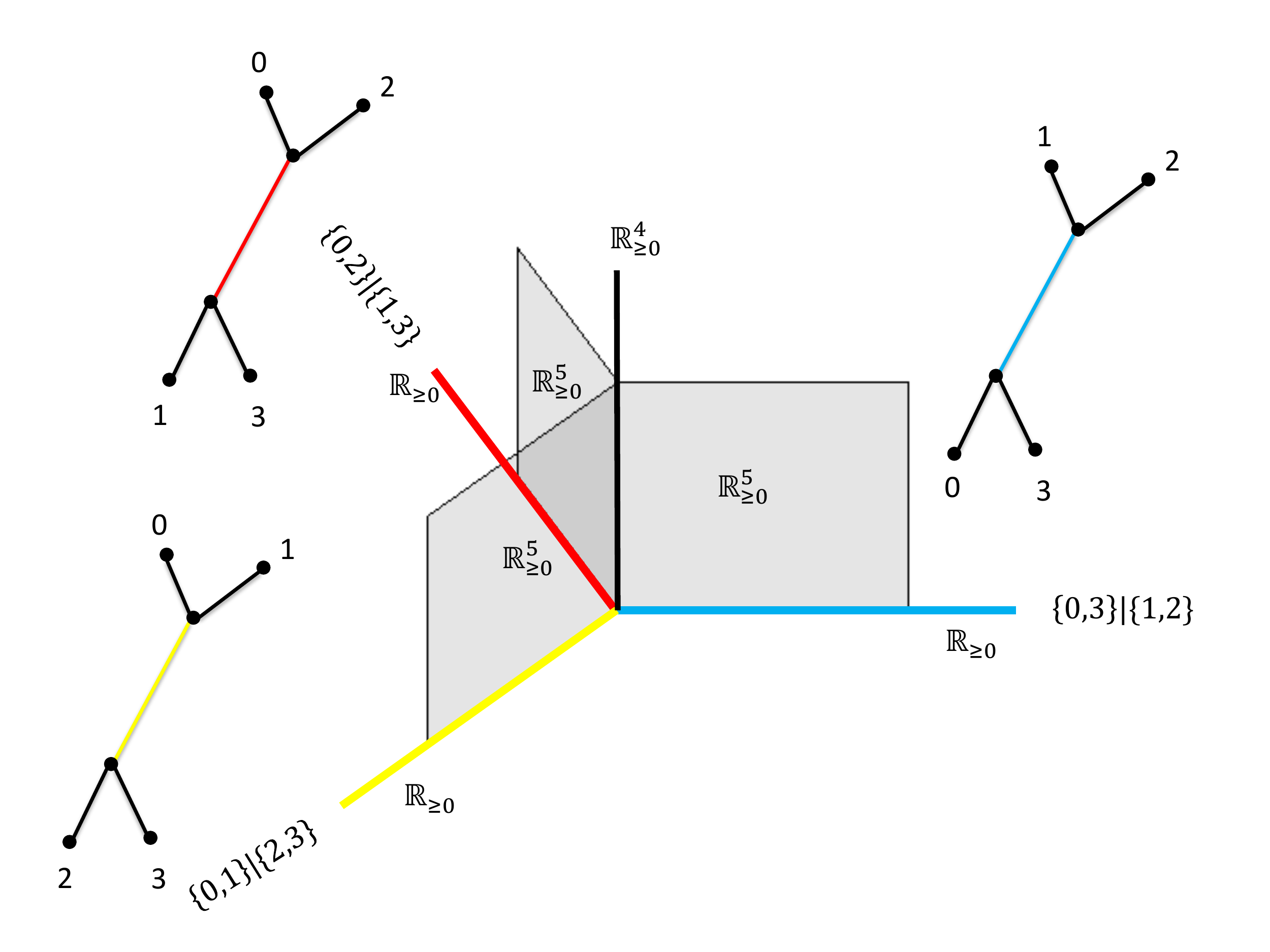}
		\caption{Depiction of $\T_3$: three copies of $\mathbb{R}^5_{\geq 0}$, $L_1,L_2,L_3$, each corresponding to one of three possible tree topologies pasted together on a copy of $\mathbb{R}^4_{\geq 0}$, denoted $L_0$. }
		\label{T_3}
	\end{figure}
	
	The sticky law of large numbers for open books \cite[Thm. 4.3]{StickyCLT}, is restated in
	this section for the special case of $\T_3$. Assume $T^1,...,T^n$ are independently sampled from a 
	probability measure, i.e.\ distribution, $\mu$ that is square-integrable and has spositive mass
	on each of $L_1^+, L_2^+, L_3^+$. 
	
	\begin{defn}
		The \emph{first moment} of $\mu$ on the $k^{th}$ page
		is the real number
		\begin{displaymath}
			m_j=\int_{L_j^+}xd\mu(x)-\sum_{i\neq j}\int_{L_k^+}xd\mu(x).
		\end{displaymath}
	\end{defn}
	
	\begin{thm}\cite[Thm. 2.9]{StickyCLT} Assuming $\mu$ is square integrable and has positive mass on each page, the moments $m_j$, $j=1,2,3$ are either\\
		\indent 1. (sticky) $m_j < 0$ for all $j\in\{1,2,3\}$,\\
		or there is exactly one index $k \in \{1,2,3\}$ such that $m_k \geq 0$, in which 
		case either \\
		\indent 2. (non-sticky) $m_k>0$, or\\
		\indent 3. (partly sticky) $m_k=0$.
	\end{thm}
	
	Let $\bar{T} = \argmin_{X\in\T_3} \int_{\T_3} d(X,T)^2d\mu(T)$.
	\begin{thm} ($\T_3$ Sticky LLN) \cite[Thm. 4.3]{StickyCLT}
		\begin{itemize}
			\item[1.] If the moment $m_j$ satisfies $m_j<0$, then there is a random integer $N$
			such that $\bar{T}_n \notin L_j^+$ for all $n \geq N$ with probability 1. Furthermore $\bar{T}\notin L_j^+$.
			\item[2.] If the moment $m_k$ satisfies $m_k>0$, then there is a random integer $N$
			such that $\bar{T}_n \in L^+_k$ for all $n \geq N$ with probability 1. Furthermore $\bar{T}\in L_k^+$.
			\item[3.] If the moment $m_k$ satisfies $m_k=0$, then there is a random integer $N$ such that $\bar{T}_n \in L_k$ for all $n \geq N$ with probability 1. Furthermore, $\bar{T} \in L_0$.
		\end{itemize}
	\end{thm}
	
	
	The focus in the proceeding sections is a Sticky LLN for $\T_3$.
	Analysis leading to a Sticky LLN for $\T_r$ in general is facilitated by
	a slight change of tack. 
	The analysis for open books uses the moments $m_k$ to
	precisely describe the conditions of stickiness.
	A comment, in \cite[Rem. 2.11]{StickyCLT} which states that $m_k$ is a directional derivative
	of the Fr\'{e}chet function from the spine into the $k^{th}$ page, 
	can be used to formulate an equivalent characterization of stickiness.
	Sec. \ref{sec:DiffAnalysis} gives a detailed analysis of directional derivatives
	from lower dimensional faces into higher dimensional faces of $\T_r$, which play 
	analogous roles to directional derivatives from the spine into pages of $\T_3$. 
	The next section describes probability measures on BHV treespace, and 
	characterizes levels of stickiness for the Fr\'{e}chet mean of a probability measure.
	
	\section{Fr\'{e}chet means for probability measures on treespace}\label{sec:ProbMeasureTreespace}
	Suppose that $\mu$ is a probability measure on BHV treespace, $\T_r$. Let $0$ denote the origin of $\T_r$. Assume that the distance between $0$ and $T\in \T_r$, $d(0,T)$ has bounded expectation under the measure $\mu$:
	\begin{align}\label{BoundedMeasure}
		\begin{displaystyle}
			\int_{\T_r} d(0,T) d \mu(T) < \infty.
		\end{displaystyle}
	\end{align}
	BHV treespace can be partitioned into disjoint open orthants, each corresponding to a unique phylogenetic tree topology.
	\sean{ A Borel set in a topological space is any set that can be created from
		intersections of open sets through the operations 
		of countable union, countable intersection, and relative complement.}
	Any probability measure $\mu$ on $\T_r$ decomposes uniquely as a weighted sum of probability measures
	$\mu_\Or$ on the open orthants of BHV treespace. 
	Therefore there are non-negative numbers $\{w_{\Or}\}$
	summing to 1 such that for any Borel set $A \subseteq \T_r$, the measure $\mu$ takes the value
	\begin{align}\label{MeasureDecomp}
		\mu(A)=\sum_{\Or} w_{\Or} \mu_{\Or}(A \cap \Or).
	\end{align}
	
	The Fr\'{e}chet function for a probability measure $\mu$ is 
	\begin{align}
		\begin{displaystyle}
			F_{\mu}(X) = \int_{\T_r}{d^2(X,T)}d\mu(T)
		\end{displaystyle}
	\end{align}
	The Fr\'{e}chet function is strictly convex and its minimizer is the Fr\'{e}chet mean, $\bar{T}$ \cite{Sturm}[Prop. 4.3].
	
	Properties for directional and partial derivatives for the Fr\'{e}chet function for a probability measure $\mu$
	follow a development similar to the
	properties for the Fr\'{e}chet function
	for a sample as in Ch. \ref{ch:FMmethods}.
	Thus we state their key properties now without detailed proofs.
	
	Let $T$ be a tree, let $X$ be a variable point, 
	Let $Y$ be a tree such that $\Or(X) \subseteq \Or(Y)$
	such that $X$ and $Y$ share a multi-vistal cell.
	Let $(A_1,B_1),\ldots,(A_{k^T},B_{k^T})$ be a support
	for the geodesic from $Y$ to $T$.
	Let $(A_1,B_1),\ldots,(A_{m^T},B_{m^T})$ be local support pairs, and let $(A_{m^T+1},B_{m^T+1}),\ldots,(A_{k^T},B_{k^T})$ be the rest of the geodesic support sequence being used to represent the geodesic between $Y$ and $T$.
	The restricted gradient of the Fr\'{e}chet function 
	in the minimal orthant containing $X$ is
	\begin{equation}\label{grad}
		[\nabla F_{\mu}(X)]_e= 2\int_{\T_r}\left\{ \begin{array}{ll}
			\left(1+ \frac{\norm{B_l}_T}{\norm{A_l}_X}\right)|e|_X & \textrm{if $e \in A_l$} \\
			\left(|e|_X-|e|_T \right) & \textrm{if $e \in C$}  
		\end{array} \right\} d\mu(T).
	\end{equation}
	if $|e|_X>0$ and $0$ otherwise.
	Let $p_e = |e|_Y-|e|_X$ for all $e \in E_Y$.
	The directional derivative of the Fr\'{e}chet function for probability measure $\mu$ from $X$ in the direction $Y$ is 
	\begin{align}\label{DirDer}
		F_{\mu}'(X,Y) =  \sum_{e \in E_X} p_e [\nabla F_\mu(X)]_e + \int_{\T_r} \left ( \sum_{l=1}^{m^T}{2\norm{B_l}_T \sqrt{ \sum_{e\in A_l}{p_e^2}}}+\sum_{e\in C^T \cap S}{-2p_e|e|_{T}} \right) d\mu(T).
	\end{align}
	
	The directional derivative decomposes the sum of
	the directional derivative along the component $Y_X$ of $Y$ 
	in $\Or(X)$ and the component $Y_\perp$ of $Y$ which 
	is perpendicular to $\Or(X)$ 
	\begin{align}
		F'_\mu(X,Y) = F'_\mu(X,Y_X)+F'_\mu(X,Y_\perp).
	\end{align}
	
	A point $X$ is the Fr\'{e}chet mean of $\mu$ if and only if it satisfies the Fr\'{e}chet optimality conditions
	\begin{align}
		[\nabla F_{\mu}(X)]_e=0 &\;\;\; \forall \; e \in E_X\\
		F_{\mu}'(X,Y) \geq 0 & \;\;\; \forall \;Y \textrm{ s.t. the component of $Y$ in $\Or(X)$ is 0}.
	\end{align}
	
	The Fr\'{e}chet mean, $\bar{T}$, for a measure $\mu$ is either 
	\begin{itemize}
		\item[1.] \emph{sticky} on the minimal open orthant containing $\bar{T}$ if $F_\mu'(\bar{T},Y)>0$ for all $Y$ which are perpendicular to $\Or(\bar{T})$ at $\bar{T}$,
		\item[2.] \emph{partially sticky} on the minimal open orthant containing $\bar{T}$ if $F_\mu'(\bar{T},Y)\geq 0$ for all $Y$ which are perpendicular to $\Or(\bar{T})$ at $\bar{T}$ and $F'(\bar{T},Y)=0$ for some $Y$ which is perpendicular to $\Or(\bar{T})$ at $\bar{T}$, or
		\item[3.] \emph{non-sticky} if $\bar{T}$ is non-degenerate.
	\end{itemize}
	\begin{lem}
		If $\mu$ is partially sticky, then the set of $Y$ s.t. $F_\mu'(\bar{T},Y)=0$ is a convex subset of treespace.
	\end{lem}
	\begin{proof}
		The directional derivative $F_\mu'(\bar{T},Y)$ is a convex function of $Y$ for $Y$ in a small neighborhood of $\bar{T}$. 
		The level set $\{Y | F_\mu'(\bar{T},Y)\leq 0\}$ is a convex set and there are no $Y$ such that
		$F_\mu'(\bar{T},Y)<0$.
	\end{proof}
	\section{LLN for sample means in BHV treespace}\label{ch:StickyLLNTreespace}
	This section starts out by stating a classical strong law of large numbers for $\T_r$ which is a special case
	of the strong law of large numbers introduced in \cite{Sturm} for metric spaces of non-positive curvature, of which treespace is a special case.
	The novel result in this section is that
	sample means in BHV treespace obey a sticky law of large numbers. 
	\begin{thm}\label{thm:TreespaceLLN} (Strong Law of Large Numbers). There is a unique point $\bar{T} \in \T_r$ such that 
		\begin{displaymath}
			\begin{displaystyle}
				\lim_{n \to \infty} \bar{T}_n = \bar{T}
			\end{displaystyle}
		\end{displaymath}
		holds with probability one. The point $\bar{T}$ is the Fr\'{e}chet mean of $\mu$.
	\end{thm}
	\begin{proof}
		This is a special case of \cite[Prop. 6.6]{Sturm} which gives 
		a Strong Law of Large Numbers for distributions on globally non-positively curved metric spaces. 
	\end{proof}
	The treespace Strong LLN (Thm. \ref{thm:TreespaceLLN}) guarantees, that for any $\epsilon > 0$, there exists a finite integer $N$ such that for all $n > N$, 
	$d(\bar{T}_n,\bar{T}) < \epsilon$. This guarantees
	that for large enough samples, every $e \in E_{\bar{T}}$ close to its length in $\bar{T}$,
	and
	that if $\bar{T}$ is degenerate then any edges in $\bar{T}_n$ which are not in $E_{\bar{T}}$
	must have lengths which are bounded above by $\epsilon$.
	The following Sticky LLN specifies that when $\bar{T}$ is degenerate
	then for large enough samples the Fr\'{e}chet sample mean, $\bar{T}_n$, will not contain
	any edges which are not in $\bar{T}$. 
	\begin{thm}(Sticky Law of Large Numbers)
		\begin{itemize}
			\item[1.] Sticky case: If $F'(\bar{T},Y)>0$ for all $Y$ which are perpendicular to $\Or(\bar{T})$ at $\bar{T}$ then there exists a random
			integer $N$ such that for all $n>N$, $F'_n(X,Y) >0$ for all $Y$ which are perpendicular to $\Or(\bar{T})$ and $E_{\bar{T}_n} = E_{\bar{T}}$.
			\item[2.] Partially sticky case: If $F'(\bar{T},Y)\geq 0$ for all $Y$ in the normal space of $\bar{T}$ and $F'(\bar{T},Y)=0$ for some $Y$ in the normal space of $\bar{T}$ then there exists a random integer $N$ 
			such that for all $n>N$, if $F'(\bar{T};Y)>0$ then $F'(X,Y)>0$, and if $F'(\bar{T};Y)=0$ then
			the sample mean $\bar{T}_n$ has a positive probability of being in $\Or(Y)$.
			\item[3.] \emph{Non-sticky case}: If $\bar{T}$ is non-degenerate then there exists a random integer $N$ such that for all $n>N$ $E_{\bar{T}_n}=E_{\bar{T}}$.
		\end{itemize}
	\end{thm}
	\begin{proof}
		The treespace Strong Law of Large Numbers (Thm. \ref{thm:TreespaceLLN} guarantees the existence of a finite integer $N$
		such that for all $n > N$, $|e|_{\bar{T}_n} > 0$ for all $e \in E_{\bar{T}}$. 
		
		Next we show 
		that there must exist some finite integer $N$ such 
		that for all $n>N$, $|e|_{\bar{T}_n}=0$ for 
		all $e$ such that $|e|_{\bar{T}}=0$.
		Essentially the proof is based on the Fr\'{e}chet optimality conditions.
		Let $X$ be a point in $\Or(\bar{T})$. 
		For every $Y$ such that the component of $Y$ in $\Or(\bar{T})$ is 
		zero the values of the directional derivatives $F'_\mu(X,Y)$ and $F'_n(X,Y)$ are independent of the precise edge lengths of $X$, provide that every edge in $\Or(\bar{T})$ has a positive length.
		We show the directional derivative 
		$F'_n(X,Y)$ for $Y$ perpendicular to $\Or(\bar{T})$ 
		approaches $F'_\mu(X,Y)$ with probability 1.
		for any $X \in \Or(\bar{T})$.
		This will suffice to prove the claim
		since we have already established that $E_{\bar{T}} \subseteq E_{\bar{T}_n}$ 
		using the Strong Law of Large Numbers.
		
		Consider a point in the minimal open orthant containing the Fr\'{e}chet mean, $X\in \Or(\bar{T})$, 
		and a point $Y$ such that the component of $Y$ in $\Or(\bar{T})$ is zero.
		Since the component of $Y$ in $\Or(\bar{T})$ is zero
		\begin{displaymath}
			F'_n(X,Y) = \frac{1}{n}\sum_{i=1}^n \left (\sum_{l=1}^{m^i}{2\norm{B^i_l} \sqrt{ \sum_{e\in A^i_l}{p_e^2}}}+\sum_{e\in C^i \cap S}{-2p_e|e|_{T^i}}  \right)
		\end{displaymath}
		Treespace $\T_r$ can be subdivided into a disjoint set of the open interiors
		of vistal cells such that for each vistal cell $\V$, the geodesic from $Y$ to any $T$ in $\V$ can be represented with the same combinatorial form.
		Thus the directional derivative from $X$ to $Y$ can be separated
		into a sum of terms each contributed by a vistal cell $\V$.
		Let $1_{T \in \textrm{int}(\V)}$ be an indicator function with value 1 if
		$T \in \textrm{int}(\V)$ and 0 otherwise.
		\begin{align}
			F'_n(X,Y) = \sum_{\V} \left( \sum_{i=1}^n \left (1_{T^i \in \textrm{int}(\V)} \frac{1}{n}\sum_{l=1}^{m^i}\left({2\norm{B^i_l} \sqrt{ \sum_{e\in A^i_l}{p_e^2}}}\right)+\sum_{e\in C^i \cap S}{-2p_e|e|_{T^i}}\right) \right)
		\end{align}
		For a fixed $\V \subseteq \Or$, each $T^i \in \V$  is equivalent to a random vector in $\Or$; thus by the classical Law of Large Numbers \cite[Sec. 7.2.5]{Shapiro2009}
		\begin{align}
			\lim_{n\to \infty} \frac{1}{n}\sum_{i=1}^n \left (1_{T^i \in \textrm{int}(\V)} \sum_{l=1}^{m^i}\left({2\norm{B^i_l} \sqrt{ \sum_{e\in A^i_l}{p_e^2}}}\right)+\sum_{e\in C^i \cap S}{-2p_e|e|_{T^i}}\right) \\= 
			\mu(\textrm{int}(\V))\int_{\textrm{int}(\V)} \left ( \sum_{l=1}^{m^T}{2\norm{B_l}_T \sqrt{ \sum_{e\in A_l}{p_e^2}}}+\sum_{e\in C^T \cap S}{-2p_e|e|_{T}} \right) d\mu_{\textrm{int}(\V)}(T)
		\end{align}
		with probability 1. Hence, by summing over all vistal cells,
		\begin{align}
			\lim_{n\to \infty} F'_n(X,Y) \to F'_{\mu}(X,Y)
		\end{align}
		with probability 1. 
	\end{proof}
	%
	
	\subsection{Extensions}
	
	In treespace and the open book, where stickiness has been observed and studied theoretically, the population is distributed on a negatively curved
	region of the sampling space. It is conjectured that 
	positively curved geometry leads to ``antistickiness" and negatively curved geometry leads to ``stickiness" \cite{StickyCLT}.
	Although the topic of antistickiness is not explored any further here, it may
	be a topic of future research. The 2-dimensional sphere would be an ideal
	model space to study antistickiness phenomenon for sample means
	on positively curved metric spaces. There may also be theoretical connections
	between robust estimation, sparse estimation, and stickiness, although investigating
	these connections is outside of the scope of this thesis.
	
	It is plausible that sticky central limit theorems will generalize
	to any dimension of treespace
	because they all share the property of global non-positive curvature.
	The sticky LLN here is a precursor for a central limit theorem because
	it characterizes the limiting distribution as sticky i.e.\ characterizes
	the support of the limiting distribution. 
	In this sense this result
	further justifies the conjecture that sticky central limit
	theorems are generalizable to any dimension of treespace. 
	A sticky central limit theorem for samples on BHV treespace
	would be an interesting topic for further research.

\chapter{Treespace kernel smoothing}
\label{ch:TreespaceKernelSmoothing}
\section{Introduction}
One way that methods for statistical analysis of populations of complex objects e.g. tree-structured data, functional data, and manifold data,  can contribute to biological sciences, such as neuroscience and developmental biology, is through the use of flexible models for studying trends in parts of anatomy. 
The motivation of this chapter is to study the variation in branching and geometry of 
3D anatomical trees as a function of age. In particular, in this chapter, Fr\'{e}chet
kernel smoothing of brain artery trees over age is introduced for
analysis of the CASILab cross-sectional dataset of human brain artery systems 
described in Ch \ref{ch:intro}. 

Kernel smoothing of tree-structured data objects was first studied in \cite{WangY}. 
The findings in that paper serve as a motivation to continue studying
non-parametric regression for tree-structured data. 
In that paper they found an increase in common structure among 20 to 30 year old subjects, and a decrease in common structure 
among 30 to 50 year old subjects. 
The method in that paper is based on connectivity only, and the goal
of this chapter is to obtain deeper insights by including additional
structural information in the tree-data object representation. 

The approach in this chapter, and that previous approach, both
use the framework of kernel smoothing (see \cite{Fan1996, Hardle1991, Wand1994} for books on kernel smoothing),  but 
differ in their representations of artery systems,
and thus in their choice of metric.
The previous work represents artery systems as binary trees, which are purely topological objects.
It uses the integer tree metric, which counts
the total number of non-common nodes between two trees \cite{WangH}.
In this chapter artery systems are represented as points in BHV treespace.
This representation captures the topological structure of the entire tree and lengths
of individual artery segments.


The organization of the rest of this chapter is as follows.
In Sec. \ref{sec:FKSmethod}, a novel method for kernel smoothing of data points
which have tree-structured response, and predictive values in a Euclidan space is presented.
In Sec. \ref{sec:FKSAngiographyDataAnalysis}
we present a case study of the sample
of human brain artery systems of normal adults collected at the UNC Chapel Hill CASILab discussed in Sec. \ref{sec:MapBrainArteryData}.
Sec. \ref{sec:MinLengthRepSeqAlg} is an appendix for descriptions of algorithms used in this research.

\section{Methodology}\label{sec:FKSmethod}

\subsection{Treespace Kernel Smoothing}
Consider a dataset of bivariate observations $(x_1,T^1),\ldots,(x_n,T^n)$, where $x_1,\ldots,x_n$ are explanatory
observations in $\mathbb{R}$ and $T^1,\ldots, T^n$ are dependent observations in a BHV treespace, $\T$.
Kernel smoothing can be used to study the relationship between
explanatory and dependent variables in a flexible way.

A weighting for a dataset is an assignment of a weight, a positive real number $w_i$, to each observation $i=1,\ldots, n$, such that $\sum w_i = 1$.
In the context of smoothing, a kernel $K$ is a function used to create weightings. A kernel takes as input two points $x$ and $x'$ in the space of explanatory variables. The operations of the kernel are (1) calculate the distance between $x$ and $x'$, (2) scale their distance by division by the \emph{band width} parameter $h$, and (3) pass the scaled distance into a non-negative, symmetric window function $D$. In summary a kernel function is
\begin{displaymath}
	K(x,x'; h) = D\left(\frac{\norm{x-x'}}{h}\right)
\end{displaymath}
For kernel $K$ and bandwidth $h$, at a given point $x$, the \emph{kernel-weight function} assigns a weight to observation $i$
\begin{displaymath}
	w_i(x) = \frac{K(x,x_i; h)}{\displaystyle\sum_i K(x,x_i;h)}
\end{displaymath}
Here, a Gaussian density function is used for the kernel.
The \emph{weighted Fr\'{e}chet 
	function} $F: {\mathcal T} \to \mathbb{R}_{\geq 0}$, which generalizes the 
Fr\'{e}chet function introduced in Ch. \ref{ch:FMmethods}, is a sum of
weighted mean squared distances
\begin{displaymath}\label{wFrF}
	F(T;w) = \sum_{i=1}^{n}{w_i d(T,T^i)^2}.
\end{displaymath}
Following a development parallel to methods for
Fr\'{e}chet means,
the unique minimizer of the weighted Fr\'{e}chet function is called the \emph{weighted Fr\'{e}chet sample mean} 
\begin{displaymath}\label{wFrM}
	\bar{T}(w) = \argmin_{T \in {\mathcal T}}{F(T;w)}.
\end{displaymath}

\begin{figure}[h!]
	\centering
	\includegraphics[width=0.90 \textwidth]{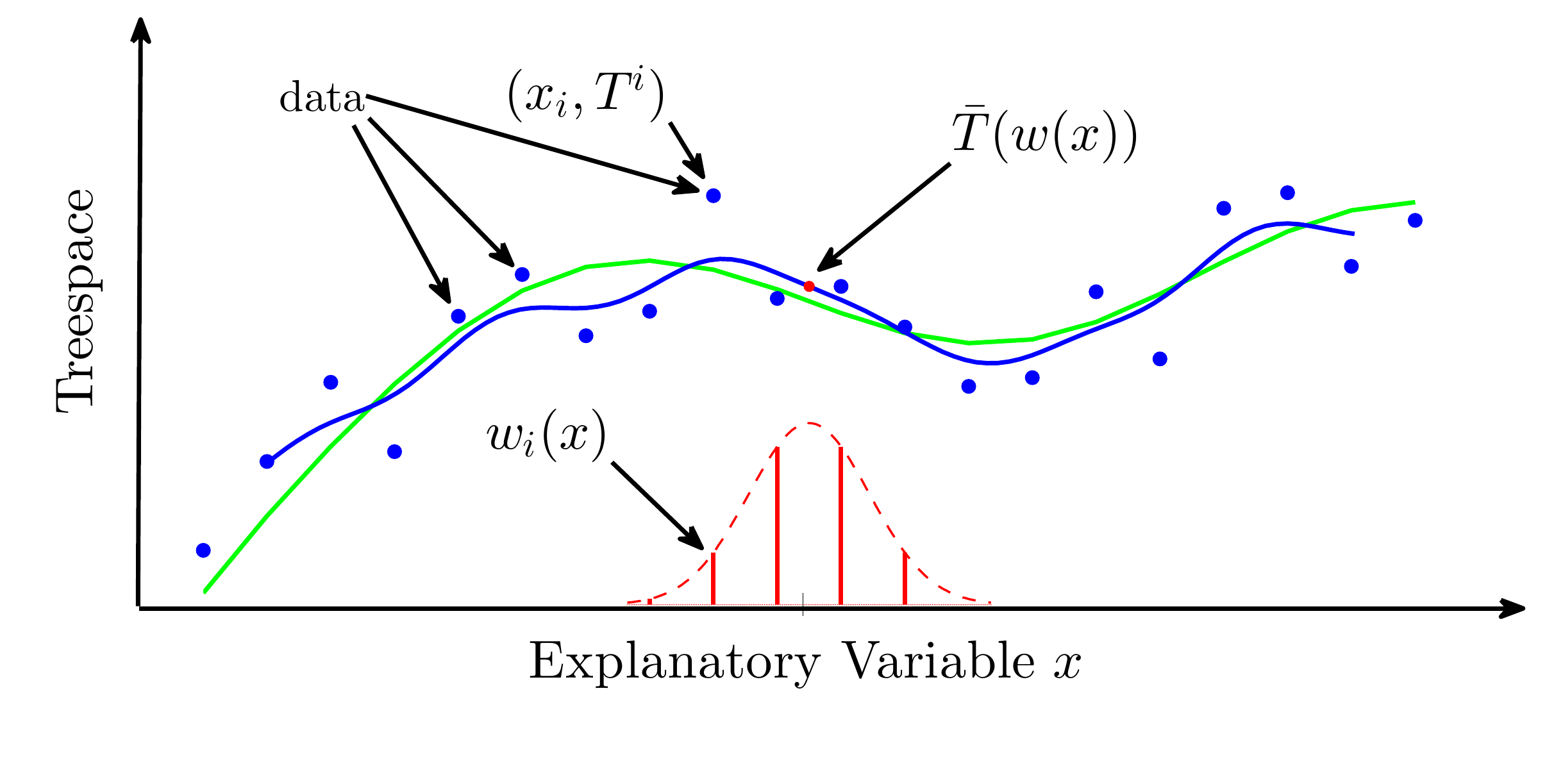}
	\caption{Schematic of Fr\'{e}chet kernel smoothing. Data are blue dots, a noisy sample from the green curve; the red dashed curve is a Gaussian kernel function and the vertical bars represent the relative weights of data points; red dot is kernel weighted mean at $x$, $\bar{T}(w(x))$; and the blue curve is the kernel smooth. }
	\label{fig:KernelSmoothExample}
\end{figure}

Fig. \ref{fig:KernelSmoothExample} shows a schematic of Fr\'{e}chet kernel smoothing. There we can see how an estimate for the response value at the point $x$ is created using a local Fr\'{e}chet mean. 

The results of kernel smoothing are dependent on the bandwidth $h$. 
When $h$ is very small, the weights are concentrated on data points within a narrow range, and the smoothed values are very localized. An extreme case, when $h=0$, results in local estimates which are either empty or just the data points themselves.
At the other extreme, when the bandwidth is very large, all data points have very similar weights, and the smoothed tree for any value $x$ is very similar to the overall Fr\'{e}chet mean of the dataset. 

The bandwidth parameter, $h$, in classic nonparametric regression,
is selected to optimize some expected error, or its empirical (e.g. cross-validation) or
asymptotic estimate. See \cite{Fan1996, Hardle1991, Wand1994} for detailed descriptions of such techniques. 
A scale space approach to the bandwidth selection problem,
which compares smooths across a range
of window sizes to look for
consistent patterns as recommended in \cite{Chaudhuri1999}, 
is used here.

Adapting the optimization algorithms 
from Ch. \ref{ch:FMmethods} to solve for a weighted Fr\'{e}chet mean $\bar{T}(w)$
is straightforward.
The step length in each iterative step of a cyclic split proximal point algorithm can be modified
to account for the weight, $w_i$, on each term in the Fr\'{e}chet sum of squares.
For a random split proximal point algorithm (a.k.a Sturm's algorithm), the probability of sampling a data point, $T^i$, at each iteration is equal to its weight, $w_i$. 
For descent methods, the only changes
come from 
appropriately multiplying terms in the Fr\'{e}chet function 
and its derivatives by the weights, $w_1,\ldots, w_n$.
\sean{However, kernel smoothing involves finding a set
	of related of Fr\'{e}chet means, thus
	more efficient methods which
	re-optimize the Fr\'{e}chet function
	after small changes in the weights
	could be of use in large scale problems.
	Algorithms for Fr\'{e}chet kernel smoothing optimization 
	problems are described in Sec. \ref{ch:SmoothingAlgorithms}.}

\subsection{Methods for summarizing treespace smooths}
Summarizing treespace smooths is a critical step in studying the relationship
between predictive and response variables.
Key summaries of each tree, such as the total number of interior edges,
and the sum of interior edge lengths are just the first step.
Comparing these summaries 
across smoothing bandwidths can provide insights
into the data.
But these comparisons are somewhat limited
in that they don't capture topological variation
across the smooth. 
Tree topologies are discussed in the introduction to phylogenetic trees in Sec. \ref{sec:PhylogeneticTrees}.

\subsubsection{Minimum Length Representative Sequences}\label{sec:MinLengthRepSeq}
Summarizing length and number of edges does not
capture the variation in tree topologies in a tree-smooth.
While counting the number of edges quantifies degeneracy,
it does not consider the variety of splits in tree topologies
across the smooth. Thus the aim here is to summarize
the topological variability in an ordered set of trees.

For a set of trees $T^1,\ldots,T^n$ a \emph{representative tree topology}, $R$,
is a tree topology such that every element, $T^i$, has a subset of the edges in the topology $R$.
A representative tree topology may not exist for an arbitrary set of trees. 
For $T^1,\ldots,T^n$, a \emph{representative set of tree topologies}, $R^1,\ldots,R^k$, has the property that every
element $T^i$ has a subset of the edges in topology $R^j$ for some $j$.
For an ordered set of trees $T^1,\ldots,T^n$ a \emph{representative sequence} is a 
sequence of tree topologies $R^1,\ldots,R^k$ such that $R^1$ is a representative tree topology for 
the first $i_1$ trees, $T^1,\ldots,T^{i_1}$, $R^2$ is a representative tree topology for the next $i_2$ trees,
$T^{i_1+1},\ldots,T^{i_2}$, and so on.
A \emph{minimum length representative sequence} is a representative sequence with
the fewest number of representative tree topologies possible. 
An algorithm for finding a minimum length representative sequence is in
the last section of this chapter, Sec. \ref{sec:MinLengthRepSeqAlg}.

In data analysis, a minimum length representative sequence
for the predicted response trees in a tree-smooth 
will be used to capture the variety of tree topologies over the range of the predictive variable. 

\section{Case study of CASILab angiography dataset}\label{sec:FKSAngiographyDataAnalysis}
%

In Ch. \ref{ch:AnalysisAngiographyData} we studied how the ratio of interior edges
to the maximum possible number of interior edges varied over reduced sets of landmarks. When this ratio is close to 1 it indicates a relatively higher degree of common structure.  In this analysis, we will work with 17 landmarks because of the relatively high degree of common structure.

This data set contains 85 subjects with ages approximately uniformly distributed in the range 20 to 80 years.
Because the density of observations affects the results a closer look at the distribution of ages will make for a more informed analysis. As we can see
in Fig. \ref{KDEAge}, there are only 2 subjects over 70. While we will use all the available data in the analysis, due to the sparsity of data over 70, we will not attempt to make any estimates in that range.
\begin{figure}[h]
	\hspace{0cm}\includegraphics{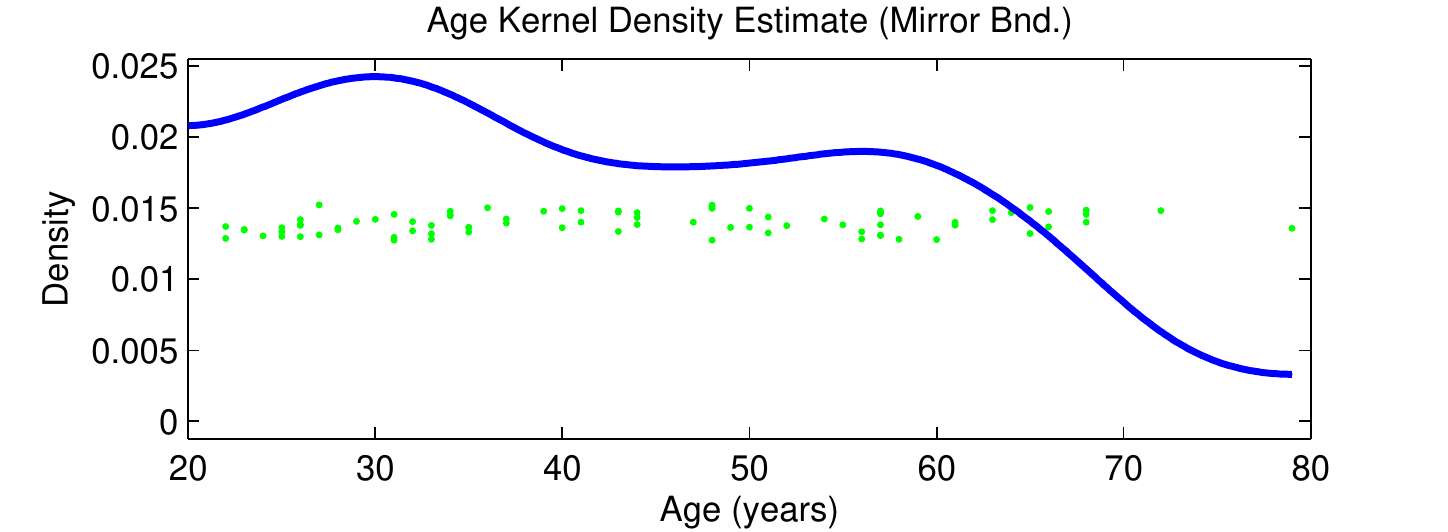}
	\caption{Kernel density estimate for distribution of subject ages.}\label{KDEAge}
\end{figure}

A Gaussian probability density is used for the kernel window function.
The effect of scale is studied across six smoothing levels: $h=1, 2, 3, 4, 5, 6$ years.
In this data analysis, the kernel estimator will be evaluated at the ages from data points.
These point estimates are adequate to study the relationship between artery trees and ages.

\begin{figure}[t]
	\subfigure[]{\includegraphics[width=0.6\textwidth]{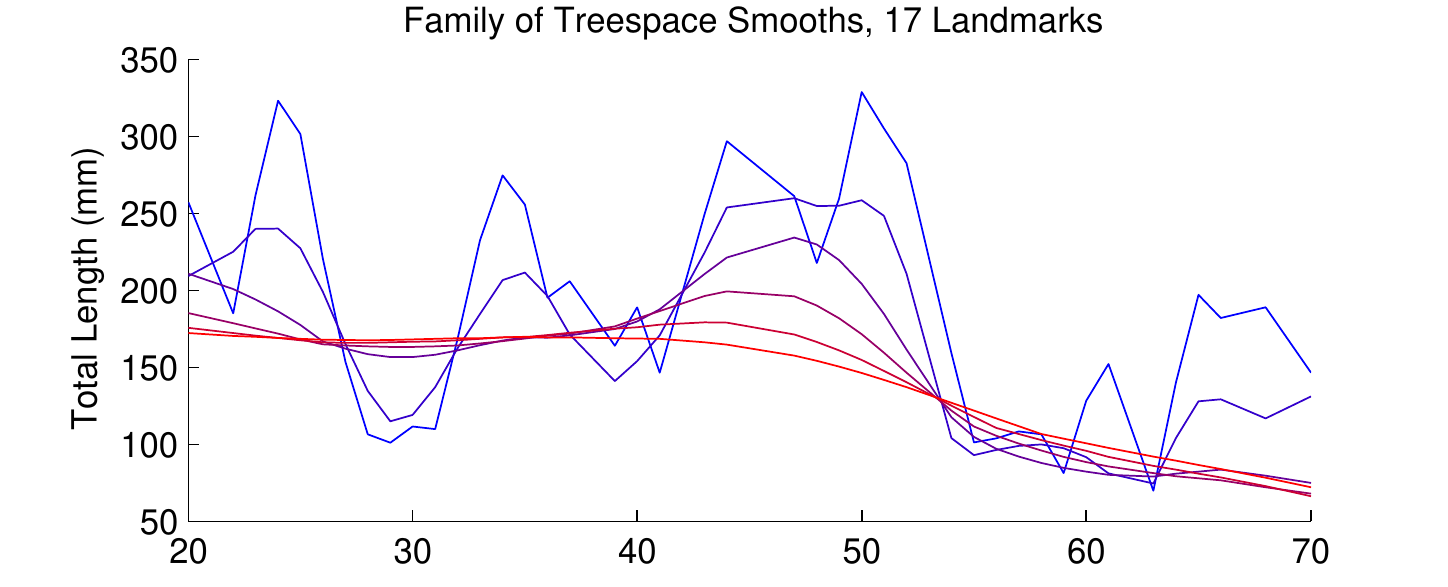}}
	\subfigure[]{\includegraphics[width=0.6\textwidth]{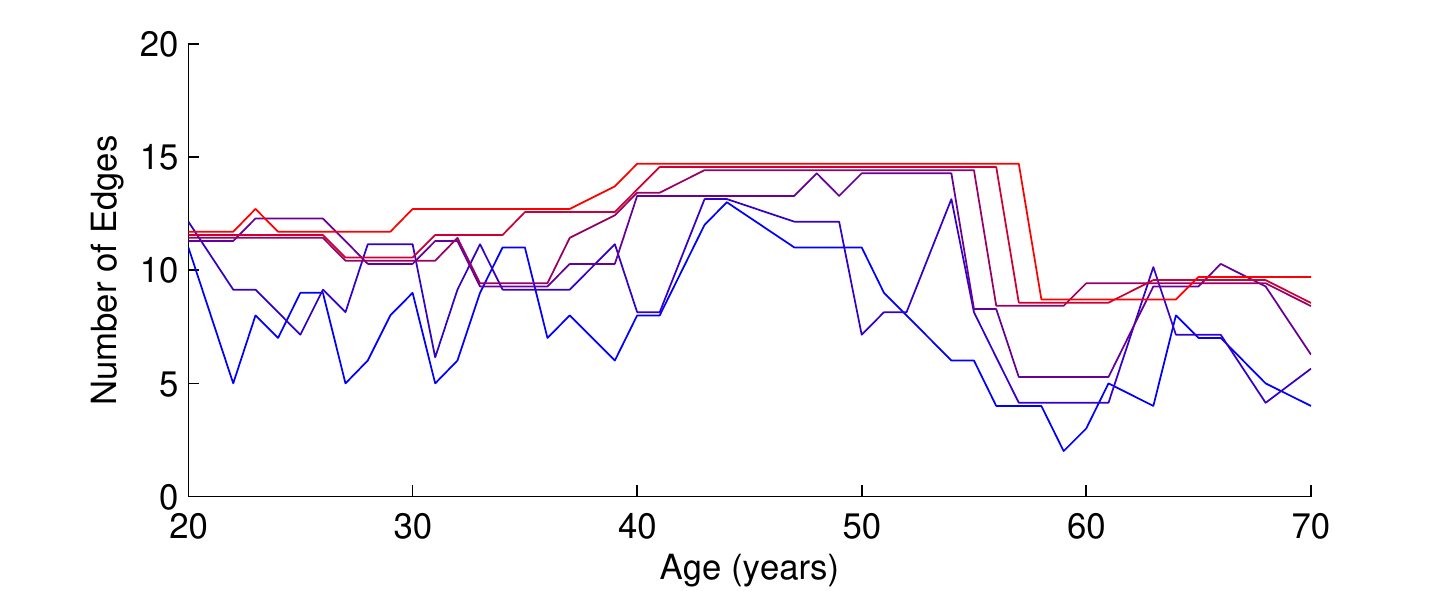}}
	\caption{Summaries from treesmooth family with bandwidths $h=1,2,3,4,5,6$ colored from blue to burgundy to red. Sum of lengths of interior edges over age (top); and number of interior edges over age. }\label{fig:TreeSmoothNEdgesLengths}
\end{figure}

Now we begin discussing the results.
We discuss the results in terms of two key summaries of treesmooths, the sum of interior edge lengths and the number of interior edges. Then we study the topological variability across ages as summarized by minimum length representative sequences.

Fig. \ref{fig:TreeSmoothNEdgesLengths} displays the sum of the interior edge lengths (top) and the number of interior edges. 
Notice that for the total length there is a clear overall downward trajectory which becomes increasingly evident as the bandwidth increases. This is consistent with a previous study, which showed a significant decrease in total length of the entire artery system with age \cite{Skwerer2014}.
Fluctuations in total length for smooths using narrower windows, $h=1,2,3$, colored blue to purple in Fig. \ref{fig:TreeSmoothNEdgesLengths}, are positively correlated
with fluctuations in the number of edges. This pattern is shown clearly
in the scatterplots of number of edges and total length in Fig. \ref{fig:TreeSmoothNEdges_TotalLength}.
For smoothing levels $h=4,5,6$, burgundy to red, a striking pattern in the scatterplots  emerges. 
The line of trees along the top of the plot coincides with the flat segment 
over the range of 40 to 55 years in Fig. \ref{fig:TreeSmoothNEdgesLengths}. We will offer a further interpretation
for this pattern in our discussion of topological variability across the smooth.

\begin{figure}[t]
	\hspace{-2cm}\includegraphics{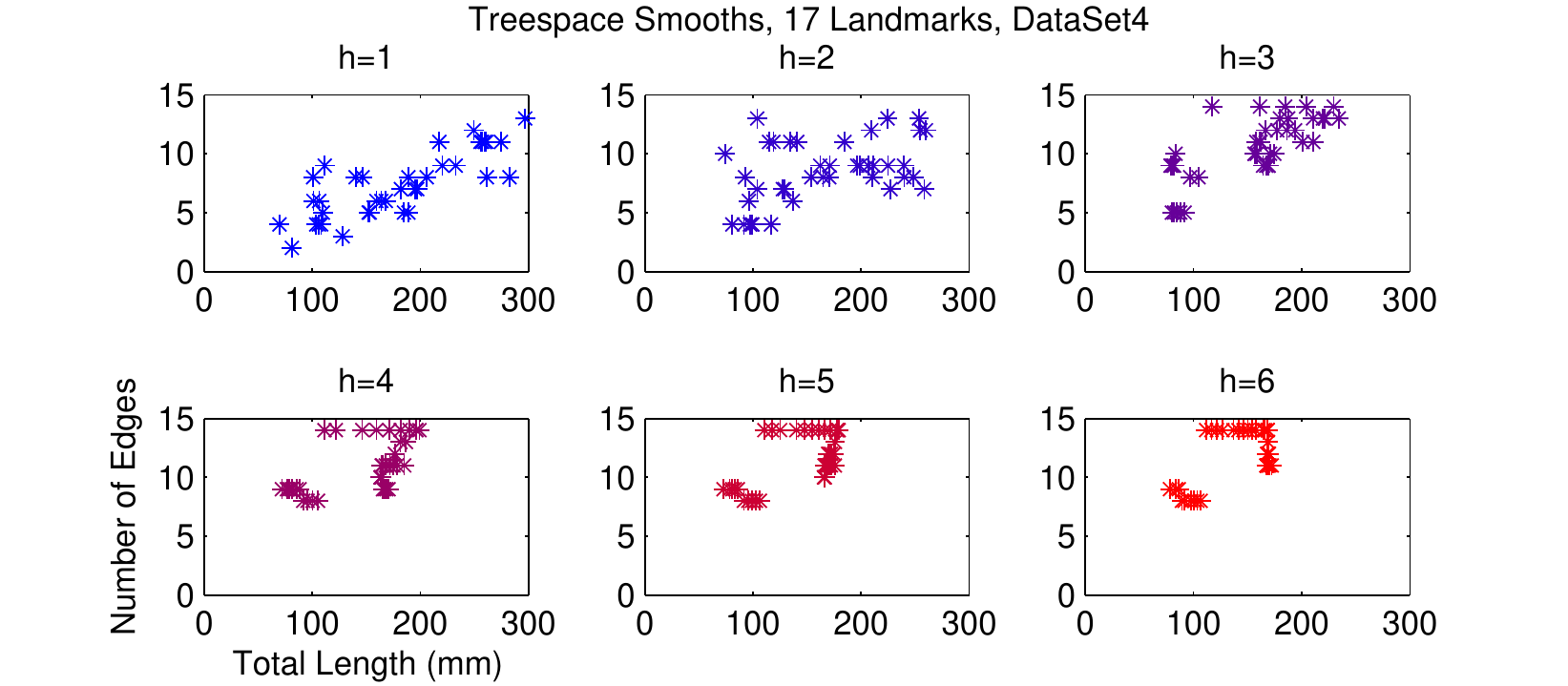}
	\caption{Scatter plots of number of interior edges and total length for smoothing windows $h=1,2,3,4,5,6$ years with the same colors as Fig. \ref{fig:TreeSmoothNEdgesLengths}. We see a positive correlation between these variables indicating that their up and down fluctuations across ages coincide. }\label{fig:TreeSmoothNEdges_TotalLength}
\end{figure}

For summarizing the topologies across the smooth we use a minimum length
representative sequence as described in Sec. \ref{sec:MinLengthRepSeq}.
Smoothing dramatically decreases the length of the minimum representative sequence. 
The sequence of data trees $T^1,\ldots,T^{85}$, ordered by ages, $x_1, \ldots, x_{85}$ are their
own minimal length representative sequence. The number of representative topologies by smoothing bandwidth are summarized in Table \ref{table: RepTop}.

\begin{table}[b]
	\centering
	\begin{tabular}{|c|c|}
		\hline
		$h$ yrs & \# Top. \\
		\hline
		0&85\\
		1&7\\
		2&2\\
		3&2\\
		4&1\\
		5&1\\
		6&1\\
		\hline
	\end{tabular}
	\caption[Minimal number of representative topologies by bandwidth]{Number of topologies in the minimal representative sequence for each smoothing window.}
	\label{table: RepTop}
\end{table}

The representative topology is the same for bandwidths $h=4,5,6$. In fact this representative topology is the topology of the overall Fr\'{e}chet mean of the entire dataset, shown in Fig. \ref{fig: RepTop456}. The patterns we see in Fig. \ref{fig:TreeSmoothNEdgesLengths} and Fig. \ref{fig:TreeSmoothNEdges_TotalLength} represent the smooth trees varying between various degenerate trees that are contractions of this representative topology. 

\begin{figure}[t]
	\includegraphics{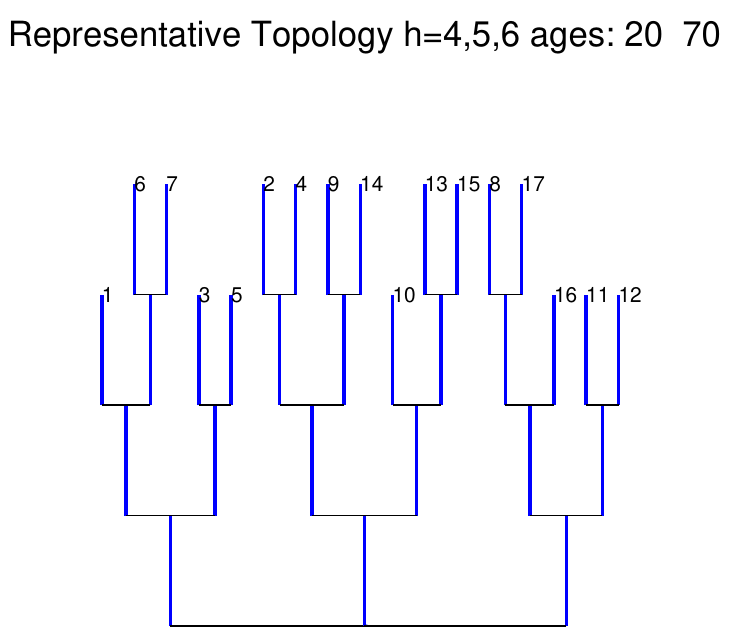}
	\caption{Representative topology for treespace smooths with bandwidths $h=4,5,6$. These bandwidths are large enough so the representative topology for the treesmooth is equal to the topology of the overall Fr\'{e}chet mean.}\label{fig: RepTop456}
\end{figure}

For smoothing bandwidths $h=1,2,3$ several representative topologies
are needed to capture different ranges of ages. 

For bandwidth $h=2$ the first representative topology captures ages 20 to 38, and the second captures ages 39 to 68.
For $h=3$ the first representative topology captures ages 20 to 39 and the second captures ages 40 to 68. The representative topologies are the same for bandwidths $h=2,3$; these are shown in Fig. \ref{fig:RepTop23}, where we can see that the representative topology for ages 20 to 38 only has one edge which is not shared by the overall Fr\'{e}chet mean, and the representative topology for ages 39 to 70 is the topology of the overall Fr\'{e}chet mean.

\begin{figure}[t]
	\includegraphics{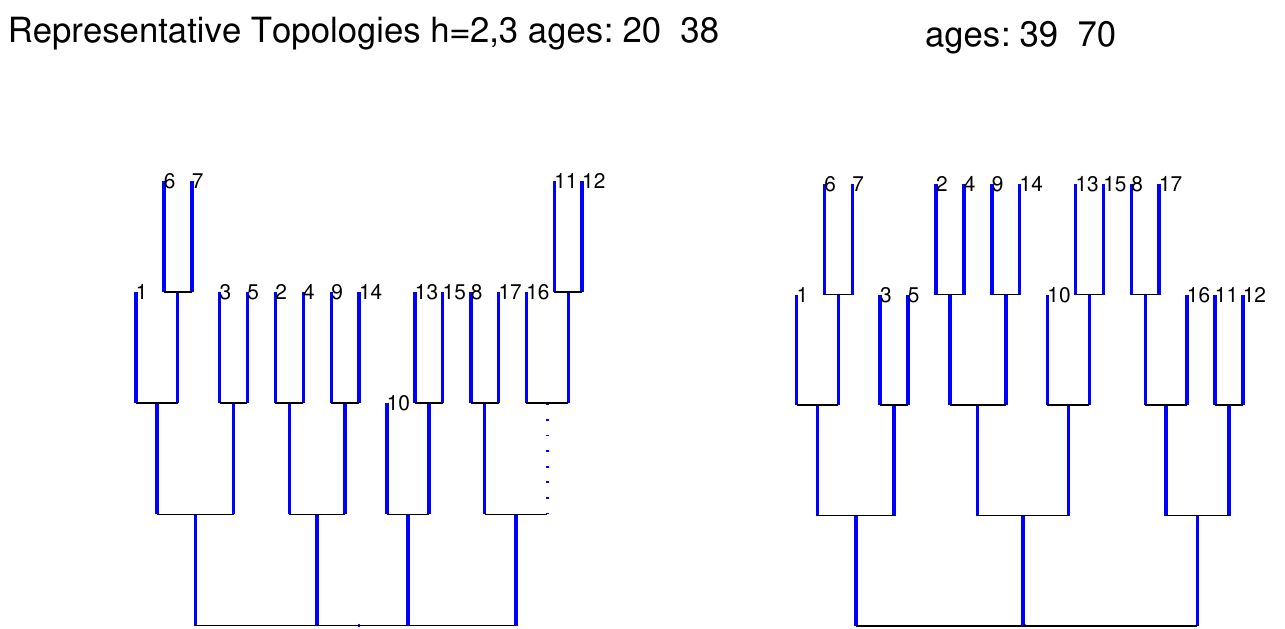}
	\caption{Representative topologies for treesmooths for both bandwidths $h=2,3$. The dotted edge in the left hand tree is the only edge which is not in the overall Fr\'{e}chet mean. }
	\label{fig:RepTop23}
\end{figure}

For $h=1$ the first representative topology captures the topologies in the treesmooth for ages 20 to 24, the second captures ages 25 to 33, the third captures 34 to 42, the fourth captures 43 to 50, the fifth captures 51 to 52, the sixth captures 54 to 62 and the seventh captures 63 to 68. 

\section{Minimum Length Representative Sequence Algorithm}\label{sec:MinLengthRepSeqAlg}
In this section we describe an algorithm to find a minimum length representative sequence for a sequence of  trees $T^1,\ldots,T^n$. 

Let $R(i,1),\ldots,R(i,k^i)$ be a minimum length representative
sequence for $T^1,\ldots,T^i$. When $i=1$ the minimum length representative sequence is $R(1,1)=T^1$.

Let $S(i,n)$ be a representative topology for $(T^{i+1},\ldots,T^{n}$) if  a representative topology exists.
Let $k^0=0$.
Let $i^*$ be an non-negative integer which minimizes $k^{i}$ among all $i$ such that $S(i,n)$ exists.
\begin{thm}
	A minimum length representative sequence for $T^1,\ldots,T^n$ is
	\begin{align}
		R(i^*,1),\ldots,R(i^*,k^{i*}),S(i^*,n)
	\end{align}
\end{thm}
\begin{proof}
	The proof will be by contradiction.
	Suppose there exists representative sequence $R^1,\ldots,R^l$ with length $l$ less than $k^{i*}+1$ for $T^1,\ldots,T^n$. Suppose that $R^l$ is representative for $T^{j+1},\ldots,T^n$ then $R^1,\ldots,R^{l-1}$ are a representative sequence for $T^1,\ldots,T^{j}$
	having length $l-1 < k^{i*}$. But this is a contradiction because $k^{i*}$ must be less than or equal to the length of any representative sequence for $T^1,\ldots,T^j$. 
\end{proof}


\clearpage
\phantomsection

{\def\chapter*#1{} 
	\begin{singlespace}
		\addcontentsline{toc}{chapter}{REFERENCES}
		\begin{center}
			\textbf{REFERENCES}
			\vspace{17pt}
		\end{center}
		
		\bibliographystyle{apalike}
		\bibliography{thesis,citations,library,trees}
	\end{singlespace}
}

\end{document}